\documentclass[accepted]{uai2025} 
                        
\usepackage[british]{babel}

\usepackage[sort]{natbib} 
\usepackage{amsmath, amsthm, amssymb}
\usepackage{mathtools} 
\usepackage{bm}
\usepackage[table,xcdraw]{xcolor}
\usepackage{booktabs} 
\usepackage{tikz} 

\usepackage{algorithm}
\usepackage{amsfonts}

\usepackage[noend]{algpseudocode}
\usepackage{multirow}
\usepackage{subcaption}
\usepackage{stfloats}

\usepackage{dsfont}

\newtheorem{theorem}{Theorem}[section]
\newtheorem{proposition}[theorem]{Proposition}
\newtheorem{corollary}[theorem]{Corollary}
\newtheorem{lemma}[theorem]{Lemma}
\newtheorem{defin}{Definition}
\newtheorem{remarkx}{Remark}

\newtheorem{assumption}{A \hspace{-5pt}}

\newtheorem{assumptionnew}{B \hspace{-2pt}}
\newtheorem{assumptiongla}{C \hspace{-2pt}}
\newtheorem{assumptiongla2}{D \hspace{-2pt}}
\DeclareSymbolFont{extraup}{U}{zavm}{m}{n}
\DeclareMathSymbol{\ballotx}{\mathalpha}{extraup}{129}

\DeclareMathOperator*{\argmin}{arg\,min}
\DeclareMathOperator*{\argmax}{arg\,max}
\DeclareMathOperator*{\kl}{KL}
\newcommand{\muu}{\pi_{\lambda}}  

\newcommand\numberthis{\addtocounter{equation}{1}\tag{\theequation}}

\DeclareMathOperator{\prox}{prox}

\newcommand{\md}{\mathrm{d}}
\def\real{\mathbb{R}}
\def\cP{{\mathcal P}}
\newcommand{\norm}[1]{\ensuremath{\Vert #1 \Vert}}

\newcommand{\bR}{\mathbb{R}}
\newcommand{\bE}{\mathbb{E}}

\newcommand{\bN}{\mathbb{N}}

\title{Proximal Interacting Particle Langevin Algorithms}

\author[1]{\href{mailto:<paula.cordero-encinar22@imperial.ac.uk>?Subject=Your UAI 2025 paper}{Paula Cordero Encinar}{}}
\author[2]{Francesca R. Crucinio}
\author[1]{O. Deniz Akyildiz}
\affil[1]{%
    Imperial College London, UK
}
\affil[2]{%
    ESOMAS, University of Turin, \& Collegio Carlo Alberto, Italy
}
  
\begin{document}
\maketitle

\begin{abstract}
We introduce a class of algorithms, termed proximal interacting particle Langevin algorithms (PIPLA), for inference and learning in latent variable models whose joint probability density is non-differentiable. Leveraging proximal Markov chain Monte Carlo techniques and interacting particle Langevin algorithms, we propose three algorithms tailored to the problem of estimating parameters in a non-differentiable statistical model. We prove nonasymptotic bounds for the parameter estimates produced by the different algorithms in the strongly log-concave setting and provide comprehensive numerical experiments on various models to demonstrate the effectiveness of the proposed methods. In particular, we demonstrate the utility of our family of algorithms for sparse Bayesian logistic regression, training of sparse Bayesian neural networks or neural networks with non-differentiable activation functions, image deblurring, and sparse matrix completion. Our theory and experiments together show that PIPLA family can be the de facto choice for parameter estimation problems in non-differentiable latent variable models.
\end{abstract}

\section{Introduction}\label{sec:intro}
Latent variable models (LVMs) are a class of probabilistic models which are widely used in machine learning and computational statistics for various applications, such as image, audio, and text modelling as well as in the analysis of biological data \citep{bishop2006pattern, murphy2012machine}. LVMs have demonstrated great success at capturing (often interpretable) latent structure in data, which is crucial in different scientific disciplines such as psychology and social sciences \citep{lvm_psychology, lvm_psy_2}, ecology \citep{lvm_ecology}, epidemiology \citep{lvm_epidemiology, muthen1992latent} and climate sciences \citep{christensen2012latent}. 

An LVM can be described as compactly as a parametrised joint probability distribution $p_\theta(x, y)\propto e^{-U(\theta, x)}$, where $\theta$ is a set of static parameters, $x$ denotes latent (unobserved, hidden, or missing) variables, and finally $y$ denotes (fixed) observed data. Given an LVM, there are two fundamental, intertwined statistical estimation tasks that must be solved simultaneously: (i) inference, which involves estimating the latent variables given the observed data and the model parameters through the computation of the posterior distribution $p_\theta(x|y)$, and (ii) learning, which involves estimating the model parameters $\theta$ given the observed data $y$ through the computation and maximisation of the marginal likelihood $p_\theta(y)$. 
The learning problem is often termed maximum marginal likelihood estimation (MMLE) and the main challenge is that $p_\theta(y)$ is often intractable.

The marginal likelihood $p_\theta(y)$ (also called the \textit{model evidence} \citep{bernardo2009bayesian}) in an LVM can be expressed as an integral, $p_\theta(y) = \int p_\theta(x, y) \mathrm{d} x$, over the latent variables. Hence the task of learning in an LVM can be formulated as solving the following optimisation problem
\begin{align}\label{eq:MMLE_problem}
\bar{\theta}_\star \in \arg \max_{\theta \in \Theta} p_\theta(y) = \arg \max_{\theta \in \Theta} \int p_\theta(x, y) \mathrm{d} x,
\end{align}
\normalsize
where $\Theta$ is the parameter space (which will be $\bR^{d_\theta}$ in our setting throughout). A classical algorithm for this setting is the celebrated expectation-maximisation (EM) algorithm \citep{dempster_1977}, which was first proposed in the context of missing data. The EM algorithm is an iterative procedure consisting of two main steps. Given a parameter estimate $\theta_k$, the expectation step (E-step) computes the expected value of the log likelihood function $\log p_\theta(x, y)$ with respect to the current conditional distribution for the latent variables given the observed data $p_{\theta_k}(x|y)$, i.e., $Q(\theta, \theta_k) = \bE_{p_{\theta_k}(x|y)}[\log p_\theta(x, y)]$. The second step is a maximisation step (M-step) which consists of maximising the expectation computed in the E-step. The EM algorithm, when it can be implemented exactly, builds a sequence of parameter estimates $(\theta_k)_{k\in \bN}$ where $\theta_k \in \arg \max_{\theta} Q(\theta, \theta_{k-1})$, which monotonically increase the marginal likelihood, i.e., $\log p_{\theta_k}(y) \geq \log p_{\theta_{k-1}}(y)$ \citep{dempster_1977}.

The wide use of the EM algorithm is due to the fact that it can be implemented using approximations for both steps \citep{wei1990monte, biometrika_1993, 1995_royal_statistical}, leveraging significant advances in Monte Carlo methods for the E-step and numerical optimisation techniques for the M-step. 
In particular, in most modern statistical models in machine learning, the posterior distribution, $p_\theta(x | y)$ for fixed $\theta$, is intractable, requiring an approximation for the E-step. One way to address this is by designing Markov chain Monte Carlo (MCMC) samplers. 
This approach has led to significant developments, where Markov kernels based on the unadjusted Langevin algorithm (ULA) \citep{tweedie_1996, alain_2017} have become a widespread choice in high dimensional settings thanks to their favourable theoretical properties \citep{alain_2019, NEURIPS2019_wibisono, RefWorks:RefID:73-de2021efficient, pmlr-v178-chewi22a, convergence_ula_2017}.

Recently, \citet{pmlr-v206-kuntz23a} explore an alternative approach for MMLE based on \citet{Neal1998}, where they exploit the fact that the EM algorithm is equivalent to performing coordinate descent of a free energy functional, whose minimum is the maximum likelihood estimate of the latent variable model. They propose several interacting particle algorithms to address the optimisation problem. 
This method has led to subsequent works \citep{akyildiz2023interacting, sharrock2024tuningfree, lim2024momentum, caprio2024error, johnston2024taming, gruffaz2024stochastic} including ours.

\paragraph{Contribution.} This work focuses on LVMs whose joint probability density is non-differentiable by leveraging proximal methods \citep{combettes2011proximal, proximal_algorithms_neal}. 
In the classical sampling case, this setting has been considered in a significant body of works, see,
for example, \citet{pereyra2016proximal,JMLR:v18:15-038, durmus_proximal,pmlr-v75-bernton18a, crucinio2023optimal, jko_forward_backward_2023_diao, pmlr-v178-chen22c,  NEURIPS2020_2779fda0, NEURIPS2019_6a8018b3_salim, pmlr-v134-lee21a}, due to significant applications in machine learning, most notably in the use of non-differentiable regularisers. For example, this type of model naturally arises when including sparsity-inducing penalties, such as Laplace priors for regression problems or Bayesian neural networks \citep{10.1162/neco.1995.7.1.117, pmlr-v97-yun19a}, and total variation priors in image processing \citep{durmus_proximal}. They are also relevant for non-differentiable activation functions in neural networks.
Specifically, we summarise our contributions below.
\begin{itemize}
\item We develop the first proximal interacting particle Langevin algorithm (PIPLA) family. Similar algorithms so far are investigated in the usual differentiable setting \citep{pmlr-v206-kuntz23a, akyildiz2023interacting, johnston2024taming}. We extend these methods to the non-differentiable setting via the use of proximal techniques. Specifically, we propose two main algorithms, termed Moreau-Yosida interacting particle Langevin algorithm (MYIPLA) and proximal interacting particle gradient Langevin algorithm (PIPGLA).
\item We present a theoretical analysis of our methods and, for comparison, of the proximal extensions we develop for the method in \citet{pmlr-v206-kuntz23a}, which we termed Moreau-Yosida particle gradient descent (MYPGD).

\item We apply our methods to a variety of examples and demonstrate that the PIPLA family is a viable option for the MMLE problem in non-differentiable settings.
\end{itemize}

This paper is organised as follows. Section \ref{sec:background} introduces the technical background necessary to develop our methods. In Sections \ref{sec:proximal_pipla} and \ref{sec:theoretical_analysis}, we present the proposed algorithms along with their theoretical analysis. Section \ref{sec:experiments} presents comprehensive numerical experiments before concluding. 

\paragraph{Notation.} We endow $\real^d$ with the Borel $\sigma$-field $\mathcal{B}(\real^d)$ with respect to
the Euclidean norm 
$\Vert\cdot\Vert$ when $d$ is clear from context. 
$\mathcal{N}(x| \mu, \Sigma)$ 
is the multivariate Gaussian, $I$ is the identity matrix, and $\mathcal{U}(x|a, b)$ is a uniform distribution. $\mathcal{C}^1$ denotes the space of continuously differentiable functions.  
We denote by $\mathcal{P}(\real^d)$ the set of probability measures over
$\mathcal{B}(\real^d)$ and endow this space with the topology of weak convergence. 
For all $p\geq 1$, we denote by
$\mathcal{P}_p(\real^d)$ the set of probability measures over
$\mathcal{B}(\real^d)$ with finite $p$-th moment. For any $\mu,\nu\in\mathcal{P}_2(\real^d)$, we denote by $W_2(\mu, \nu)$ the
$2$-Wasserstein distance between $\mu$ and $\nu$. 
$(\mathbf B_t)_{t\geq 0}$ is a $d$-dimensional Brownian motion, and $(\xi_n)_{n\in\mathbb{N}}$ is a sequence of i.i.d. $d$-dimensional standard Gaussian random variables.

In the following we adopt the convention that the solution to the continuous time solution is given by bold letters, whilst other processes (including the time discretisation) are not.
\section{Background}\label{sec:background}
We present the background and setting for our analysis.

\subsection{Langevin Dynamics}
At the core of our approach is the use of Langevin diffusion processes \citep{tweedie_1996}, which are widely used for building advanced sampling algorithms. The Langevin stochastic differential equation (SDE) is given by
\begin{equation}\label{eq:langevin_diffusion}
    \md \mathbf X_t = -\nabla U(\mathbf X_t)\md t + \sqrt{2} \md \mathbf B_t,
\end{equation}
\normalsize
where $U:\mathbb{R}^d\to\mathbb{R}$ is a continuously differentiable function. Under mild assumptions, \eqref{eq:langevin_diffusion} admits a strong solution, and its associated semigroup has a unique invariant distribution given by $\pi(x)\propto e^{-U(x)}$ \citep{pavliotis2014stochastic}.
In most cases, solving \eqref{eq:langevin_diffusion} analytically is not possible, however, we can resort to a discrete-time Euler-Maruyama approximation with step size $\gamma$ which gives the following Markov chain
\begin{equation}
\label{eq:ula}
    X_{n+1} = X_n -\gamma\nabla U(X_n) + \sqrt{2\gamma}\xi_{n+1}.
\end{equation}
\normalsize
This algorithm, known as the unadjusted Langevin algorithm (ULA) \citep{alain_2019}, exhibits favourable properties when $U$ is $\mu$-strongly convex and $L$-smooth (i.e. $\nabla U(x)$ is $L$-Lipschitz).
In particular, it converges exponentially fast to its biased limit $\pi^\gamma$ and the asymptotic bias is of order $\gamma^{1/2}$ \citep{alain_2019}. 

\subsection{MMLE with Langevin Dynamics}
A recent approach to solve the MMLE problem in \eqref{eq:MMLE_problem} is to build an extended stochastic dynamical system which can be run in the space $\bR^{d_\theta} \times \bR^{d_x}$, with the aim of jointly solving the problem of latent variable sampling and parameters optimisation. \citet{pmlr-v206-kuntz23a} first proposed a method termed \textit{particle gradient descent} (PGD) which builds on the observation made in \citet{Neal1998}, that the EM algorithm can be expressed as a minimisation problem of the free-energy objective (see Appendix~\ref{app:intro_pgd} for details). By constructing a gradient flow with respect to this functional, \citet{pmlr-v206-kuntz23a} propose the following system of SDEs
\small
\begin{align}
    \md\bm{\theta}_t^N &= -\frac{1}{N}\sum_{i=1}^N \nabla_{\theta} U(\bm{\theta}_t^N, \mathbf{X}_{t}^{i, N}) \md t,\label{eq:pdg_theta}\\
    \md \mathbf{X}_{t}^{i,N}&=-\nabla_{x} U(\bm{\theta}_t^N, \mathbf{X}_t^{i, N})\md t +\sqrt{2} \md \mathbf{B}_t^{i, N},  \quad i = 1,\dots,N,\nonumber
\end{align}
\normalsize
where $U(\theta, x) = -\log p_\theta(x, y)$ given  fixed observations $y$.
By introducing a noise term in the dynamics of $\theta$ \eqref{eq:pdg_theta}, which may facilitate escaping local minima in non-convex settings, \citet{akyildiz2023interacting} propose an interacting Langevin SDE:
\small
\begin{align}
\label{eq:ipla}
    \md \bm{\theta}_{t}^N&=-\frac{1}{N}\sum_{i=1}^N \nabla_{\theta} U(\bm{\theta}_t^N, \mathbf X_t^{i, N})\md t +\sqrt{\frac{2}{N}} \md \mathbf{B}_t^{0, N},\\
    \md \mathbf{X}_{t}^{i,N}&=-\nabla_{x} U(\bm{\theta}_t^N, \mathbf{X}_t^{i, N})\md t +\sqrt{2} \md \mathbf{B}_t^{i, N},  \quad i = 1,\dots,N.\notag \notag
\end{align}
\normalsize
An Euler-Maruyama discretisation of~\eqref{eq:ipla}, provides the interacting particle Langevin algorithm (IPLA).
Under strong-convexity and smoothness of $U$, IPLA and PGD exhibit favourable convergence properties \citep{akyildiz2023interacting, caprio2024error}.
\subsection{Proximal Methods}\label{subsec:proximal_alg}
Proximal methods \citep{combettes2011proximal,proximal_algorithms_neal} use proximity mappings of convex functions, instead of gradient mappings, to construct fixed point schemes and compute function minima. We now introduce some important definitions. Consider $U:\mathbb{R}^d\to \mathbb{R}$.
\begin{defin}[Proximity mappings]
\label{def:proxy}
The $\lambda$-proximity mapping or proximal operator function of $U$ is defined for any $\lambda>0$ as 
\begin{equation*}
    \prox_U^{\lambda}(x)\coloneqq \argmin_{z\in\mathbb{R}^{d}}\; \left\lbrace U(z) + \Vert z-x\Vert^2/{(2\lambda)}\right\rbrace.
\end{equation*}
\normalsize
\end{defin}
Intuitively, the proximity operator behaves similarly to a gradient map by moving points in the direction of the minimisers of $U$. In fact, when $U$ is differentiable, the proximal mapping corresponds to the implicit gradient step, as opposed to the explicit gradient step, which is known to be more stable \citep{proximal_algorithms_neal}. 
Note that as $\lambda\to 0$, the proximity operator converges to the identity operator, while as $\lambda\to \infty$, the proximity operator maps all points to the set of minimisers of $U$.

The idea of proximal methods is to approximate the non-differentiable target density $\pi \propto e^{-U}$, where $U$ is a convex lower semi-continuous function, by substituting the potential $U$ with a smooth approximation $U^\lambda$, where the level of smoothness is controlled by the parameter $\lambda>0$. The proximity map allows us to define a family of approximations to $\pi$, indexed by $\lambda$, and referred to as Moreau-Yosida approximation \citep{Moreau1965}. 
\begin{defin}
[$\lambda$-Moreau-Yosida approximation]
 For any $\lambda>0$, define the $\lambda$-Moreau-Yosida approximation of $U$ as 
   \begin{align*}
       U^{\lambda}(x)&\coloneqq  \min_{z\in\mathbb{R}^{d}} \left\lbrace U(z) + \Vert z-x\Vert^2/(2\lambda)\right\rbrace\\
       &= U(\prox_U^{\lambda}(x))+\Vert \prox_U^{\lambda}(x)-x\Vert^2/(2\lambda).
   \end{align*}
   \normalsize
   Consequently, we define the $\lambda$-Moreau-Yosida approximation of $\pi$ as the following density $\pi_{\lambda}(x)\propto e^{-U^{\lambda}(x)}$.
   \label{def:my}
\end{defin}

The approximation $\pi_\lambda$ converges to $\pi$ as $\lambda\to 0$ \citep{Rockafellar_B._1998, durmus_proximal} and is differentiable even if $\pi$ is not, with log-gradient
\begin{align}
\label{eq:proximity}
    \nabla\log\pi_{\lambda}(x) = -\nabla U^{\lambda}(x) = (\prox_U^{\lambda}(x)-x)/\lambda.
\end{align}
\normalsize
Since $\pi_\lambda$ is continuously differentiable, proximal MCMC methods \citep{pereyra2016proximal, durmus_proximal} rely on discretisations of the Langevin diffusion associated with $\pi_\lambda$, given by replacing the drift term in Eq. \eqref{eq:langevin_diffusion} with $\nabla U^\lambda$,
to approximately sample from $\pi$.
We consider two classes of proximal Langevin algorithms, based on different discretisation schemes: Euler-Maruyama discretisations \cite{pereyra2016proximal,durmus_proximal}, and splitting schemes \cite{klatzer2023accelerated, habring2023subgradient, ehrhardt2023proximal, durmus2019analysis, NEURIPS2019_6a8018b3_salim}.

\subsubsection{Proximal Langevin methods}
When $U(x) = -\log \pi(x)$ can be expressed as $U(x) = g_1(x) + g_2(x)$, with $g_1, g_2$ convex lower bounded functions, $g_1$ differentiable and $g_2$ proper and lower semi-continuous, \citet{durmus_proximal} consider  $\pi_\lambda\propto e^{-U^\lambda}$ with $U^\lambda(x)= g_1(x)+g_2^\lambda(x)$, and the corresponding Langevin diffusion
\small
\begin{equation*}
    \md \mathbf{X}_{\lambda, t} = -(\nabla g_1(\mathbf{X}_{\lambda, t}) + \nabla g_2^\lambda(\mathbf{X}_{\lambda, t}))\md t +\sqrt{2} \md \mathbf B_t, \ t\geq 0.
\end{equation*}
\normalsize 
An Euler-Maruyama discretisation with step size $\gamma>0$ results in the Moreau-Yosida ULA (MYULA) algorithm.
\subsubsection{Proximal gradient MCMC methods}
\label{sec:pgla} 
Inspired by the proximal gradient algorithm (see, e.g., \citet[Section 4.2]{proximal_algorithms_neal} or \citet{combettes2011proximal}), which is a forward-backward splitting optimisation algorithm, \citet{NEURIPS2019_6a8018b3_salim} propose a sampling algorithm for the case $U^\lambda(x) = g_1(x) + g_2^\lambda(x)$ termed proximal gradient Langevin algorithm (PGLA). The method consists of a forward step in the direction of $\nabla g_1$ with an additional stochastic term, followed by a backward step using the proximity map of $g_2$:
\begin{align*}
X_{n+1/2} &= X_n -\gamma\nabla g_1(X_n)+\sqrt{2\gamma}\xi_{n+1}\\
X_{n+1} &= \prox^{\lambda}_{g_2}\left(X_{n+1/2}\right),
\end{align*}
\normalsize 
These algorithms were originally proposed as an alternative to MYULA to deal with cases in which $U$ is the sum of a differentiable likelihood $g_1$ and a compactly supported $g_2$, since the application of the proximity map after the addition of the stochastic term guarantees that $X_{n+1}$ remains in the support of $g_2$ \citep{NEURIPS2020_2779fda0}.
In addition, \citet{ehrhardt2023proximal} give conditions under which using an approximate proximity map does not affect numerical results and generalise existing convergence bounds.

\section{Proximal Interacting Particle Methods for MMLE}\label{sec:proximal_pipla}
Our goal is to extend interacting particle algorithms for the MMLE problem \eqref{eq:MMLE_problem} to cases where the distribution $p_\theta(x, y) = \pi(\theta, x) \propto e^{-U(\theta, x)}$ may be non-differentiable. 
To achieve this, we build on the previously presented methodology, approximating the target distribution $\pi\propto e^{-U} = e^{- g_1 -g_2}$ with a Moreau-Yosida envelope $\pi_{\lambda}$ and deriving a numerical scheme using either an Euler-Maruyama discretisation or a splitting scheme. 
In particular, we introduce three classes of proximal algorithms.

Recall that $U(\theta, x) = - \log p_\theta(x, y)$ where $y$ is fixed. 
The proximal map in the MMLE setting is given below.
\begin{remarkx}
\label{rem:proxy}
In our scenario, the $\argmin$ in the proximal map is taken over both variables $(\theta, x)$, that is, 
\begin{align*}
    &\prox_U^{\lambda}(\theta, x) = (\prox_U^{\lambda}(\theta, x)_{\theta}, \prox_U^{\lambda}(\theta, x)_x)\\
    &=\argmin_{z_0\in\mathbb{R}^{d_{\theta}}, z\in\mathbb{R}^{d_x}}\; \left\lbrace U(z_0, z) + \Vert (z_0, z)-(\theta, x)\Vert^2/(2\lambda)\right\rbrace.
\end{align*}
\normalsize
\end{remarkx}

\subsection{Proximal Interacting Particle Algorithms}
Our algorithms are based on different discretisation schemes of the following continuous-time interacting SDEs:
\small
\begin{align}
        \md \bm{\theta}_{t}^N&=-\frac{1}{N}\sum_{i=1}^N \nabla_{\theta} U^{\lambda}(\bm{\theta}_{t}^N, \mathbf{X}_{t}^{i, N})\md t +\sqrt{\frac{2}{N}} \md \mathbf{B}_t^{0, N},\label{eq:particles_proximal_1}\\
        \label{eq:particles_proximal_2}
    \md \mathbf{X}_{t}^{i,N}&=-\nabla_{x} U^{\lambda}(\bm{\theta}_{t}^N, \mathbf{X}_{t}^{i, N})\md t +\sqrt{2} \md \mathbf{B}_t^{i, N},
\end{align}
\normalsize
where $U^\lambda = g_1 + g_2^{\lambda}$ is the Moreau-Yosida approximation of $U$. 
As in the case of the interacting SDE in Eq.~\eqref{eq:ipla}, we show that~\eqref{eq:particles_proximal_1}--\eqref{eq:particles_proximal_2} converges to an SDE of the McKean--Vlasov type (MKVSDE) as $N\to \infty$. In particular, if the potential $U$ is regular enough, the MKVSDE that~\eqref{eq:particles_proximal_1}--\eqref{eq:particles_proximal_2} approximates becomes arbitrarily close to that approximated by~\eqref{eq:ipla} as $\lambda\to 0$ (see Appendix~\ref{app:gf} for a proof).

\subsubsection{Moreau-Yosida Interacting Particle Langevin Algorithm (MYIPLA)}\label{sec:myula}
If we consider $U^{\lambda} = g_1 + g_2^{\lambda}$ as in \citet{durmus_proximal}, and substitute its gradient in the Euler-Maruyama discretisation of \eqref{eq:particles_proximal_1}--\eqref{eq:particles_proximal_2} we derive  \textit{MYIPLA} (Moreau-Yosida Interacting Particle Langevin algorithm):
\small
\begin{align}
    \theta_{n+1}^N =& \Big(1-\frac{\gamma}{\lambda}\Big)\theta_{n}^N  + \frac{\gamma}{N}\sum_{i=1}^N \Big(-\nabla_{\theta} g_1(\theta_n^N, X_n^{i, N})\nonumber \\
    &+ \frac{1}{\lambda} \prox_{g_2}^{\lambda}(\theta_n^N, X_n^{i, N})_{\theta}\Big)+\sqrt{\frac{2\gamma}{N}}\xi_{n+1}^{0, N},\label{eq:pip-myula_theta}\\
    X_{n+1}^{i, N} =& \Big(1-\frac{\gamma}{\lambda}\Big)X_{n}^{i, N} -\gamma \nabla_{x} g_1(\theta_n^N, X_n^{i, N}) \nonumber\\
    &+ \frac{\gamma}{\lambda}\prox_{g_2}^{\lambda}(\theta_n^N, X_n^{i, N})_x+\sqrt{2\gamma}\;\xi_{n+1}^{i, N},\label{eq:pip-myula_x}
\end{align}
\normalsize
where the notation $\prox_{g_2}^{\lambda}(\theta, X)_{\theta} $, $\prox_{g_2}^{\lambda}(\theta, X)_{x}$ refers to the $\theta$ and $x$ component of the proximal mapping $\prox_{g_2}^{\lambda}$, as mentioned in Remark~\ref{rem:proxy}. 
Setting $\gamma = \lambda$ and $g_1 = 0$ as in \citet{pereyra2016proximal}, we obtain a specific case of the previous algorithm, that we refer to as \textit{PIPULA} (proximal interacting particle ULA).

\subsubsection{Proximal Interacting Particle Gradient Langevin Algorithm (PIPGLA)}
Inspired by the proximal gradient method  \citep{NEURIPS2019_6a8018b3_salim, ehrhardt2023proximal}, we employ a splitting scheme to discretise~\eqref{eq:particles_proximal_1}--\eqref{eq:particles_proximal_2} and obtain \textit{PIPGLA} (Proximal Interacting Particle Gradient Langevin algorithm).
In this case, we perform one ULA step for both the $\theta$ and $x$ component using $\nabla g_1$ followed by a backward step using $\prox_{g_2}^\lambda$:
\small
\begin{align}
    \theta_{n+1/2}^N &= \theta_{n}^N  - \frac{\gamma}{N}\sum_{i=1}^N \nabla_{\theta} g_1(\theta_n^N, X_n^{i, N}) +\sqrt{\frac{2\gamma}{N}}\xi_{n+1}^{0, N}, \label{eq:pip-gla} \\
    X_{n+1/2}^{i, N} &= X_{n}^{i, N} -\gamma \nabla_{x} g_1(\theta_n^N, X_n^{i, N})+\sqrt{2\gamma}\;\xi_{n+1}^{i, N}, \label{eq:pip-gla-x-half}\\
    \theta_{n+1}^N &= \frac{1}{N}\sum_{i=1}^N\prox_{g_2}^{\lambda} \bigg(\theta_{n+1/2}^N, X_{n+1/2}^{i, N}\bigg)_{\theta},\label{eq:pip-gla-theta} \\
    X_{n+1}^{i, N} &= \prox_{g_2}^{\lambda} \big(\theta_{n+1/2}^N, X_{n+1/2}^{i, N}\big)_{x},\label{eq:pip-gla-x}
\end{align}
\normalsize
Setting $\lambda=\gamma$, as is common in proximal gradient algorithms (see, e.g., \citet{NEURIPS2019_6a8018b3_salim}), we obtain an algorithm similar to MYIPLA except for the fact that $\nabla g_2^\lambda$ is evaluated at $(\theta_{n+1/2}, X^{i,N}_{n+1/2})$ instead of $(\theta_n, X^{i,N}_n)$.
 
Similarly to PGLA, this algorithm ensures that $X_{n+1}^{i, N}$ belongs to the support of $g_2$ for all $i$. 
If the parameter space $\Theta$ is convex, then also $\theta_{n+1}^N$ belongs to the support of $g_2$ since $\theta^N_{n+1}$ is a convex combination of elements of $\Theta$.

\subsection{Proximal Particle Gradient Descent Methods (PPGD)}
All the methods above incorporate a noise term in the $\theta$-dimension, making the system more akin to a Langevin-type system, as in \citet{akyildiz2023interacting}.
However, we also explore the case where the noise term is removed from the dynamics of $\theta$, similar to \citet{pmlr-v206-kuntz23a}. 
By removing the noise term from the dynamics of $\theta$ in \eqref{eq:particles_proximal_1}, we obtain the following system of SDEs
\begin{align*}
        \md \bm{\theta}_{t}^N&=-\frac{1}{N}\sum_{i=1}^N \nabla_{\theta} U^{\lambda}(\bm{\theta}_{t}^N, \mathbf{X}_{t}^{i, N})\md t,\\
    \md \mathbf{X}_{t}^{i,N}&=-\nabla_{x} U^{\lambda}(\bm{\theta}_{t}^N, \mathbf{X}_{t}^{i, N})\md t +\sqrt{2} \md \mathbf{B}_t^{i, N}.
\end{align*}
\normalsize
Discretising this SDE system, we obtain similar algorithms to MYIPLA and PIPULA without the term $\sqrt{2\gamma/N}\xi_{n+1}^{0, N}$ in \eqref{eq:pip-myula_theta}, which can be seen as proximal versions of PGD. 
Accordingly, we term these methods as proximal PGD (PPGD) and Moreau-Yosida PGD (MYPGD), respectively. We provide a detailed description of these methods and their theoretical analysis in Appendix~\ref{app:proximal_particle_gradient_descent_methods}.

\section{Nonasymptotic analysis}\label{sec:theoretical_analysis}
This section presents a theoretical analysis of the parameter estimates obtained by the PIPLA family. 
It is important to note that similar assumptions across algorithms enable a fair comparison of convergence rates.

\subsection{Assumptions}
Let $g_1,g_2:\mathbb{R}^{d_\theta}\times\mathbb{R}^{d_x}\to \mathbb{R}$ and $U = g_1 + g_2$. 
\begin{assumption}\label{assumption_1}
    $g_1$ is continuously differentiable, convex, $L_{g_1}$-smooth and lower bounded, and $g_2$ is proper, convex, lower semi-continuous and lower bounded. 
\end{assumption}
\textbf{A\ref{assumption_1}} guarantees that $\prox_{g_2}^\lambda$ is well defined and that $\nabla U^{\lambda}$ is Lipschitz in both variables with constant $L\leq L_{g_1} + \lambda^{-1}$ \cite[Proposition 1]{durmus_proximal}. 
\begin{assumption}\label{assumption_2}
   The initial condition $Z_0^N = (\theta_0,  N^{-1/2}X_0^{1,N},\dots,$ $ N^{-1/2}X_0^{N,N})$ satisfies $\mathbb{E}[\Vert Z_0^N\Vert^2]\leq H$ for $H>0$. 
\end{assumption}
\begin{assumption}\label{assumption_1prime}
    $g_2$ is Lipschitz with constant $\Vert g_2\Vert_{\textnormal{Lip}}$. 
\end{assumption}
\begin{assumption}\label{assumption_3}
    $g_1$ is $\mu$-strongly convex. 
\end{assumption}
\begin{remarkx}
Let $v=(\theta, x)$ and $v'=(\theta', x')$. We have that $\nabla g_2^{\lambda}(v) = (v-\prox_{g_2}
^{\lambda}(v))/\lambda$ and $\prox_{g_2}
^{\lambda}$ is firmly non expansive \cite[Eq. (7)]{durmus_proximal} which implies Lipschitzness. 
By the Cauchy-Schwarz inequality, 
\small
\begin{align*}
     &\langle v - v', \nabla g_2^{\lambda}(v)-\nabla g_2^{\lambda}(v')\rangle  \\
&\geq \frac{1}{\lambda}(\Vert v-v'\Vert^2 - \Vert v-v'\Vert \Vert \prox_{g_2}
^{\lambda}(v)-\prox_{g_2}
^{\lambda}(v')\Vert)\geq 0.
\end{align*}
\normalsize
Therefore, under \textbf{A\ref{assumption_3}}, $\nabla U^{\lambda}$ is also $\mu$-strongly convex.
\end{remarkx}
\begin{remarkx}\label{remark_differentiability} Let $\Omega\subset\mathbb{R}^{d_{\theta}}\times \mathbb{R}^{d_x}$ denote the (nonempty) set where $g_2$ is twice differentiable. Theorem 25.5 in \citep{Rockafellar_1970}, guarantees that if $g_2$ is a proper, convex function, then $g_2$ is differentiable except in a set of measure zero, i.e., $\Omega^{\mathsf{c}}$ has measure zero. Also, by Alexandrov's Theorem \citep{second_order_convex}, $g_2$ is twice differentiable almost everywhere---in particular, these points form a subset of $\textnormal{\textsf{dom}}(\nabla g_2)$. In addition, the Hessian $\nabla^2 g_2$ (or alternatively its distributional counterpart, $D^2 g_2$, if $\nabla^2 g_2$ does not exist) is symmetric and positive definite \citep[Proposition 7.11]{RefWorks:RefID:77-alberti1999geometrical}. 
\end{remarkx}

Let $\Bar{\theta}_{\star}$ be the maximiser of $p_\theta(y)$. Let $\mathsf{m}_{\Bar{\theta}_\star}$ be the restriction of the Lebesgue measure $\mathsf{m}$ on $\bR^{d_\theta} \times \bR^{d_x}$ to the set $\{\bar{\theta}_\star\} \times \bR^{d_x}$, which is well defined (see, e.g., \cite[Section 10.6]{Bogachev2007}).
\begin{assumption}\label{assumption_4}
We assume that $\mathsf{m}_{\Bar{\theta}_\star}(\Omega^{\mathsf{c}} \cap (\{\bar{\theta}_\star\} \times \bR^{d_x})) = 0$. Moreover,
\begin{align*}
    &\mathbb{E}_X[\Vert\nabla_{\theta} U(\Bar{\theta}_{\star}, X)\Vert] \leq A,\\
    &\mathbb{E}_X [\Vert\nabla^2_{(\theta, x)} g_2(\Bar{\theta}_\star, X)\nabla_{(\theta, x)} g_2(\Bar{\theta}_\star, X)\Vert]\leq B
\end{align*}
\normalsize
where $X\sim \rho_{\Bar{\theta}_\star}(x)$ with $\rho_{\Bar{\theta}_\star}(x) \propto e^{-U^{\lambda}(\Bar{\theta}_{\star},x)}$. 
\end{assumption} 

\subsection{The proof strategy}
Our objective is to find the MMLE $\bar{\theta}_\star = \argmax_{\theta} p_\theta(y)$,
where $p_\theta(y) = \int e^{-U(\theta, x)} \md x$. Therefore, we provide an upper bound on the distance between the iterates of our algorithms and $\bar{\theta}_\star$, that is, $\mathbb{E}[\Vert\Bar{\theta}_{\star} -\theta_n\Vert^2]^{1/2}$.

Let $(\bm{\theta}^N_t)_{t\geq 0}$ be the $\theta$-marginal of the solution to the SDE \eqref{eq:particles_proximal_1}--\eqref{eq:particles_proximal_2} and $(\theta_n^N)_{n\in\mathbb{N}}$ be the $\theta$ iterates of any algorithm which is a discretisation of \eqref{eq:particles_proximal_1}--\eqref{eq:particles_proximal_2}. Denote the $\theta$-marginal of the target measure of \eqref{eq:particles_proximal_1}--\eqref{eq:particles_proximal_2} by $\pi_{\lambda, \Theta}^N$,
\begin{align}\label{eq:marginal_distribution_system_sdes}
    \pi^N_{\lambda,\Theta} (\theta) &\propto 
    \biggr(\int_{\real^{d_x}} e^{-U^\lambda(\theta, x)} \md x \biggr)^N.
\end{align}
\normalsize
Note that $\mathbb{E}[\Vert  \Bar{\theta}_{\star}-\theta_n^N\Vert^2]^{1/2} = W_2(\delta_{\Bar{\theta}_\star}, \mathcal{L}(\theta_n^N))$. Applying the triangle inequality of the $W_2$ metric, it follows:
\begin{align}\label{eq:splitting}
    W_2(\delta_{\Bar{\theta}_\star}, \mathcal{L}(\theta_n^N))
    \leq &\underbrace{W_2(\delta_{\Bar{\theta}_{\star}}, \pi_{\lambda, \Theta}^N)}_{\textnormal{concentration}} + \underbrace{W_2(\pi_{\lambda, \Theta}^N, \mathcal{L}(\boldsymbol{\theta}_{n\gamma}^N))}_{\textnormal{convergence}}\nonumber\\
    &+\underbrace{W_2(\mathcal{L}(\theta_{n}^N), \mathcal{L}(\boldsymbol{\theta}_{n\gamma}^N))}_{\textnormal{discretisation}}.
\end{align}
\normalsize
The \textit{concentration} term characterises the concentration of the $\theta$-marginal of the target measure of the SDE \eqref{eq:particles_proximal_1}--\eqref{eq:particles_proximal_2}, given by \eqref{eq:marginal_distribution_system_sdes}, around the maximiser of $p_\theta(y)$, denoted by $\bar{\theta}_\star$. Handling this term is not trivial since the maximisers of $p_\theta(y)$ and of $p_\theta^\lambda(y):= \int_{\real^{d_x}}p_\theta^\lambda(x, y)\md x$ are not necessarily the same as we clarify in the next section. The \textit{convergence} term captures the convergence of the solution of the SDE to its target measure \eqref{eq:marginal_distribution_system_sdes}.  Finally, the \textit{discretisation} term characterises the error introduced by discretising the SDE. We provide nonasymptotic results for the convergence of the PIPLA family, proofs of the results are provided in Appendices \ref{ap:theory_IPLA} and \ref{app:proximal_particle_gradient_descent_methods}.

\subsection{Nonasymptotic analysis of MYIPLA}
\begin{theorem}[MYIPLA]\label{th:myipla_main_text}
Let \textbf{A\ref{assumption_1}}--\textbf{A\ref{assumption_4}} hold. Let $\theta_n^N$ denote the iterate \eqref{eq:pip-myula_theta} and $\Bar{\theta}_{\star}$ be the maximiser of $p_\theta(y)$. Fix $\gamma_0\in(0, \min\{(L_{g_1}+\lambda^{-1})^{-1}, 2\mu^{-1}\})$. Then for every $\lambda > 0$ and $\gamma\in(0,\gamma_0]$, one has
\begin{align*}
\mathbb{E}&[\Vert\theta_n^N-\Bar{\theta}_{\star}\Vert^2]^{1/2}\leq \frac{\lambda}{\mu} \Big(\frac{\Vert g_2\Vert_{\textnormal{Lip}}^2}{2} A  +  B\Big) + \sqrt{\frac{d_{\theta}}{N\mu}} \\
&+ e^{-\mu n\gamma}\bigg(\mathbb{E}[\Vert Z_0^N-z_{\star}\Vert^2]^{1/2}+\Big(\frac{d_xN+d_{\theta}}{N\mu}\Big)^{1/2}\bigg) \\&+ C_1(1+\sqrt{d_{\theta}/N+d_x})\gamma^{1/2}+\mathcal{O}(\lambda^2),
\end{align*}
\normalsize
for all $n\in\mathbb{N}$, where $z_\star = (\theta_\star, N^{-1/2}x_\star,\dots,N^{-1/2}x_\star)$ and $(\theta_\star, x_\star)$ is the minimiser of $U^{\lambda}$ and $C_1>0$ is a constant independent of $t,n,N,\gamma,\lambda,d_{\theta}, d_x$.
\end{theorem}
See Appendix \ref{ap:MYIPLA} for the full proof.
Below, we provide a sketch of the proof, following the error decomposition in \eqref{eq:splitting}. 
The \textit{concentration} term can be split into two parts: (1) the distance between the MMLE of the original distribution, $\bar{\theta}_{\star}$, and the MMLE of the MY approximation, $\bar{\theta}_{\star, \lambda}$, and (2) the concentration of the MMLE of the MY approximation around the target regularised marginal measure $\pi_{\lambda, \Theta}^N$ provided in \eqref{eq:marginal_distribution_system_sdes}. This results in
\begin{equation*}
    W_2(\delta_{\Bar{\theta}_\star}, \pi_{\lambda, \Theta}^N)\leq \| \Bar{\theta}_\star - \Bar{\theta}_{\star, \lambda}\| + W_2(\delta_{\Bar{\theta}_{\star, \lambda}}, \pi_{\lambda, \Theta}^N),
\end{equation*}
\normalsize
where $\Bar{\theta}_{\lambda,\star}$ denotes the maximiser of $p_\theta^\lambda(y)$.
We derive in Proposition \ref{prop:theta_convergence_lambda} a novel bound for the distance between the maximisers $\Bar{\theta}_\star, \Bar{\theta}_{\star, \lambda}$, given by
\begin{equation*}
    \Vert\Bar\theta_{\star}-\Bar{\theta}_{\star, \lambda}\Vert \leq \frac{\lambda}{\mu} \left(\frac{\Vert g_2\Vert_{\textnormal{Lip}}^2}{2} A  +  B\right) + \mathcal{O}(\lambda^2),
\end{equation*}
\normalsize
with $A, B$ given in \textbf{A\ref{assumption_4}}. 
We observe that stronger regularisation (i.e., larger values of $\lambda$) increases the distance between the two MMLEs—intuitively, in the limit $\lambda \to \infty$, all points collapse to the minimiser of the potential $U$, therefore this leads to a larger difference between $\bar{\theta}_{\star}$ and $\bar{\theta}_{\star, \lambda}$.

Next, consider the stationary measure of the SDE \eqref{eq:particles_proximal_1}--\eqref{eq:particles_proximal_2}, denoted by $\pi_{\lambda, \star}^N(\theta, x_1, \ldots, x_N)$, and its $\theta$-marginal given by $\pi_{\lambda, \Theta}^N$. Using a form of the Prékopa-Leindler inequality for strong convexity \citep[Theorem 3.8]{saumard2014log}, $\pi_{\lambda, \Theta}^N$ is $N\mu$-strongly log-concave and by \citep[Lemma A.8]{sinho_lemma_A8}
\begin{align}\label{eq:concentration_term_aux}
W_2(\delta_{\Bar{\theta}_{\star, \lambda}}, \pi_{\lambda, \Theta}^N)\leq \sqrt{\frac{d_{\theta}}{N\mu}},
\end{align}
\normalsize
which concludes the bound for the \textit{concentration} term.
Besides, the \textit{convergence} term is characterised by the exponential decay of the Wasserstein-2 distance between $\pi_{\lambda, \Theta}^N$ and the $\theta$-marginal of the solution of the SDE $\mathcal{L}(\boldsymbol{\theta}_{n\gamma}^N)$, that is,
\smaller
\begin{align*}
    W_2(\pi_{\lambda, \Theta}^N, \mathcal{L}(\boldsymbol{\theta}_{n\gamma}^N))\leq e^{-\mu n\gamma}\bigg(\mathbb{E}[\Vert Z_0^N-z_{\star}\Vert^2]^{1/2}+\sqrt{\frac{d_xN+d_{\theta}}{N\mu}}\bigg).
\end{align*}
\normalsize
Finally, the \textit{discretisation} term is bounded by
\begin{align*}
    W_2(\mathcal{L}(\theta_{n}^N), \mathcal{L}(\boldsymbol{\theta}_{n\gamma}^N))\leq C_1(1+\sqrt{d_{\theta}/N+d_x})\gamma^{1/2}.
\end{align*}
\normalsize
This bound is derived using a strategy similar to that used in classical Langevin methods.
\subsection{Nonasymptotic analysis of PIPGLA}
We introduce an extra assumption regarding the least norm element $\nabla^0g_2$ in the subdifferential of $g_2$, defined in \ref{app:pip_gla analysis}.
\begin{assumptiongla}\label{assumption_1_gla}
    We assume that $\Vert\nabla^0 g_2(\theta, x)\Vert^2\leq C$ for all $\theta\in\mathbb{R}^{d_\theta}$ and $x\in\mathbb{R}^{d_x}$.
\end{assumptiongla} 
For the convergence analysis of PIPGLA, we present a novel proof that differs from the error decomposition used for MYIPLA.
\begin{theorem}[PIPGLA]\label{th:pipgla_main_text}
    Let \textbf{A\ref{assumption_1}}, \textbf{A\ref{assumption_2}}, \textbf{A\ref{assumption_3}} and \textbf{C\ref{assumption_1_gla}}
 hold. 
 Let $\theta_n^N$ denote the iterate \eqref{eq:pip-gla-theta} and $\Bar{\theta}_{\star}$ be the maximiser of $p_\theta(y)$.
 Then for $\gamma\leq1/L_{g_1}$ and $\gamma\leq \lambda\leq\gamma/(1-\mu\gamma)$,  
 the following holds
\begin{align}
\label{eq:pipgla_bound}
    \bE&\left[\|\theta_n^N - \bar{\theta}_\star\|^2\right]^{1/2} \leq \sqrt{\frac{\lambda^{n}(1-\gamma\mu)^{n}}{\gamma^{n}}} W_2(\mathcal{L}({Z_0^N}), \pi^N)\nonumber\\
    &+\sqrt{\frac{d_{\theta}}{N\mu}}+\sqrt{\frac{\lambda(2\gamma L_{g_1}(d_\theta +N d_x)+\lambda N C)}{N\left(1-\lambda(1-\mu\gamma)/\gamma\right)}}\nonumber,
\end{align}
\normalsize
for all $n\in\mathbb{N}$, with $Z_0^N$ given in \textbf{A\ref{assumption_2}} and $C>0$ given in \textbf{C\ref{assumption_1_gla}} and independent of $t,n,N,\gamma,d_{\theta}, d_x$.
\end{theorem}
See Appendix \ref{app:pip_gla analysis} for the proof.
Similarly to the previous result, we can split the errors as follows
\begin{equation*}
    W_2(\mathcal{L}({\theta}^N_n), \delta_{\Bar{\theta}_\star})\leq W_2(\delta_{\Bar{\theta}_{\star}}, \pi_{\Theta}^N) + W_2(\pi_{ \Theta}^N, \mathcal{L}({\theta_n^N})).
\end{equation*}
\normalsize
The  \emph{concentration} term $W_2(\delta_{\Bar{\theta}_\star}, \pi_{\Theta}^N)$ can be bounded as in \eqref{eq:concentration_term_aux}, thanks to the strong convexity assumption \textbf{A\ref{assumption_1}} and the fact that the set of non-differentiable points of the potential $U(\theta, x)$ has measure zero.
For the other error term, we derive novel nonasymptotic bounds for general $\lambda$, unlike \citet{NEURIPS2019_6a8018b3_salim} (see Corollary~\ref{cor:pipgla_final} in App.~\ref{app:pip_gla analysis}).
Intuitively, $W_2(\pi_{\Theta}^N, \mathcal{L}({\theta_n^N}))$ controls the convergence of the solution of the SDE to $\pi_{\Theta}^N$ and the error due to time discretisation. The bound for the latter term is derived without introducing the law of the solution of the SDE, $\mathcal{L}(\boldsymbol{\theta}_{n\gamma}^N)$, in contrast to the approach taken for MYIPLA.

\subsection{Algorithm comparison}\label{subsec:algorithm_comparison}
Theorems~\ref{th:myipla_main_text}, ~\ref{th:pipgla_main_text} and ~\ref{th:mypgd} (Appendix) allow us to derive complexity estimates for $\lambda$, the number of particles $N$, $\gamma$, and the number of steps $n$ to achieve an error $\mathbb{E}\left[\Vert \theta_n^N-\bar{\theta}^\star \Vert^2\right]^{1/2}=\mathcal{O}(\varepsilon)$, see blue columns of Table \ref{tab:comparison_table}. These bounds are expressed in terms of the key parameters $d_\theta, d_x$.
Details for deriving these bounds are provided in Appendix~\ref{app:algorithm_comparison}.
It is important to mention that although all algorithms yield the same complexity estimates in terms of $\gamma$, MYPGD requires more stringent assumptions on $\gamma$.
Specifically, while MYIPLA (Theorem~\ref{th:myipla_main_text}) requires $\gamma_0 < \min\{(L_{g_1}+\lambda^{-1})^{-1}, 2\mu^{-1}\}$, MYPDG (Theorem~\ref{th:mypgd}) requires $\gamma_0 < (L_{g_1}+\lambda^{-1}+\mu)^{-1}$, which is strictly smaller. 
In contrast, PIPGLA allows for a more flexible choice of $\gamma \leq 1/L_{g_1}$, but requires stronger assumptions on $\lambda$.

We also compare the algorithms in terms of their computational requirements. In particular, we evaluate the computational cost of running each algorithm for $n$ iterations with $N$ particles and a time discretisation step $\gamma$, while guaranteeing an $\mathcal{O}(\varepsilon)$ error.
The comparison is based on the number of component-wise evaluations of $\nabla g_1$ and $\prox_{g_2}^\lambda$, and independent standard $1$-$d$ Gaussians samples. These costs are summarised in the green columns of Table \ref{tab:comparison_table}.

\begin{table*}[t]
\smaller
\caption{Comparison between algorithms (Section \ref{subsec:algorithm_comparison}). $\delta>0$ is any small positive constant.}\label{tab:comparison_table}
\begin{tabular}{|c ||c |c |c|c|c|c|}
\hline
\cellcolor{white!10} & \cellcolor{blue!10}  $\lambda$ & \cellcolor{blue!10}  $N$ & \cellcolor{blue!10}  $\gamma$ & \cellcolor{blue!10} $n$ &
\cellcolor{green!10} Evaluations of $\nabla g_1$  and $\prox_{g_2}^\lambda$ &
\cellcolor{green!10} Indep. $1$d Gaussians
\\
\hline \hline
\cellcolor{white!10}\bfseries  MYIPLA & \cellcolor{blue!10} $\mathcal{O}(\varepsilon)$ & \cellcolor{blue!10} $\mathcal{O}(d_\theta \varepsilon^{-2})$ & \cellcolor{blue!10} $\mathcal{O}(d_x^{-1}\varepsilon^2)$ & \cellcolor{blue!10} $\mathcal{O}(d_x \varepsilon^{-2-\delta})$ &\cellcolor{green!10} $\mathcal{O}(d_\theta d_x(d_\theta+d_x)\varepsilon^{-4-\delta})$ & \cellcolor{green!10} $\mathcal{O}(d_\theta d_x^2\varepsilon^{-4-\delta})$\\
\hline
\cellcolor{white!10}\bfseries  PIPGLA & \cellcolor{blue!10}  $\mathcal{O}(\varepsilon^2)$ & \cellcolor{blue!10} $\mathcal{O}(d_\theta \varepsilon^{-2})$ & \cellcolor{blue!10} $\mathcal{O}(d_x^{-1}\varepsilon^2)$ & \cellcolor{blue!10} $\mathcal{O}(\log \varepsilon^2/\log d_x)$& \cellcolor{green!10} $\mathcal{O}(d_\theta (d_\theta+d_x)\varepsilon^{-2}\frac{\log \varepsilon^{2}}{\log d_x})$  & \cellcolor{green!10}  $\mathcal{O}(d_\theta d_x\varepsilon^{-2}\frac{\log \varepsilon^{2}}{\log d_x})$\\
\hline
\cellcolor{white!10}\bfseries  MYPGD & \cellcolor{blue!10}  $\mathcal{O}(\varepsilon)$ & \cellcolor{blue!10} $\mathcal{O}(d_x \varepsilon^{-2})$ & \cellcolor{blue!10} $\mathcal{O}(d_x^{-1}\varepsilon^2)$ & \cellcolor{blue!10} $\mathcal{O}(d_x \varepsilon^{-2-\delta})$ & \cellcolor{green!10} $\mathcal{O}(d_x^2(d_\theta+d_x)\varepsilon^{-4-\delta})$ & \cellcolor{green!10} $\mathcal{O}( d_x^3\varepsilon^{-4-\delta})$\\
\hline
\end{tabular}
\end{table*}

\begin{table*}[b]
\smaller
\caption{Performance of Bayesian logistic regression for Laplace and uniform priors.} 
\label{table-logistic-comparison}
\centering
\begin{tabular}{llllll}
  \toprule
  Algorithm     & Approx./Iterative  & \multicolumn{2}{c}{NMSE (\%)} & \multicolumn{2}{c}{Times (s)} \\
  \cmidrule(r){3-4}
  \cmidrule(r){5-6}
       &   &  Laplace & Unif & Laplace & Unif \\
  \midrule
 \multirow{2}{*}{MYPGD} & Approx  & $6.09\pm 0.34$  &   $\mathbf{0.60\pm 0.23}$& $91.9\pm4.8$ &$109.3\pm4.6$  \\
 & Iterative  & $4.44\pm1.40$  &   $-$& $129.7\pm 15.8$ & $-$  \\
  \midrule
  \multirow{2}{*}{MYIPLA} & Approx  & $4.42\pm 1.32$  & $15.26\pm 4.44$ & $\mathbf{89.9\pm4.2}$& $\mathbf{97.0\pm4.2}$\\
 & Iterative  &  $4.67\pm1.60$ &   $-$& $120.5\pm10.1$ & $-$  \\
 \midrule
  \multirow{2}{*}{PIPGLA} & Approx  & $2.30\pm0.58$  & $6.83\pm 3.97$ & $116.5\pm5.5$& $103.1\pm8.0$\\
   & Iterative  & $\mathbf{2.02\pm0.54}$  & $-$ & $122.9\pm6.9$& $-$\\
  \bottomrule
\end{tabular}
\end{table*}
\normalsize

\section{Experiments}\label{sec:experiments}
The code is available in \url{https://github.com/paulaoak/proximal-ipla}. 
Additional experiments for low-rank matrix completion are provided in Appendix \ref{app:low_rank_matrix}, where the goal is to recover a low-rank matrix from partially observed and noisy data.

\subsection{Bayesian Logistic Regression}\label{subsec:bayesian_logistic_regression}
We consider a similar set-up to \citet{RefWorks:RefID:73-de2021efficient}
and employ a synthetic dataset consisting of $d_y=900$ datapoints (see Appendix \ref{ap:experiments_logistic} for details). The latent variables are the $d_x=50$ regression weights, to which we assign an isotropic Laplace prior 
$p_{\theta}(x) = \prod_{i=1}^{d_x} \text{Laplace}(x_i|\theta, 1)$ or a uniform prior $p_{\theta}(x) = \prod_{i=1}^{d_x} \mathcal{U}(x_i|-\theta, \theta)$.  The likelihood is given by
$p_{\theta}(y|x) =\prod_{j=1}^{d_y} s(v_j^T x)^{y_j}s(-v_j^T x)^{1-y_j}$, where $v_j \sim \mathcal{U}(-1, 1)^{\otimes d_x}$ are a set of $d_x$-dimensional covariates sampled from a uniform distribution and $s(u)$ is the logistic function. 
The true value of $\theta$ is set randomly to $\theta = -4$ for the Laplace prior and $\theta = 1.5$ for the uniform one.

Approximations of the proximal mapping for both priors are derived in Appendices~\ref{ap:proximal_op_laplace_mean} and \ref{ap:proximal_op_uniform_mean}, as closed-form solutions are unavailable. 
We compare these approximations to an iterative approach for computing the proximal map, which is feasible only for the Laplace prior due to instabilities in the uniform case.
Figure~\ref{fig:logistic-rates} shows the variance of the $\theta$ estimates produced by MYIPLA and PIPGLA computed over 100 runs for different initialisations (left) and the sequence of $\theta$ estimates for 50 particles (right) in the case of a Laplace prior. We observe that the variance of the parameter estimates decreases with rate $\mathcal{O}(1/N)$ as suggested by Theorems~\ref{th:myipla_main_text} and \ref{th:pipgla_main_text}, and that the iterative algorithms have a slightly lower variance compared to their approximate versions, with PIPGLA having better performance than MYIPLA.
It is also important to highlight that for all algorithms considered, approximate solvers are on average 25\% faster than iterative solvers (see Table \ref{table-logistic-comparison}). Besides, in Table \ref{table-logistic-comparison} (expanded in Table \ref{table-logistic-comparison-extended}), we compare the performances of the different proximal algorithms through the normalised MSEs (NMSE) for $\theta$. 
In the uniform case, MYPGD outperforms all other algorithms; this is likely due to the lack of diffusive term in the corresponding SDE which is beneficial when dealing with a compactly supported prior. While PIPGLA also produces estimates within the support, they exhibit a larger bias.

\begin{figure}[t]
    \centering
    \begin{subfigure}[b]{0.23\textwidth}
        \centering
        \includegraphics[width=\textwidth]{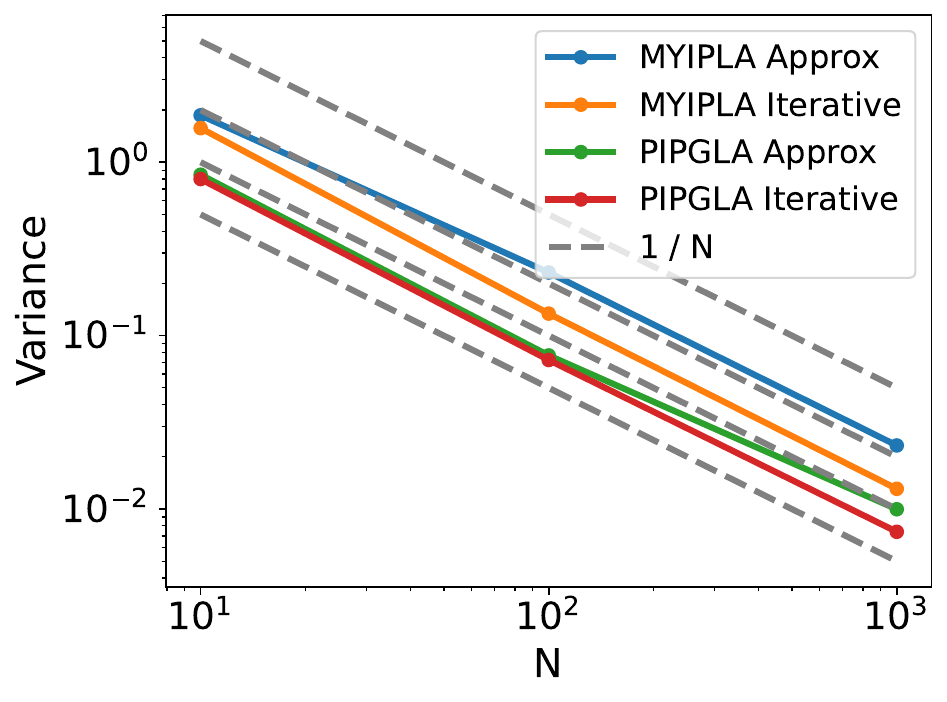}
        \label{fig:subfigure1}
    \end{subfigure}
    \hfill
    \begin{subfigure}[b]{0.23\textwidth}
        \centering
        \includegraphics[width=\textwidth]{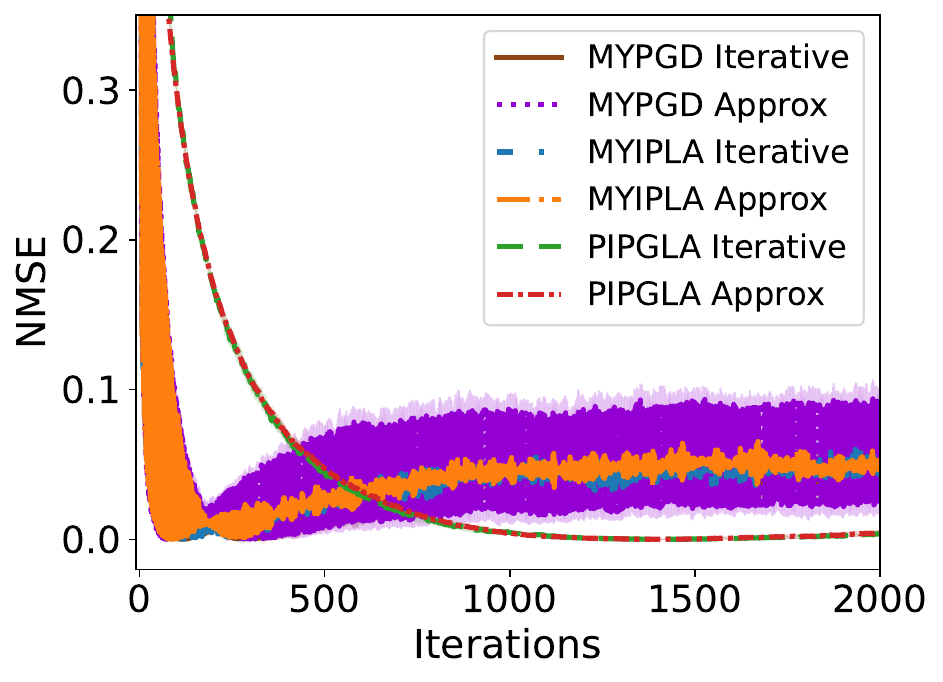}
        \label{fig:subfigure2}
    \end{subfigure}
    \caption{Laplace prior. Left: convergence rate of the variance of the parameter estimates against $N$ produced by MYIPLA and PIPGLA over 100 runs. We see that the $\mathcal{O}(1/N)$ convergence rate holds for the second moments. Right: evolution of the normalised MSE for 50 particles over 100 runs.}
    \label{fig:logistic-rates}
\end{figure}

\subsection{Bayesian Neural Network}\label{subsec:bayesian_nn} 
To evaluate our algorithms on complex multimodal posteriors, we consider a Bayesian neural network with a sparsity-inducing prior on the weights for MNIST image classification.
Following \citet{JMLR:v23:20-1426} and \citet{pmlr-v206-kuntz23a}, we use a two-layer network with \emph{tanh} activation functions and avoid the cost of computing the gradients on a big dataset by subsampling 1000 data points with labels 4 and 9.  
The input layer of the network has 40 nodes and 784 inputs and the output layer has 2 nodes. 
The latent variables are the weights, $w\in\mathbb{R}^{d_w:=40\times 784}$ and $v\in\mathbb{R}^{d_v:=2\times 40}$ of the input and output layers, respectively.  
We assign priors $p_{\alpha}(w)=\prod_i \text{Laplace}(w_i|0, e^{2\alpha})$ and $p_{\beta}(v)=\prod_i \text{Laplace}(v_i|0, e^{2\beta})$ and learn $\theta = (\alpha, \beta)$ from the data. 
One may ask whether the Laplace prior is more appropriate in this setting than the Normal one.
\citet{jaynes_prior} shows that the Laplace prior naturally arises for Bayesian neural network models (see Appendix \ref{ap:experiments_bnn_sparsity_inducing_prior} for details).
We analyse the sparsity-inducing effect of the Laplace prior by examining the distribution of the weights for a randomly chosen particle from the final particle cloud and comparing it to that obtained with a Normal prior.
We note that our methods (Fig. \ref{fig:laplace}, \ref{fig:laplace_pipgla}) lead to final weights with values highly concentrated around zero in comparison to the Normal prior (Fig. \ref{fig:normal}). 
The sparse representation of our algorithm has the advantage of producing models that are smaller in terms of memory usage when small weights are zeroed out. This is investigated in Table \ref{table-bnn-sparsity} in the Appendix. Furthermore, we compare the performance of the PIPLA family against IPLA which ignores the non-differentiability of the model density. 
Figure~\ref{fig:laplace_ipla} shows that, despite using a Laplace prior, IPLA fails to induce sufficient sparsity compared to our proposed methods.
Quantitative results for the variance of the weights and error metrics are shown in Table \ref{tab:comparison_table}, comparing our approach with other algorithms in the literature.
Appendix \ref{ap:experiments_bnn_addtional_dataset} provides additional results on more complex datasets. In particular, we apply our methods to a classification task using CIFAR10 dataset. 
Furthermore, in Appendix \ref{ap:experiments_bnn_activation_function}, we also explore the application of our methods to neural networks with non-differentiable activation functions, such as ReLU.

\begin{figure}[t]
\centering
  \begin{subfigure}[b]{0.21\textwidth}
      \centering
      \includegraphics[width=\textwidth]{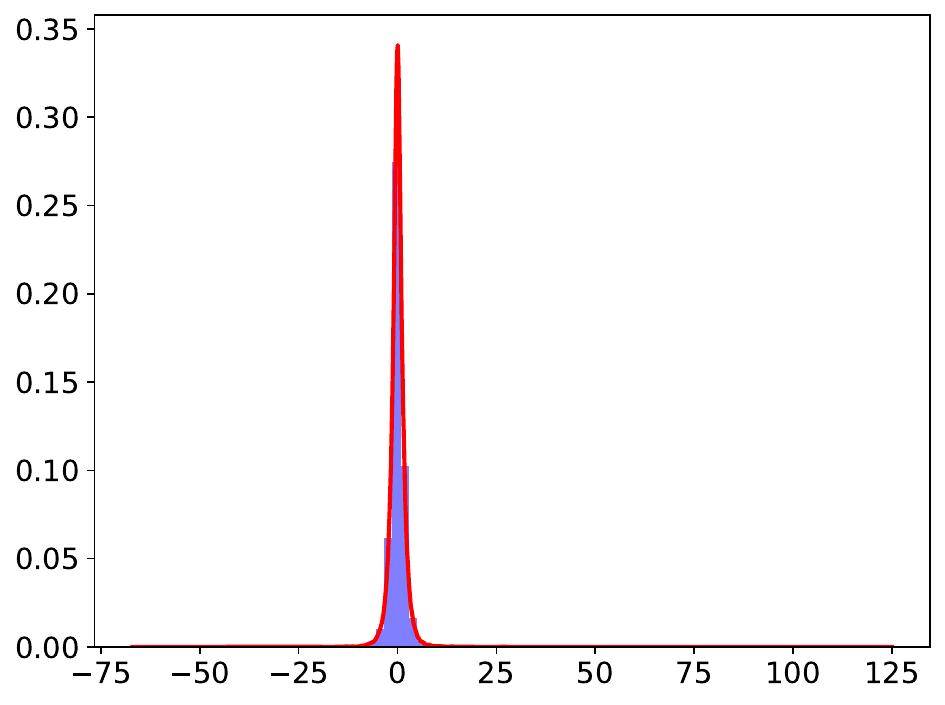}
      \caption{{\small Laplace prior, MYIPLA}}
      \label{fig:laplace}
  \end{subfigure}
  \hfill
        \begin{subfigure}[b]{0.21\textwidth}
      \centering
      \includegraphics[width=\textwidth]{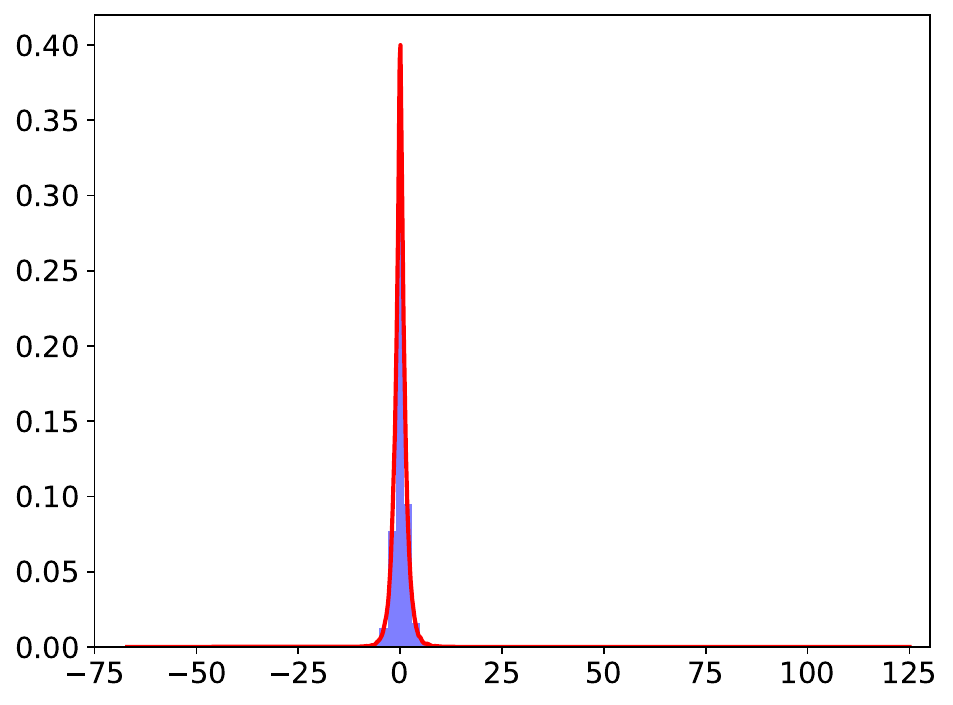} 
      \caption{{\small Laplace prior, PIPGLA}}
      \label{fig:laplace_pipgla}
  \end{subfigure}
    \hfill
\begin{subfigure}[b]{0.21\textwidth}
      \centering
      \includegraphics[width=\textwidth]{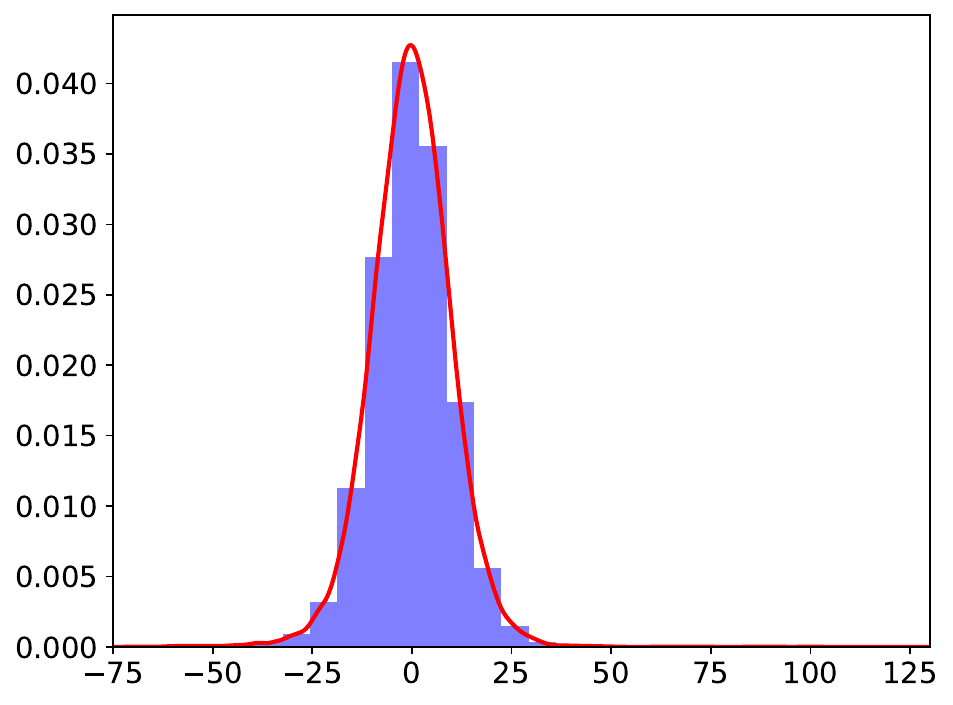} 
      \caption{{\small Normal prior, IPLA}}
      \label{fig:normal}
  \end{subfigure}
  \hfill
    \begin{subfigure}[b]{0.21\textwidth}
      \centering
      \includegraphics[width=\textwidth]{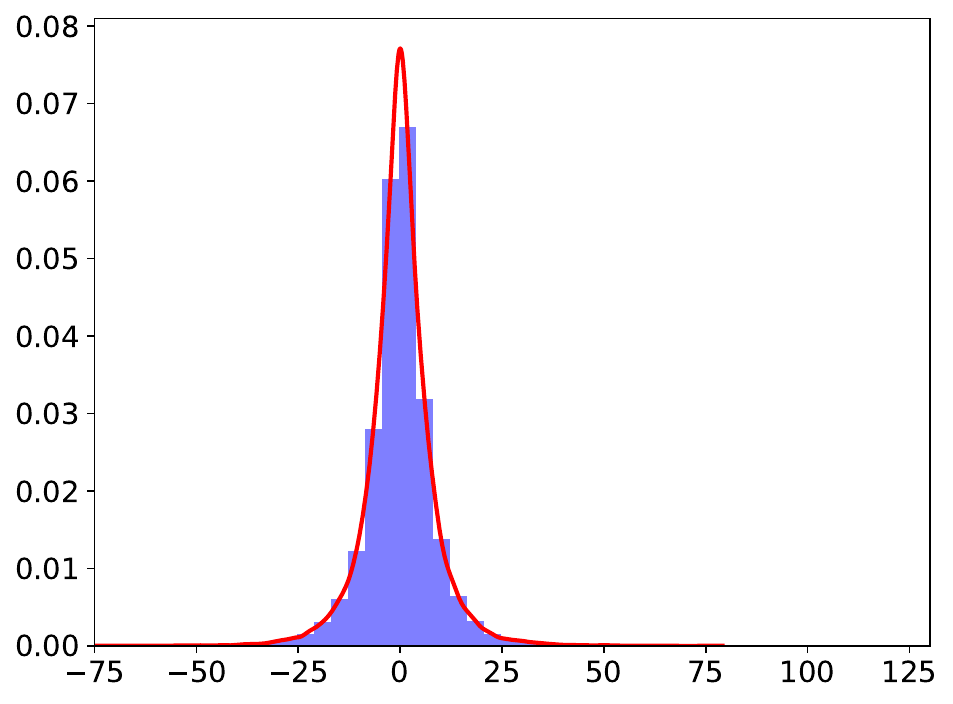}
      \caption{{\small Laplace prior, IPLA}}
      \label{fig:laplace_ipla}
  \end{subfigure}
\caption{Histogram (blue) and density estimation (red) of the BNN weights for a randomly chosen particle. Our methods (top) produce sparser weights, which is crucial for compressibility, compared to IPLA (bottom), which ignores the non-differentiabilities.}
\label{fig:histogram_weights}
\end{figure}

\begin{table}[t]
\smaller
  \caption{Bayesian neural network. Test errors and log pointwise predictive density (LPPD) achieved using the final particle cloud with $N = 50$. Computation times and standard deviation of the empirical distribution of the weight matrix $w$ are also provided.}
  \label{table-bnn-comparison}
  \centering
  \begin{tabular}{lllll}
    \toprule
    Algorithm     & Error $(\%)$     &  LPPD $(\times 10^{-1})$ & Times (s) & Std. $w$\\
    \midrule
    MYPGD & $\mathbf{1.50\pm0.77}$  & $-1.02\pm 0.15$  &   $20$ & $2.02$\\
    MYIPLA & $2.00\pm 0.85$  & $-1.07\pm0.19$ & $22$& $2.27$ \\
    PIPGLA & $2.00\pm0.75$  & $\mathbf{-0.96\pm 0.09}$  &   $33$ & $\mathbf{1.73}$\\
    PGD & $2.00\pm0.98$  & $-0.98\pm 0.10$  &   $\mathbf{19}$ & $8.80$\\
    SOUL &  $6.37\pm1.56$  & $-3.50\pm 2.43$  &   ${55}$ & $12.09$\\
    IPLA & $1.99\pm1.01$  & $-1.01\pm 0.25$  &   $\mathbf{19}$ & $11.70$\\
    \bottomrule
  \end{tabular}
\end{table}
\normalsize

\subsection{Image Deblurring}\label{subsec:image_deblurring}
The objective of image deconvolution is to recover a high-quality image from a blurred and noisy observation $y = Hx +\varepsilon$, where $H$ is a circulant blurring matrix and $\varepsilon\sim\mathcal{N}(0, \sigma^2 I)$. This inverse problem is ill-conditioned, a challenge that Bayesian methods address by incorporating prior knowledge. 
A common choice is the total variation prior, which promotes smoothness and is defined as $ TV(x)=\Vert \nabla_d x\Vert_{1}$, where $\Vert \cdot\Vert_{1}$ is the $\ell_1$ norm and $\nabla_d$ is the two-dimensional discrete gradient operator. 
However, the strength of this prior depends on a hyperparameter $\theta$ that typically requires manual tuning. Instead of fixing this parameter manually, we estimate its optimal value.
Thus, the posterior distribution for the model takes the form
$p_\theta(x|y)\propto C(\theta)\exp\left(-\Vert y- Hx\Vert^2/(2\sigma^2)-e^\theta TV(x)\right)$. For the experiments, we use the algorithm proposed by \citep{douglas_56} to numerically evaluate the proximal operator of the total variation norm.  Qualitative results are presented in Figure \ref{fig:image_deconvolution}, with additional results provided in Appendix \ref{app:image_deblurring}.

\begin{figure}[h]
    \centering

    \begin{subfigure}[b]{0.14\textwidth}
        \centering
        \includegraphics[width=\textwidth]{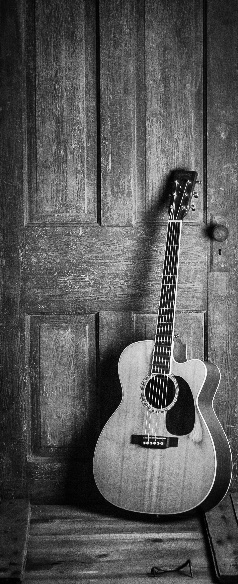}
        \caption{Original}
        \label{fig:subfig1}
    \end{subfigure}
    \hfill
    \begin{subfigure}[b]{0.14\textwidth}
        \centering
        \includegraphics[width=\textwidth]{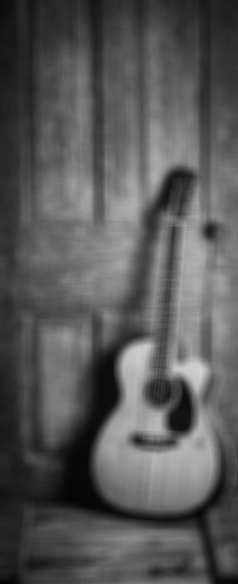}
        \caption{Blurred}
        \label{fig:subfig2}
    \end{subfigure}
    \hfill
    \begin{subfigure}[b]{0.14\textwidth}
        \centering
        \includegraphics[width=\textwidth]{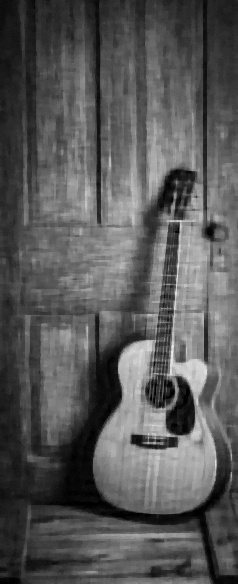}
        \caption{MYIPLA}
        \label{fig:subfig3}
    \end{subfigure}

    \caption{Image deblurring experiment.}
    \label{fig:image_deconvolution}
\end{figure}

\section{Conclusion}\label{sec:conclusion} 
Our algorithms present a novel approach for handling Bayesian models arising from different types of non-differentiable regularisations, including Lasso, elastic net, nuclear-norm and total variation norm.
While our theoretical guarantees are established under strong convexity assumptions, in practice, our methods perform well under more general conditions and demonstrate robustness and stability across a range of regularisation parameter values. 
Moreover, unlike standard Langevin algorithms—which fail to converge for light-tailed distributions \citep{roberts1996exponential}—our proximal variants remain effective, due to the implicit regularisation introduced by the proximal map.
Future work holds many promising avenues. 
Our theoretical framework can be extended to the non-convex setting using recent non-convex optimisation bounds \citep{zhang2023nonasymptotic} and their non-differentiable adaptations. 
Additionally, our novel bounds on the difference between the true minimiser and the minimiser of the Moreau-Yosida approximation can also be used within recent multiscale approaches to extend them to non-differentiable settings, see, e.g., \citet{akyildiz2024multiscale}.

\begin{acknowledgements} 
PCE is supported by EPSRC through the Modern Statistics and Statistical Machine Learning (StatML) CDT programme, grant no. EP/S023151/1.
FRC gratefully acknowledges the ``de Castro" Statistics Initative at the \textit{Collegio Carlo Alberto} and the \textit{Fondazione Franca e Diego de Castro}.
\end{acknowledgements}

\bibliography{references}

\newpage

\onecolumn

\title{Proximal Interacting Particle Langevin Algorithms\\(Supplementary Material)}
\maketitle

\appendix

\section{Theoretical Analysis of Proximal Interacting Langevin Algorithms}\label{ap:theory_IPLA}
\subsection{Approximation of minimisers}
{\normalsize
Before proceeding with the proof of Theorem~\ref{th:myipla_main_text} and~\ref{th:pipgla_main_text}, we derive a result controlling the distance between the maximiser of $p_\theta(y) = \int_{\real^{d_x}} e^{-U(\theta, x)}\md x$, denoted by $\Bar{\theta}_{\star}$, and the maximiser of $p_\theta^\lambda(y)= \int_{\real^{d_x}} e^{-U^\lambda(\theta, x)}\md x$, denoted by $\Bar{\theta}_{\lambda,\star}$.

For simplicity let us denote
\begin{align*}
k_\lambda(\theta) := \int_{\real^{d_x}} e^{-U^\lambda(\theta, x)}\md x,\qquad\qquad k(\theta) := \int_{\real^{d_x}} e^{-U(\theta, x)}\md x
\end{align*}
and $K_\lambda(\theta):=-\log k_\lambda(\theta)$.

Let $\Omega\subset\mathbb{R}^{d_{\theta}}\times \mathbb{R}^{d_x}$ denote the set on which $g_2$ is twice differentiable. Let $\Bar{\theta}_{\star}$ be the maximiser of $k$  
and let $\tilde{\Omega} = \Omega \cap (\{\Bar{\theta}_\star\} \times \bR^{d_x})$.
By \textbf{A}\ref{assumption_4}
\begin{equation*}
    k(\Bar{\theta}_{\star}) = \int e^{-U(\Bar{\theta}_{\star}, x)} \md x = \int_{\tilde{\Omega}} e^{-U(\Bar{\theta}_{\star}, x)} \md x.
\end{equation*}
Since, $k$ achieves a maximum at $\Bar{\theta}_{\star}$ and $g_1,g_2$ are differentiable in $\tilde{\Omega}$, then
\begin{equation}\label{eq:substract}
    0 = \nabla_{\theta} k(\Bar{\theta}_{\star}) = \nabla_{\theta} \int_{\tilde{\Omega}} e^{-U(\Bar{\theta}_{\star}, x)} \md x =- \int_{\tilde{\Omega}} \Big(\nabla_{\theta} g_1(\Bar{\theta}_{\star}, x) + \nabla_{\theta} g_2(\Bar{\theta}_{\star}, x)\Big)e^{-U(\Bar{\theta}_{\star}, x)} \md x.
\end{equation}
\begin{proposition}[Convergence of minimisers]
\label{prop:theta_convergence_lambda}
Under assumption \textbf{A\ref{assumption_1}}, the Lipschitzness of $g_2$ given by \textbf{A\ref{assumption_1prime}}, the strong convexity assumption \textbf{A\ref{assumption_3}} and \textbf{A\ref{assumption_4}}, it follows that
    \begin{equation*}
\Vert \Bar{\theta}_{\lambda,\star}-\Bar{\theta}_{\star}\Vert \leq  \frac{\lambda}{\mu} \Big(\frac{\Vert g_2\Vert_{\textnormal{Lip}}^2}{2} A  +  B\Big) +\mathcal{O}(\lambda^2),
\end{equation*}
where $\Vert g_2\Vert_{\textnormal{Lip}}$ is the Lipschitz constant of $g_2$ and $A, B$ are given in \textbf{A\ref{assumption_4}}.
\end{proposition}
\begin{proof}
To obtain a bound of $\Vert \Bar{\theta}_{\lambda,\star}-\Bar{\theta}_{\star}\Vert$ in terms of $\lambda$, we first define the measure $\pi_{\lambda}^1\propto e^{-U^{\lambda}(\theta,x)}$ and observe that $\pi_{\lambda}^1$ is $\mu$-strongly log-concave, since $\pi_{\lambda}^1\propto e^{-U^{\lambda}(\theta, x)}$ and $U^{\lambda}$ is strongly convex by \textbf{A\ref{assumption_3}}.
Therefore, by a form of the Prékopa-Leindler inequality for strong convexity \citep[Theorem 3.8]{saumard2014log}, $\pi_{\lambda, \Theta}^1\propto e^{-K_{\lambda}(\theta)} = k_{\lambda}(\theta)$ is $\mu$-strongly log-concave, which results in 
\begin{equation}\label{eq:prekopa}
    \langle\Bar{\theta}_{\lambda,\star} -\Bar{\theta}_{\star}, \nabla K_{\lambda}(\Bar{\theta}_{\lambda,\star})-\nabla K_{\lambda}(\Bar{\theta}_{\star})\rangle\geq \mu \Vert\Bar{\theta}_{\lambda,\star} -\Bar{\theta}_{\star}\Vert^2.
\end{equation} 
Since, $\Bar{\theta}_{\lambda,\star}$ is the maximiser of $k_{\lambda}(\theta)$ and $k_{\lambda}(\theta)$ is differentiable, it follows that $\nabla k_{\lambda}(\Bar{\theta}_{\lambda,\star})= 0$ and therefore $\nabla K_{\lambda}(\Bar{\theta}_{\lambda,\star})= 0$. Using the Cauchy-Schwarz inequality, we can rearrange \eqref{eq:prekopa} to obtain
\begin{equation}\label{eq:theta_bound}
    \Vert\Bar{\theta}_{\lambda,\star} -\Bar{\theta}_{\star}\Vert \leq \frac{1}{\mu}\Vert\nabla K_{\lambda}(\Bar{\theta}_{\lambda,\star})-\nabla K_{\lambda}(\Bar{\theta}_{\star})\Vert = \frac{1}{\mu}\Vert \nabla K_{\lambda}(\Bar{\theta}_{\star})\Vert = \frac{1}{\mu k_{\lambda}(\Bar{\theta}_{\star})}\Vert \nabla k_{\lambda}(\Bar{\theta}_{\star})\Vert.
\end{equation}
We now focus on the term $\Vert \nabla k_{\lambda}(\Bar{\theta}_{\star})\Vert$
\begin{equation*}
    \Vert \nabla k_{\lambda}(\Bar{\theta}_{\star})\Vert = \Big\Vert \nabla_{\theta} \int e^{-U^{\lambda}(\Bar{\theta}_{\star}, x)} \md x\Big\Vert = \Big\Vert \int \nabla_{\theta}U^{\lambda}(\Bar{\theta}_{\star}, x) e^{-U^{\lambda}(\Bar{\theta}_{\star}, x)} \md x\Big\Vert.
\end{equation*}

Recall that $U^{\lambda}(\theta, x) = g_1(\theta, x) + g_2^{\lambda}(\theta, x)$.
For simplicity, let us assume that $\nabla_{\theta} g_1(\theta, x) = 0$, later we will show that this condition is not necessary. Then, we have that
\begin{equation}\label{eq:bound}
    \Vert \nabla k_{\lambda}(\Bar{\theta}_{\star})\Vert =\Big\Vert \int \nabla_{\theta}g_2^{\lambda}(\Bar{\theta}_{\star}, x) e^{-U^{\lambda}(\Bar{\theta}_{\star}, x)} \md x\Big\Vert = \Big\Vert \int \frac{\Bar{\theta}_{\star}-\prox_{g_2}
^{\lambda}(\Bar{\theta}_{\star}, x)_{\theta}}{\lambda}  e^{-U^{\lambda}(\Bar{\theta}_{\star}, x)} \md x\Big\Vert.
\end{equation}
Since $g_2$ is convex, the problem $\min_u\left\lbrace g_2(u) + \Vert v-u\Vert^2/(2\lambda) \right\rbrace$ with $v=(\theta, x)$ has a unique minimum $w$ that satisfies $\lambda \nabla g_2(w) - (v-w) = 0$.
We consider the implicit system $\phi(\lambda, w) = \lambda \nabla g_2(w) - (v-w)$, and note that $\phi(0, v) = 0$ and 
\begin{equation*}
    \frac{\partial \phi(\lambda, w)}{\partial w} = \lambda\nabla^2 g_2(w) + I\succ 0,
\end{equation*}
i.e. positive definite due to  {Remark \ref{remark_differentiability}} and assumption \textbf{A\ref{assumption_4}}.
Thus the Jacobian of $\phi$ w.r.t. $w$ is invertible.
Hence, the implicit function theorem shows that there is some locally defined $\zeta$, such that $\zeta(0) = v$ and $\phi(\lambda, \zeta(\lambda)) = 0$. Furthermore, 
\begin{align*}
    \frac{\partial\phi(\lambda, \zeta(\lambda))}{\partial\lambda}\bigg\vert_{\lambda = 0} &= \bigg(\nabla g_2(\zeta(\lambda)) + \lambda \nabla^2 g_2(\zeta(\lambda))\frac{\partial\zeta(\lambda)}{\partial\lambda} + \frac{\partial\zeta(\lambda)}{\partial\lambda}\bigg)\bigg\vert_{\lambda = 0} = 0,\\
    \frac{\partial^2\phi(\lambda, \zeta(\lambda))}{\partial\lambda^2}\bigg\vert_{\lambda = 0} &= \bigg(2\nabla^2 g_2(\zeta(\lambda))\frac{\partial\zeta(\lambda)}{\partial\lambda}+ \lambda \frac{\partial}{\partial\lambda}\Big(\nabla^2 g_2(\zeta(\lambda))\frac{\partial\zeta(\lambda)}{\partial\lambda}\Big) + \frac{\partial^2\zeta(\lambda)}{\partial^2\lambda}\bigg)\bigg\vert_{\lambda = 0} = 0,\\
\end{align*}
which provides
\begin{align*}
    \frac{\partial\zeta(0)}{\partial\lambda} &= -\nabla g_2(v),\\
    \frac{\partial^2\zeta(0)}{\partial\lambda^2} &= -2 \nabla^2 g_2(v)\nabla g_2(v).
\end{align*}
Using Taylor's expansion at $\lambda = 0$ we have
\begin{equation}\label{eq:taylors2}
    \prox_{g_2}^{\lambda}(v) = \zeta(\lambda) = v - \lambda\nabla g_2(v) - \lambda^2 \nabla^2 g_2(v)\nabla g_2(v)+\mathcal{O}(\lambda^3).
\end{equation}
Therefore
\begin{equation}\label{eq:proximal_operator2}
    \frac{\Bar{\theta}_{\star}-\prox_{g_2}
^{\lambda}(\Bar{\theta}_{\star}, x)_{\theta}}{\lambda} = \nabla_{\theta} g_2(\Bar{\theta}_{\star}, x) + \lambda[\nabla^2_{(\theta, x)} g_2(\Bar{\theta}_{\star}, x)]_{1:d_\theta}\nabla_{(\theta, x)} g_2(\Bar{\theta}_{\star}, x) + \mathcal{O}(\lambda^2),
\end{equation}
where the notation $[\nabla^2_{(\theta, x)} g_2(\theta, x)]_{1:d_\theta}$ means that we only take the first $d_\theta$ rows of the Hessian. For simplicity, we denote $h(\theta, x)= [\nabla^2_{(\theta, x)} g_2(\theta, x)]_{1:d_\theta}\nabla_{(\theta, x)} g_2({\theta}, x)$.

Substituting \eqref{eq:proximal_operator2} in \eqref{eq:bound}, we have
{\normalsize
\begin{align*}
    \Vert \nabla &k_{\lambda}(\Bar{\theta}_{\star})\Vert = \Big\Vert \int_{\Tilde{\Omega}} (\nabla_{\theta} g_2(\Bar{\theta}_{\star}, x) + \lambda h(\Bar{\theta}_{\star}, x) + \mathcal{O}(\lambda^2)) e^{-U^{\lambda}(\Bar{\theta}_{\star}, x)} \md x\Big\Vert \\ 
    &\leq\Big\Vert \int_{\Tilde{\Omega}} \nabla_{\theta} g_2(\Bar{\theta}_{\star}, x) e^{-U^{\lambda}(\Bar{\theta}_{\star}, x)} \md x\Big\Vert + \lambda\bigg(\int_{\Tilde{\Omega}} \Vert h(\Bar{\theta}_{\star}, X) \Vert \frac{e^{-U^\lambda(\Bar{\theta}_\star, x)}}{k_{\lambda}(\Bar{\theta}_{\star})} \md x\bigg) k_{\lambda}(\Bar{\theta}_{\star})+ \mathcal{O}(\lambda^2) k_{\lambda}(\Bar{\theta}_{\star})\\
    &=\Big\Vert \int_{\Tilde{\Omega}} \nabla_{\theta} g_2(\Bar{\theta}_{\star}, x) e^{-U^{\lambda}(\Bar{\theta}_{\star}, x)} \md x\Big\Vert + \lambda\mathbb{E}_X[\Vert h(\Bar{\theta}_{\star}, X) \Vert]k_{\lambda}(\Bar{\theta}_{\star}) + \mathcal{O}(\lambda^2) k_{\lambda}(\Bar{\theta}_{\star}).
    \end{align*}}
Subtracting~\eqref{eq:substract} in the first term, we have 
{\normalsize
\begin{align*}
    \Vert \nabla &k_{\lambda}(\Bar{\theta}_{\star})\Vert \leq \Big\Vert \int_{\Tilde{\Omega}} \nabla_{\theta} g_2(\Bar{\theta}_{\star}, x) \big(e^{-U^{\lambda}(\Bar{\theta}_{\star}, x)} - e^{-U(\Bar{\theta}_{\star}, x)}\big) \md x\Big\Vert  + \big(\lambda\mathbb{E}_X[\Vert h(\Bar{\theta}_{\star}, X) \Vert]+\mathcal{O}(\lambda^2)\big) k_{\lambda}(\Bar{\theta}_{\star})\\
    &\leq \int_{\Tilde{\Omega}} \Vert\nabla_{\theta} g_2(\Bar{\theta}_{\star}, x)\Vert e^{-U^{\lambda}(\Bar{\theta}_{\star}, x)} \big(1- e^{-U(\Bar{\theta}_{\star}, x) + U^{\lambda}(\Bar{\theta}_{\star}, x)}\big) \md x + \big(\lambda\mathbb{E}_X[\Vert h(\Bar{\theta}_{\star}, X) \Vert]+\mathcal{O}(\lambda^2)\big)k_{\lambda}(\Bar{\theta}_{\star})\\
    &\leq \big(1- e^{-\lambda \Vert g_2\Vert_{\textnormal{Lip}}^2/2}\big) \int_{\Tilde{\Omega}}\Vert\nabla_{\theta} g_2(\Bar{\theta}_{\star}, x)\Vert  e^{-U^{\lambda}(\Bar{\theta}_{\star}, x)} \md x + \big(\lambda\mathbb{E}_X[\Vert h(\Bar{\theta}_{\star}, X) \Vert]+\mathcal{O}(\lambda^2)\big)k_{\lambda}(\Bar{\theta}_{\star})\\
    &= \big(1- e^{-\lambda \Vert g_2\Vert_{\textnormal{Lip}}^2/2}\big) \bigg(\int_{\Tilde{\Omega}}\Vert\nabla_{\theta} g_2(\Bar{\theta}_{\star}, x)\Vert  \frac{e^{-U^{\lambda}(\Bar{\theta}_{\star}, x)}}{k_{\lambda}(\Bar{\theta}_{\star})} \md x\bigg)\; k_{\lambda}(\Bar{\theta}_{\star}) \\ &+ \big(\lambda\mathbb{E}_X[\Vert h(\Bar{\theta}_{\star}, X) \Vert]+\mathcal{O}(\lambda^2)\big)k_{\lambda}(\Bar{\theta}_{\star})
\end{align*}}
where we used the fact that, since $g_2$ is Lipschitz by \textbf{A\ref{assumption_1}}, $0\leq U(v)-U^{\lambda}(v) \leq \frac{\lambda \Vert g_2\Vert_{\textnormal{Lip}}^2}{2}$ for $v=(\theta, x)$, as shown in the proof of \citet[Proposition 1]{durmus_proximal}.

By \textbf{A\ref{assumption_4}}, we further have
$\mathbb{E}_X[\Vert \nabla_{\theta} g_2(\Bar{\theta}_{\star}, X) \Vert]\leq A$, $\mathbb{E}_X[\Vert h(\Bar{\theta}_{\star}, X) \Vert]\leq B$ and thus
{\normalsize
\begin{align*}
  \Vert \nabla k_{\lambda}(\Bar{\theta}_{\star})\Vert
    &\leq\big(1- e^{-\lambda \Vert g_2\Vert_{\textnormal{Lip}}^2 /2}\big)  \mathbb{E}_X [\Vert\nabla_{\theta} g_2(\Bar{\theta}_{\star}, X)\Vert] k_{\lambda}(\Bar{\theta}_{\star}) + \big(\lambda\mathbb{E}_X[\Vert h(\Bar{\theta}_{\star}, X) \Vert]+\mathcal{O}(\lambda^2)\big)k_{\lambda}(\Bar{\theta}_{\star})\\
   &= \bigg(\lambda  \frac{\Vert g_2\Vert_{\textnormal{Lip}}^2}{2} \mathbb{E}_X [\Vert\nabla_{\theta} g_2(\Bar{\theta}_{\star}, X)\Vert]  + \lambda\mathbb{E}_X[\Vert h(\Bar{\theta}_{\star}, X) \Vert]+\mathcal{O}(\lambda^2)\bigg) k_{\lambda}(\Bar{\theta}_{\star})\\
    & \leq \lambda \Big(\frac{\Vert g_2\Vert_{\textnormal{Lip}}^2}{2} A  +  B\Big) k_{\lambda}(\Bar{\theta}_{\star})+\mathcal{O}(\lambda^2)k_{\lambda}(\Bar{\theta}_{\star}). \numberthis \label{eqn2}
\end{align*}}
Putting together \eqref{eq:theta_bound} and  \eqref{eqn2}, we get that
\begin{equation*}
    \Vert\Bar{\theta}_{\lambda,\star} -\Bar{\theta}_{\star}\Vert \leq \frac{1}{\mu k_{\lambda}(\Bar{\theta}_{\star})}\Vert \nabla k_{\lambda}(\Bar{\theta}_{\star})\Vert\leq \frac{\lambda}{\mu} \Big(\frac{\Vert g_2\Vert_{\textnormal{Lip}}^2}{2} A  +  B\Big) +\frac{1}{\mu}\mathcal{O}(\lambda^2) = \frac{\lambda}{\mu} \Big(\frac{\Vert g_2\Vert_{\textnormal{Lip}}^2}{2} A  +  B\Big) +\mathcal{O}(\lambda^2).
\end{equation*}

For the case $\nabla_{\theta} g_1(\theta, x) \neq 0$, the same results follows since
{\normalsize
\begin{align*}
    \Vert \nabla &k_{\lambda}(\Bar{\theta}_{\star})\Vert = \Big\Vert \int_{\Tilde{\Omega}} \Big(\nabla_{\theta} g_1(\Bar{\theta}_{\star}, x) + \frac{\Bar{\theta}_{\star}-\prox_{g_2}
^{\lambda}(\Bar{\theta}_{\star}, x)}{\lambda}\Big)  e^{-U^{\lambda}(\Bar{\theta}_{\star}, x)} \md x\Big\Vert \\
&\leq \Big\Vert \int_{\Tilde{\Omega}} \big(\nabla_{\theta} g_1(\Bar{\theta}_{\star}, x) + \nabla_{\theta} g_2(\Bar{\theta}_{\star}, x)\big) e^{-U^{\lambda}(\Bar{\theta}_{\star}, x)} \md x\Big\Vert + \big(\lambda\mathbb{E}_X[\Vert h(\Bar{\theta}_{\star}, X) \Vert]+\mathcal{O}(\lambda^2)\big) k_{\lambda}(\Bar{\theta}_{\star})\\
&= \Big\Vert \int_{\Tilde{\Omega}} \nabla_{\theta} U(\Bar{\theta}_{\star}, x) \big(e^{-U^{\lambda}(\Bar{\theta}_{\star}, x)} - e^{-U(\Bar{\theta}_{\star}, x)}\big) \md x\Big\Vert + \big(\lambda\mathbb{E}_X[\Vert h(\Bar{\theta}_{\star}, X) \Vert]+\mathcal{O}(\lambda^2)\big) k_{\lambda}(\Bar{\theta}_{\star})\\
&\leq  \int_{\Tilde{\Omega}}\Vert\nabla_{\theta} U(\Bar{\theta}_{\star}, x)\Vert   e^{-U^{\lambda}(\Bar{\theta}_{\star}, x)} \big(1- e^{-U(\Bar{\theta}_{\star}, x) + U^{\lambda}(\Bar{\theta}_{\star}, x)}\big) \md x \\
&+ \big(\lambda\mathbb{E}_X[\Vert h(\Bar{\theta}_{\star}, X) \Vert]+\mathcal{O}(\lambda^2)\big) k_{\lambda}(\Bar{\theta}_{\star})\\
&\leq  \big(1- e^{-\lambda \Vert g_2\Vert_{\textnormal{Lip}}^2/2}\big) \int_{\Tilde{\Omega}} \Vert\nabla_{\theta} U(\Bar{\theta}_{\star}, x)\Vert e^{-U^{\lambda}(\Bar{\theta}_{\star}, x)} \md x + \big(\lambda\mathbb{E}_X[\Vert h(\Bar{\theta}_{\star}, X) \Vert]+\mathcal{O}(\lambda^2)\big) k_{\lambda}(\Bar{\theta}_{\star})\\
    &= \big(1- e^{-\lambda \Vert g_2\Vert_{\textnormal{Lip}}^2/2}\big) \mathbb{E}_X[\Vert\nabla_{\theta} U(\Bar{\theta}_{\star}, x)\Vert] k_{\lambda}(\Bar{\theta}_{\star}) + \big(\lambda\mathbb{E}_X[\Vert h(\Bar{\theta}_{\star}, X) \Vert]+\mathcal{O}(\lambda^2)\big) k_{\lambda}(\Bar{\theta}_{\star})\\
    &= \bigg(\lambda  \frac{\Vert g_2\Vert_{\textnormal{Lip}}^2}{2} \mathbb{E}_X [\Vert\nabla_{\theta} g_2(\Bar{\theta}_{\star}, X)\Vert]  + \lambda\mathbb{E}_X[\Vert h(\Bar{\theta}_{\star}, X) \Vert]+\mathcal{O}(\lambda^2)\bigg) k_{\lambda}(\Bar{\theta}_{\star})\\
    & \leq \lambda \Big(\frac{\Vert g_2\Vert_{\textnormal{Lip}}^2}{2} A  +  B\Big) k_{\lambda}(\Bar{\theta}_{\star})+\mathcal{O}(\lambda^2)k_{\lambda}(\Bar{\theta}_{\star}).
\end{align*}}
\end{proof}}

\subsection{MYIPLA}\label{ap:MYIPLA}
{\normalsize
Following \citet{akyildiz2023interacting}, we have the following results.

\begin{proposition} \label{prop:strong_solution}
Assuming conditions \textbf{A\ref{assumption_1}} and \textbf{A\ref{assumption_2}} hold, there exist a unique strong solution to \eqref{eq:particles_proximal_1}--\eqref{eq:particles_proximal_2}.
\end{proposition} 
\begin{proof}
    The proof follows from \citet{karatzas} and \citet[Proposition 1]{akyildiz2023interacting}.
\end{proof}
\begin{proposition}[Invariant measure]\label{prop:invariant_measure}
    For any $N\in\mathbb{N}$, the measure $\pi_{\lambda, \star}^N(\theta,x_1,\dots,x_N)\propto \exp(-\sum_{i=1}^N U^{\lambda}(\theta, x_i))$ is an invariant measure for the interacting particle system \eqref{eq:particles_proximal_1}-\eqref{eq:particles_proximal_2}.
\end{proposition}
\begin{proof}
    The proof follows from Proposition 2 of \citet{akyildiz2023interacting}.
\end{proof}
Therefore, the system \eqref{eq:particles_proximal_1}-\eqref{eq:particles_proximal_2} has an invariant measure which admits
\begin{equation*}    
\pi_{\lambda,\Theta}^N(\theta) \propto \int_{\mathbb{R}^{d_x}}\dots \int_{\mathbb{R}^{d_x}} e^{-\sum_{i=1}^N U^{\lambda}(\theta, x_i)} \md x_1\dots \md{x_{N}} = \Big(\int_{\mathbb{R}^{d_x}} e^{-U^{\lambda}(\theta, x)}\md x\Big)^N = k_{\lambda}(\theta)^N,
\end{equation*}
as $\theta$-marginal and can therefore act as a global optimiser of $k_{\lambda}(\theta)$, or more precisely of $\log k_{\lambda}(\theta)$. 
That is, let $K_{\lambda}(\theta) = -\log k_{\lambda}(\theta)$, then $
    \pi_{\lambda,\Theta}^N(\theta)\propto \exp({-NK_{\lambda}(\theta)})$, concentrates around the minimiser of $K_{\lambda}(\theta)$, hence the maximiser of $k_{\lambda}(\theta)$ 
as $N\to\infty$. This is a classical setting in global optimisation, where $N$ acts as the inverse temperature parameter. We now analyse the rate at which $\pi_{\lambda, \Theta}$ concentrates around the maximiser of $k(\theta)$. 

\begin{proposition}[Concentration bound]\label{prop:concentration_bound}
    Let $\pi_{\lambda,\Theta}^N$ be as defined above and $\Bar{\theta}_{\star}$, $\Bar{\theta}_{\lambda,\star}$ be the maximisers of $k(\theta)$, $k_{\lambda}(\theta)$, respectively.  Then, under assumption \textbf{A\ref{assumption_1}}, the Lipschitzness of $g_2$ (\textbf{A\ref{assumption_1prime}}), the strong convexity assumption \textbf{A\ref{assumption_3}} and assumption \textbf{A\ref{assumption_4}}, it follows
    \begin{equation*}
W_2(\pi_{\lambda,\Theta}^N,\delta_{\Bar{\theta}_{\star}}) \leq W_2(\pi_{\lambda,\Theta}^N,\delta_{\Bar{\theta}_{\lambda,\star}}) + \Vert \Bar{\theta}_{\lambda,\star}-\Bar{\theta}_{\star}\Vert \leq \sqrt{\frac{d_{\theta}}{\mu N}} + \frac{\lambda}{\mu} \Big(\frac{\Vert g_2\Vert_{\textnormal{Lip}}^2}{2} A  +  B\Big) +\mathcal{O}(\lambda^2),
\end{equation*}
where $d_{\theta}$ is the dimension of the parameter space $\Theta$ and $\Vert g_2\Vert_{\textnormal{Lip}}$ is the Lipschitz constant for $g_2$.
\end{proposition}

\begin{proof}
Using a form of the Prékopa-Leindler inequality for strong convexity \citep[Theorem 3.8]{saumard2014log}, $\pi_{\lambda, \Theta}^N$ is $N\mu$-strongly log-concave. Since it is also smooth, we can apply Lemma A.8 of \citet{sinho_lemma_A8} to obtain a convergence bound for $W_2(\pi_{\lambda,\Theta}^N,\delta_{\Bar{\theta}_{\lambda,\star}})$, 
\begin{equation}\label{eq:bound_11}
W_2(\pi_{\lambda,\Theta}^N,\delta_{\Bar{\theta}_{\lambda,\star}})\leq \sqrt{\frac{d_{\theta}}{\mu N}}.
\end{equation}
On the other hand, the 2-Wasserstein distance between two degenerate distributions satisfies
\begin{equation}\label{eq:bound_22}
W_2(\delta_{\Bar{\theta}_{\lambda,\star}},\delta_{\Bar{\theta}_{\star}}) = \Vert \Bar{\theta}_{\lambda,\star}-\Bar{\theta}_{\star}\Vert. 
\end{equation}
By the triangular inequality the Wasserstein distance $W_2(\pi_{\lambda,\Theta}^N,\delta_{\Bar{\theta}_{\star}})$ is upper bounded by the sum of \eqref{eq:bound_11}-\eqref{eq:bound_22}.
We then conclude using Proposition~\ref{prop:theta_convergence_lambda}.
\end{proof}

The main difference with earlier works in the previous concentration bound are the second and third terms, which result from the Moreau-Yosida approximation of the non-differentiable target density $\pi$.
Following on the assumptions made above and the smoothness of $\pi_{\lambda,\Theta}^N$, we have similar results to Propositions 4 and 5 of \citet{akyildiz2023interacting}, establishing exponential ergodicity of \eqref{eq:particles_proximal_1}-\eqref{eq:particles_proximal_2} and analysing the time-discretised scheme \eqref{eq:pip-myula_theta}-\eqref{eq:pip-myula_x}.

Combining all these results, we can provide specific bounds on the accuracy of MYIPLA in terms of $N,\gamma, n, \lambda$ and the convexity properties of $U$.
\begin{theorem}[Theorem~\ref{th:myipla_main_text} restated]
    Let \textbf{A\ref{assumption_1}}--\textbf{A\ref{assumption_4}} hold. Then for every $\lambda$ and $\gamma_0\in(0, \min\{(L_{g_1}+\lambda^{-1})^{-1}, 2\mu^{-1}\})$ there exist constants $C_1>0$ independent of $t,n,N,\gamma,\lambda,d_{\theta}, d_x$ such that for every $\gamma\in(0,\gamma_0]$, one has
\begin{align*}
    \mathbb{E}[\Vert\theta_n^N-\Bar{\theta}_{\star}\Vert^2]^{1/2}\leq &\sqrt{\frac{d_{\theta}}{N\mu}} + \frac{\lambda}{\mu} \Big(\frac{\Vert g_2\Vert_{\textnormal{Lip}}^2}{2} A  +  B\Big) +e^{-\mu n\gamma}\bigg(\mathbb{E}[\Vert Z_0^N-z_{\star}\Vert^2]^{1/2}+\Big(\frac{d_xN+d_{\theta}}{N\mu}\Big)^{1/2}\bigg) \\&+ C_1(1+\sqrt{d_{\theta}/N+d_x})\gamma^{1/2} +\mathcal{O}(\lambda^2)
\end{align*}
for all $n\in\mathbb{N}$, where $z_\star = (\theta_*, N^{-1/2}x_\star,\dots,N^{-1/2}x_\star)$ and $(\theta_\star, x_\star)$ is the minimiser of $U^{\lambda}$.
\end{theorem}

\begin{proof}
Let us denote by $\mathcal{L}(\theta_n^N)$ the $\theta$-marginal of the measure associated to the law of MYIPLA and  $\mathcal{L}(\boldsymbol{\theta}_t^N)$ the $\theta$-marginal of the measure associated to the law of interacting particle system \eqref{eq:particles_proximal_1}-\eqref{eq:particles_proximal_2} at time $t$.
The expectation of the norm can be decomposed  into a term involving the difference between the maximisers of the marginal maximum likelihood of the Moreau-Yosida approximation of the joint density and the original density, a term concerning the concentration of $\pi_{\lambda,\Theta}^N$ around the marginal maximum likelihood estimator $\Bar{\theta}_\star$, a term describing the convergence of \eqref{eq:particles_proximal_1}-\eqref{eq:particles_proximal_2} to its invariant measure, and a term involving the error induced by the time discretisation.

The first two terms are upper bounded by Proposition \ref{prop:concentration_bound}, the third and fourth inequalities result from Proposition 4 and Proposition 5 of \citet{akyildiz2023interacting}.
\end{proof}}

\subsection{PIPULA}
\label{ap:pipula_analysis}
{\normalsize
To study the theoretical guarantees of PIPULA, we observe that PIPULA is equivalent to MYIPLA when $\gamma = \lambda$ and $g_1=0$. 
We recall that in the case $g_1 = 0$, $\nabla U^{\gamma}$ is Lipschitz in both variables with constant $L\leq \gamma^{-1}$ \citep{durmus_proximal}.
To obtain a similar result to Theorem~\ref{th:myipla_main_text} we introduce the following additional assumption.

\begin{assumptionnew}\label{assumption_22}
    We assume that there exists $\mu>0$ such that $\langle v-v', \prox_U^{\gamma}(v) - \prox_U^{\gamma}(v')\rangle\leq (1-\mu) \Vert v-v'\Vert^2$, for all $v,v'\in \mathbb{R}^{d_{\theta}}\times \mathbb{R}^{d_x}$.
\end{assumptionnew}

\textbf{B\ref{assumption_22}} implies that $\nabla U^{\gamma}$ is $\mu$-strongly convex, i.e. $\langle v-v', \nabla U^{\gamma}(v) - \nabla U^{\gamma}(v')\rangle \geq \mu \Vert v-v'\Vert^2$ for all $v,v'\in \mathbb{R}^{d_x}\times \mathbb{R}^{d_{\theta}}$.
In addition, since $U$ is a proper convex function we have that $U$ is twice differentiable almost everywhere (see the discussion below \textbf{A\ref{assumption_4}}).
Let $\Omega\subset\mathbb{R}^{d_{\theta}}\times \mathbb{R}^{d_x}$ denote the points where $U$ is twice differentiable, $\Bar{\theta}_{\star}$ be the maximiser of $k$  
and $\Tilde{\Omega} = \Omega\cap (\{\Bar{\theta}_{\star}\}\times \mathbb{R}^{d_x})$.

Using a similar strategy to that used to obtain the error bound in Theorem \ref{th:myipla_main_text}, we obtain the following result for PIPULA.

\begin{corollary}\label{th:pipula_theorem}
    Let \textbf{A\ref{assumption_1}}--\textbf{A\ref{assumption_1prime}}, \textbf{A\ref{assumption_4}} and
\textbf{B\ref{assumption_22}} hold. Then for every $\gamma_0\in(0, 2\mu^{-1})$ there exist constants $C_1>0$ independent of $t,n,N,\gamma,\lambda,d_{\theta}, d_x$ such that for every $\gamma\in(0,\gamma_0]$, one has
\begin{align*}
    \mathbb{E}[\Vert\theta_n^N-\Bar{\theta}_{\star}\Vert^2]^{1/2}\leq &\sqrt{\frac{d_{\theta}}{N\mu}} + \frac{\gamma}{\mu} \Big(\frac{\Vert g_2\Vert_{\textnormal{Lip}}^2}{2} A  +  B\Big)\\
    &+e^{-\mu n\gamma}\bigg(\mathbb{E}[\Vert Z_0^N-z_{\star}\Vert^2]^{1/2}+\Big(\frac{d_xN+d_{\theta}}{N\mu}\Big)^{1/2}\bigg) \\&+ C_1(1+\sqrt{d_{\theta}/N+d_x})\gamma^{1/2} +\mathcal{O}(\gamma^2)
\end{align*}
for all $n\in\mathbb{N}$, where $z_\star = (\theta_*, N^{-1/2}x_\star,\dots,N^{-1/2}x_\star)$ and $(\theta_\star, x_\star)$ is the minimiser of $U^{\gamma}$.
\end{corollary}
\begin{proof}
    Under
\textbf{A\ref{assumption_1}}--\textbf{A\ref{assumption_1prime}}, \textbf{A\ref{assumption_4}} and
\textbf{B\ref{assumption_22}} Propositions \ref{prop:theta_convergence_lambda}, \ref{prop:strong_solution}, \ref{prop:invariant_measure},
\ref{prop:concentration_bound} and Proposition 4 of \citet{akyildiz2023interacting} hold with $\lambda=\gamma$.
To obtain a bound on the discretisation error observe that under \textbf{B\ref{assumption_22}} $U^\gamma$ is strongly convex, since $\nabla U^\gamma$ is also Lipschitz continuous with constant $L\leq \gamma^{-1}$ \citep[Proposition 1]{durmus_proximal}, we have that $\nabla U^\gamma$ is co-coercive (see Theorem 1 in \cite{gao2018properties})
\begin{align}
\label{eq:cocoercive}
\langle \nabla U^\gamma(x)-\nabla U^\gamma(y),x-y \rangle &\geq \frac{1}{L}\Vert \nabla U^\gamma(x)-\nabla U^\gamma(y) \Vert^2\\
&\geq \gamma\Vert \nabla U^\gamma(x)-\nabla U^\gamma(y) \Vert^2,\notag
\end{align}
for every $x,y \in \mathbb{R}^d$. By plugging this result into the proof of \citet[Lemma B.1]{akyildiz2023interacting} we obtain an equivalent result to that of \citet[Proposition 5]{akyildiz2023interacting}:
\begin{align*}
    W_2(\mathcal{L}(\theta_{n}^N), \mathcal{L}(\boldsymbol{\theta}_{n\gamma}^N)) \leq C_1 (1+\sqrt{d_\theta/N+d_x})\gamma^{1/2},
\end{align*}
where $C_1>0$ is independent of $t,n,N,\gamma, d_\theta, d_x$ and $\gamma \in (0, \gamma_0)$ with $\gamma_0 \in (0, 2\mu^{-1})$.
\end{proof}}

\subsection{PIPGLA}
\label{app:pip_gla analysis}
{\normalsize
Recall that PIPGLA is given by the following scheme
\begin{align*}
    \theta_{n+1/2}^N &= \theta_{n}^N  - \frac{\gamma}{N}\sum_{j=1}^N \nabla_{\theta} g_1(\theta_n^N, X_n^{j, N}) +\sqrt{\frac{2\gamma}{N}}\xi_{n+1}^{0, N},\\
    X_{n+1/2}^{i, N} &= X_{n}^{i, N} -\gamma \nabla_{x} g_1(\theta_n^N, X_n^{i, N})+\sqrt{2\gamma}\;\xi_{n+1}^{i, N},\\
    \theta_{n+1}^N &= \frac{1}{N}\sum_{i=1}^N\prox_{g_2}^{\lambda} \bigg(\theta_{n+1/2}^N, X_{n+1/2}^{i, N}\bigg)_{\theta},\quad
    X_{n+1}^{i, N} = \prox_{g_2}^{\lambda} \big(\theta_{n+1/2}^N, X_{n+1/2}^{i, N}\big)_{x}.
\end{align*}
We want to prove a bound for $W_2(\mathcal{L}(\theta^N_n), \delta_{\Bar{\theta}_\star})$, where $\mathcal{L}(\theta^N_n)$ denotes the distribution of the random variable $\theta_n^N$. Applying the triangular inequality for Wasserstein distances
\begin{equation}\label{eq:wasserstein_distance_decomposition}
    W_2(\mathcal{L}({\theta}^N_n), \delta_{\Bar{\theta}_\star})\leq W_2(\delta_{\Bar{\theta}_{\star}}, \pi_{\Theta}^N) + W_2(\pi_{ \Theta}^N, \mathcal{L}({\theta_n^N})).
\end{equation}
The \emph{concentration} term $W_2(\delta_{\Bar{\theta}_{\star}}, \pi_{\Theta}^N)$ is analysed in Lemma \ref{lemma:aux_first_term_pipgla} and Theorem \ref{th:pipgla_appendix}. For the term $W_2(\pi_{ \Theta}^N, \mathcal{L}({\theta_n^N}))$, we derive a novel bound. The roadmap for obtaining this latter bound is as follows:
\begin{enumerate}
    \item We first analyse the PGLA updates targetting a distribution $\pi_\lambda \propto \exp(-(g_1+ g_2^\lambda))$ in $\mathbb{R}^d$ given by
\begin{align*}
    Z_{n} &= X_{n} -\gamma \nabla g_1(X_n),\\
    Y_{n} & = Z_{n} +\sqrt{2\gamma}\;\xi_{n},\\
    X_{n+1} &= \text{prox}_{g_2}^{\lambda} (Y_{n}).
\end{align*}
We provide an error on $W_2(\mathcal{L}({X_{n+1}}),\pi)$ where $\pi_\lambda \propto \exp(-(g_1+ g_2))$ for arbitrary $\lambda$ and $\gamma$. One key idea for this is the use of the minimal section which quantifies the least norm element in the subdifferential set $\partial g(x)$ introduced in Definition \ref{definition:minimal_section}.
\item  We show in Corollary \ref{remark:pipgla_rewrite_N} that the results established in 1. can be applied to a proximal gradient scheme in which the noise is scaled by $\sqrt{N}$
\begin{align*}
V_{n+1/2} &= V_{n} -\gamma\nabla g_1(V_n)+\sqrt{\frac{2\gamma}{N}}\;\xi_{n},\\
V_{n+1} &= \text{prox}_{g_2}^{\lambda} (V_{n+1/2})= V_{n+1/2}-\lambda \nabla g_2^\lambda(V_{n+1/2}).
\end{align*}

\item Taking $Z_n = (\theta_n^N, N^{-1/2}X_n^{1, N},\dots, N^{-1/2}X_n^{N, N})$, PIPGLA can be expressed as 
\begin{align*}
    Z_{n+1/2} &= Z_n -\gamma\nabla G_1(Z_n) + \sqrt{\frac{2\gamma}{N}} \xi_{n+1},\\
    Z_{n+1}&= \text{prox}_{G_2}^{\lambda} (Z_{n+1/2}) =Z_{n+1/2} -\lambda\nabla G_2^\lambda(Z_{n+1/2})\notag.
\end{align*}
where $G_1$ and $G_2$ are defined as 
\begin{align*}
G_1(z_\theta, z_1,\dots, z_N) = \frac{1}{N}\sum_{i=1}^N  g_1(z_\theta, \sqrt{N} z_i), \\
G_2^\lambda(z_\theta, z_1,\dots, z_N) = \frac{1}{N}\sum_{i=1}^N  g_2^\lambda(z_\theta, \sqrt{N} z_i).
\end{align*}
Since $G_1$ and $G_2^\lambda$ preserve the properties (strong convexity and Lipschitzness) of $g_1$ and $g_2^\lambda$ (see Corollary \ref{cor:pipgla_final}), we use the result in 2. to bound the error $W_2(\mathcal{L}({Z_n}), \pi^N)$. Finally, using a data processing inequality it follows $W_2(\mathcal{L}({\theta_n^N}), \pi_{\Theta}^N)\leq W_2(\mathcal{L}({Z_n}), \pi^N)$.
\end{enumerate}

We begin by presenting some results that will be useful for proving a bound for $W_2(\pi_{\Theta}^N, \mathcal{L}({\theta_n^N}))$.

\subsubsection{Error bound for proximal gradient Langevin algorithm}

We collect here a number of results adapted from \cite{NEURIPS2019_6a8018b3_salim} which show convergence of the proximal gradient Langevin algorithm (PGLA) introduced in \cite{NEURIPS2019_6a8018b3_salim} and recalled in Section~\ref{sec:pgla}.
In particular, we derive convergence of the splitting scheme for general $\lambda$ (which includes as a special case $\lambda=\gamma$) when both $\nabla g_1$ and $\prox_{g_2}^\lambda$ can be computed exactly (which is a special case of the result in \cite{NEURIPS2019_6a8018b3_salim} in the case $\lambda=\gamma$).

For convenience we consider the following decomposition of the PGLA update targeting a distribution $\muu \propto \exp(-(g_1+ g_2^\lambda))$ over $\real^d$
\begin{align*}
    Z_{n} &= X_{n} -\gamma \nabla g_1(X_n),\\
    Y_{n} & = Z_{n} +\sqrt{2\gamma}\;\xi_{n},\\
    X_{n+1} &= \prox_{g_2}^{\lambda} (Y_{n}).
\end{align*}
We are going to derive a bound for $W_2(\mathcal{L}({X_{n+1}}),\pi)$, where $\mathcal{L}({X_{n+1}})$ denotes the distribution of the random variable $X_{n+1}$. For every $\pi$-integrable function $g:\mathbb{R}^{d}\to \mathbb{R}$, we define $\mathcal{E}_g(\pi)=\int g\md \pi$ and we denote $\mathcal{F}= \mathcal{E}_{g_1+g_2} + \mathcal{H}$, where $\mathcal{H}$ is the negative entropy $\mathcal{H}(\pi)=\int \log\pi\;\md \pi$. We also introduce the subdifferential of a convex function and its minimal section, since we will use them for our proofs.
\begin{defin}\label{definition:minimal_section}
[Subdifferential and minimal section]
   \em For any convex function $g:\mathbb{R}^d\to\mathbb{R}$, its subdifferential evaluated at $x$ is the set
   \begin{equation*}
       \partial g(x) := \{ d\in\mathbb{R}^d \;|\; g(x)+ \langle d, y - x \rangle \leq g(y) \forall y\in\mathbb{R}^d\}.
   \end{equation*}
   Thanks to \citet[Proposition 16.4]{combettes_bauschke_2017}, we have that $\partial g(x)$ is a nonempty closed convex set. So the projection of 0 onto $\partial g(x)$, that is, the least norm element in the set  $\partial g(x)$, is well defined, and we refer to this element as $\nabla^0 g(x)$. Following \citet[Section 3.1]{NEURIPS2019_6a8018b3_salim}, we name the function $\nabla^0 g: \mathbb{R}^d\to\mathbb{R}$ the minimal section of $\partial g$.
   \em
   \label{def:subdiff}
\end{defin}

Following \cite{NEURIPS2019_6a8018b3_salim}, we derive our results under the following assumptions on $g_1, g_2$.

\begin{assumptiongla2}\label{assumption_1_gla2}
    We assume that $g_1\in\mathcal{C}^1$ is convex, gradient Lipschitz with constant $L_{g_1}$ and lower bounded, and $g_2$ is proper, convex, lower semi-continuous and lower bounded. 
\end{assumptiongla2}

\begin{assumptiongla2}\label{assumption_3_gla}
    $g_1$ is $\mu$-strongly convex.
\end{assumptiongla2}

\begin{assumptiongla2}\label{assumption_2_gla2}
    Assume that $\Vert\nabla^0 g_2(x)\Vert^2\leq C$ for every $x\in\mathbb{R}^d$.
\end{assumptiongla2}

In particular, \textbf{D\ref{assumption_1_gla2}} and \textbf{D\ref{assumption_3_gla}} are equivalent to \textbf{A\ref{assumption_1}} and \textbf{A\ref{assumption_3}} for the target $\pi_\lambda(\theta, x)\propto \exp(-U^\lambda(\theta, x))$. Similarly, \textbf{D\ref{assumption_2_gla2}} is equivalent to \textbf{C\ref{assumption_1_gla}}.

\begin{lemma}\label{lemma:pipgla_lemma_1}
    Let \textbf{D\ref{assumption_1_gla2}} and \textbf{D\ref{assumption_3_gla}} hold and assume $\gamma\leq 1/L_{g_1}$. Then, for all $n\in\mathbb{N}$ 
\begin{equation*}
    2\gamma\big[\mathcal{E}_{g_1}(\mathcal{L}({Z_n}))-\mathcal{E}_{g_1}(\pi)\big] \leq (1-\gamma\mu)W_2^2( \mathcal{L}({X_n}), \pi) - W_2^2( \mathcal{L}({Z_n}), \pi).
\end{equation*}
\end{lemma}
\begin{proof}
Let $a\in\mathbb{R}^{d}$, using that $g_1$ is $\mu$--strongly convex
\begin{align*}
\Vert Z_n - a\Vert^2 &= \Vert X_{n} -a\Vert^2 -2\gamma \langle \nabla g_1(X_n), X_n-a \rangle + \gamma^2\Vert \nabla g_1(X_n)\Vert^2\\
&\leq \Vert X_{n} -a\Vert^2 +2\gamma \big( g_1(a)- g_1(X_n)-\frac{\mu}{2} \Vert X_n-a \Vert^2\big) + \gamma^2\Vert \nabla g_1(X_n)\Vert^2\\
&=(1-\gamma\mu)\Vert X_{n} -a\Vert^2 +2\gamma \big( g_1(a)- g_1(X_n)\big) + \gamma^2\Vert \nabla g_1(X_n)\Vert^2.\numberthis\label{eq:lemma_step}
\end{align*}
Since $g_1$ is gradient Lipschitz with constant $L_{g_1}$ and $Z_n-X_n = -\gamma\nabla g_1(X_n)$
\begin{align*}
    g_1(Z_n)&\leq g_1(X_n) + \langle \nabla g_1(X_n), Z_n-X_n\rangle + \frac{L_{g_1}}{2}\Vert Z_n -X_n\Vert^2\\
    &=g_1(X_n) -\gamma\Big( 1-\frac{\gamma L_{g_1}}{2}\Big)\Vert\nabla g_1(X_n)\Vert^2 \\
    &\leq g_1(X_n) -\frac{\gamma}{2}\Vert\nabla g_1(X_n)\Vert^2,
\end{align*}
where in the last inequality we have used that $\gamma\leq 1/L_{g_1}$. Reordering terms gives the following upper bound
\begin{equation*}
\gamma^2\Vert\nabla g_1(X_n)\Vert^2 \leq 2\gamma\big(g_1(X_n)-g_1(Z_n)\big).
\end{equation*}
Plugging this into \eqref{eq:lemma_step} we have
\begin{equation*}
\Vert Z_n - a\Vert^2\leq(1-\gamma\mu)\Vert X_{n} -a\Vert^2 +2\gamma \big( g_1(a)- g_1(Z_n)\big).
\end{equation*}
It is important to note, the above inequality is true for any $a$, $X_n$, and $Z_n$ where $Z_n = X_n - \gamma \nabla g_1(X_n)$ (as deterministic vectors with appropriate dimension). Now, let $(a, X_n) \sim \nu(\md a, \md x_n)$ with marginal $\nu^a(\md a) = \pi(\md a)$. Taking conditional expectation w.r.t. $Z_n$ given $\sigma(a, X_n)$, we obtain
\begin{equation*}
    \mathbb{E}\big[\Vert Z_n - a\Vert^2\big | \sigma(a, X_n) ]\leq(1-\gamma\mu) \| X_{n} -a\|^2 +2\gamma \big( g_1(a) - \bE \big[ g_1(Z_n) | \sigma(a, X_n)\big]\big).
\end{equation*}
By taking the unconditional expectation (i.e. w.r.t. $\nu$), we get
\begin{equation*}
    \mathbb{E}\big[\Vert Z_n - a\Vert^2\big ]\leq(1-\gamma\mu) \bE_\nu\big[\Vert X_{n} -a\Vert^2\big] +2\gamma \big( \mathcal{E}_{g_1}(\pi)- \mathcal{E}_{g_1}(\mathcal{L}({Z_n}))\big).
\end{equation*}
By the definition of the Wasserstein distance we get
\begin{equation*}
    W_2^2(\mathcal{L}({Z_n}),\pi)\leq(1-\gamma\mu)\mathbb{E}_\nu\big[\Vert X_{n} -a\Vert^2\big] +2\gamma \big( \mathcal{E}_{g_1}(\pi)- \mathcal{E}_{g_1}(\mathcal{L}({Z_n}))\big).
\end{equation*}
Note that the last inequality is true for all $\nu$ with prescribed marginal above. In particular, we can take the infimum over all such couplings and inequality would still hold for the infimum. This leads to
\begin{equation*}
    W_2^2(\mathcal{L}({Z_n}),\pi)\leq(1-\gamma\mu)W_2^2(\mathcal{L}({X_{n}}), \pi) +2\gamma \big( \mathcal{E}_{g_1}(\pi)- \mathcal{E}_{g_1}(\mathcal{L}({Z_n}))\big),
\end{equation*}
which is the desired result.
\end{proof}

\begin{lemma}\label{lemma:bound_second_step}
Let $g:\mathbb{R}^d\to\mathbb{R}$ be a convex function and $g^\lambda$ its $\lambda$-Moreau-Yosida approximation. Consider $a, y_0, y_1\in\mathbb{R}^d$ such that $y_1=\prox_g^\lambda(y_0)$. Then,
\begin{equation*}
\Vert y_1 - a\Vert^2 \leq \Vert y_0 - a\Vert^2 - 2\lambda \left(g(y_0)- g(a)\right)+ \lambda^2\Vert \nabla g^0 (y_0)\Vert^2.
\end{equation*}
\end{lemma}
\begin{proof}
Recalling that $\prox_g^\lambda(y_0) = y_0 -\lambda\nabla g^\lambda(y_0)$ we have
\begin{equation}
\label{eq:part1}
\Vert y_1-a\Vert^2 = \Vert y_0-a\Vert^2 - 2\lambda\langle \nabla g^\lambda(y_0), y_0 -a\rangle + \lambda^2\Vert \nabla g^\lambda(y_0)\Vert^2.
\end{equation}
Using that $y_1 = y_0 -\lambda\nabla g^\lambda(y_0)$ we can write
\begin{equation*}
\langle \nabla g^\lambda(y_0), y_0 -a\rangle = \langle \nabla g^\lambda(y_0), y_1 -a\rangle+\lambda\norm{\nabla g^\lambda(y_0)}^2.
\end{equation*}
Since $\nabla g^\lambda(y_0)$ belongs to the subdifferential of $g(y_1)$, i.e. $\nabla g^\lambda(x)\in \partial g(\text{prox}_g^\lambda(x))$  \citep[Proposition 16.44]{combettes_bauschke_2017}, we further have that
\begin{equation*}
    \langle \nabla g^\lambda(y_0), y_1 -a\rangle \geq g(y_1)-g(a),
\end{equation*}
from which we obtain
\begin{equation*}
-2\lambda\langle \nabla g^\lambda(y_0), y_0 -a\rangle \leq  -2\lambda\left(g(y_1)-g(a)+\lambda\norm{\nabla g^\lambda(y_0)}^2\right).
\end{equation*}
Recalling the definition of Moreau-Yosida approximation in Definition~\ref{def:my} we have that $g^\lambda(y_0)=g(y_1) + \Vert y_0-y_1\Vert^2/(2\lambda)$; plugging this into the equation above gives
\begin{align}
\label{eq:part2}
    -2\lambda\langle \nabla g^\lambda(y_0), y_0 -a\rangle &\leq  -2\lambda\left(g^\lambda(y_0)-g(a)\right)-2\lambda^2\norm{\nabla g^\lambda(y_0)}^2+\norm{y_1-y_0}^2\\
    &=-2\lambda(g^\lambda(y_0)-g(a))-\lambda^2\norm{\nabla g^\lambda(y_0)}^2.\notag
\end{align}
Finally, using \citet[Lemma 9]{NEURIPS2019_6a8018b3_salim} which states that $g^\lambda(x)\geq g(x)-\lambda\Vert\nabla^0 g(x)\Vert/2$, where $\nabla^0 g$ is the minimal section introduced in Definition \ref{def:subdiff}, and combining ~\eqref{eq:part1} and~\eqref{eq:part2} we have
\begin{align*}
  \Vert y_1-a\Vert^2 &\leq \Vert y_0-a\Vert^2 -2\lambda\left(g^\lambda(y_0)-g(a)\right)\leq \Vert y_0-a\Vert^2 -2\lambda\left(g(y_0)-g(a)\right) + \lambda^2 \Vert \nabla^0 g(x)\Vert^2.
\end{align*}
\end{proof}

\begin{lemma}\label{lemma:pipgla_lemma_2}
Let \textbf{D\ref{assumption_1_gla2}}--\textbf{D\ref{assumption_2_gla2}} hold.
Then,
\begin{equation*}
2\lambda\big[ \mathcal{E}_{g_2}(\mathcal{L}({Y_n}))- \mathcal{E}_{g_2}(\pi)\big]\leq W_2^2(\mathcal{L}({Y_n}),\pi) - W_2^2(\mathcal{L}({X_{n+1}}),\pi) + \lambda^2 C 
\end{equation*}
\end{lemma}
\begin{proof}
Applying Lemma \ref{lemma:bound_second_step} with $y_0= Y_n$, $y_1= X_{n+1}$ and $g= g_2$, we have
\begin{align*}
   \Vert X_{n+1}-a\Vert^2 &\leq  \Vert Y_n-a\Vert^2 - 2\lambda \big(g_2(Y_n)- g_2(a) \big) +\lambda^2\Vert\nabla^0 g_2(Y_n)\Vert^2.
\end{align*}    
Now, let $a$ be a random vector sampled from the distribution with density $\pi$. Taking expectations in the previous expression and using the definition of the Wasserstein distance we obtain
\begin{align*}
    W_2^2(\mathcal{L}({X_{n+1}}), \pi) &\leq  \mathbb{E}\big[\Vert Y_n-a\Vert^2\big]- 2\lambda\big(\mathcal{E}_{g_2}(\mathcal{L}({Y_n}))- \mathcal{E}_{g_2}(\pi) \big) +\lambda^2\mathbb{E}[\Vert\nabla^0 g_2(Y_n)\Vert^2]\\
    &\leq  \mathbb{E}\big[\Vert Y_n-a\Vert^2\big]- 2\lambda\big(\mathcal{E}_{g_2}(\mathcal{L}({Y_n}))- \mathcal{E}_{g_2}(\pi) \big) +\lambda^2 C.
\end{align*}
Finally, taking the infimum over all couplings $Y_n, a$ of $\mathcal{L}({Y_n}),\pi$, it follows that
\begin{equation*}
    W_2^2(\mathcal{L}({X_{n+1}}), \pi) \leq  W_2^2(\mathcal{L}({Y_n}),\pi)- 2\lambda\big(\mathcal{E}_{g_2}(\mathcal{L}({Y_n}))- \mathcal{E}_{g_2}(\pi) \big) +\lambda^2 C.
\end{equation*}
\end{proof}

\begin{theorem}\label{th:main_theorem_pipgla}
    Let \textbf{D\ref{assumption_1_gla2}}--\textbf{D\ref{assumption_2_gla2}} hold and assume that $\gamma\leq 1/L_{g_1}$.
    Then, for all $n\in\mathbb{N}$
    \begin{align*}
2\gamma\kl(\mathcal{L}({Y_n})\;|\;\pi) \leq& (1-\gamma\mu) W_2^2(\mathcal{L}({X_n}), \pi)- \frac{\gamma}{\lambda}W_2^2(\mathcal{L}({X_{n+1}}), \pi) \\
&- \left(1-\frac{\gamma}{\lambda}\right)W_2^2(\mathcal{L}({Y_{n}}), \pi) + \gamma(2 \gamma L_{g_1} d +  \lambda C) .
\end{align*}
\end{theorem}
\begin{proof}
Since $g_1+g_2$ is convex by assumption the following holds $\pi\in \mathcal{P}_2(\mathbb{R}^d)$, $\mathcal{H}(\pi)<\infty$, $\mathcal{E}_{g_1+g_2}(\pi)<\infty$ and for all $\mu\in\mathcal{P}_2(\mathbb{R}^d)$ satisfying $\mathcal{E}_{g_1+g_2}(\mu)<\infty$,
\begin{equation*}
    \kl(\mu\;|\;\pi) = 
    \mathcal{E}_{g_1 + g_2}(\mu)+\mathcal{H}(\mu) -(\mathcal{E}_{g_1 + g_2}(\pi)+\mathcal{H}(\pi)) = \mathcal{F}(\mu)-\mathcal{F}(\pi).
\end{equation*}
We can further decompose $\mathcal{E}_{g_1 + g_2}(\mu) = \mathcal{E}_{g_1}(\mu) + \mathcal{E}_{g_2}(\mu)$.
Using \citet[Lemma 5]{durmus2019analysis} we have that the negative entropy satisfies the following inequality
\begin{equation}\label{eq:entropy_result}
    2\gamma\big[\mathcal{H}(\mathcal{L}({Y_n}))-\mathcal{H}(\pi)\big]\leq W_2^2(\mathcal{L}({Z_n}),\pi)-W_2^2(\mathcal{L}({Y_{n}}), \pi). 
\end{equation}
Since $g_1$ is $L_{g_1}$-gradient Lipschitz and strongly convex, it follows that
\begin{equation*}
    0\leq g_1(Y_n) - g_1(Z_n) + \langle\nabla g_1(Z_n), Z_n-Y_n\rangle\leq \frac{L_{g_1}}{2}\Vert Y_n-Z_n\Vert^2.
\end{equation*}
Note that $Y_n-Z_n = \sqrt{2\gamma}\xi_n$ is independent of $Z_n$, $\mathbb{E}[Y_n-Z_n]=0$ and $\mathbb{E}[\Vert Y_n-Z_n\Vert^2]=2\gamma d$, where $d$ is the dimension of the standard Gaussian random variable $\xi_n$. Therefore, taking expectations in the previous inequality we get
\begin{equation}\label{eq:potential_result_1}
2\gamma\big[\mathcal{E}_{g_1}(\mathcal{L}({Y_n})) - \mathcal{E}_{g_1}(\mathcal{L}({Z_n}))\big] \leq {2\gamma^2 L_{g_1} d}.
\end{equation}
On the other hand, by Lemmas \ref{lemma:pipgla_lemma_1} and \ref{lemma:pipgla_lemma_2} we have
\begin{align}
2\gamma\big[\mathcal{E}_{g_1}(\mathcal{L}({Z_n}))-\mathcal{E}_{g_1}(\pi)\big] &\leq (1-\gamma\mu)W_2^2( \mathcal{L}({X_n}), \pi) - W_2^2( \mathcal{L}({Z_n}), \pi),\label{eq:potential_result_2}\\    
2\gamma\big[ \mathcal{E}_{g_2}(\mathcal{L}({Y_n}))- \mathcal{E}_{g_2}(\pi)\big]&\leq \frac{\gamma}{\lambda} W_2^2(\mathcal{L}({Y_n}),\pi) - \frac{\gamma}{\lambda} W_2^2(\mathcal{L}({X_{n+1}}),\pi) + \gamma\lambda C. \label{eq:potential_result_3}
\end{align}
Summing up \eqref{eq:entropy_result}-\eqref{eq:potential_result_3} and using that $\kl(\mathcal{L}({Y_n})\;|\;\pi)= \mathcal{F}(\mathcal{L}({Y_n}))-\mathcal{F}(\pi)$ we have the desired result.

\end{proof}

\begin{corollary}\label{col:pipgla}
     Let \textbf{D\ref{assumption_1_gla2}}--\textbf{D\ref{assumption_2_gla2}} hold and assume that $\gamma\leq 1/L_{g_1}$ and $\gamma\leq \lambda\leq\gamma/(1-\mu\gamma)$. Then
    \begin{equation*}
    W_2^2(\mathcal{L}({X_n}),\pi)\leq \frac{\lambda^n(1-\gamma\mu)^n}{\gamma^n}W_2^2(\mathcal{L}({X_0}),\pi) + \frac{\lambda(2 \gamma L_{g_1} d + \lambda C)}{1-\lambda(1-\mu\gamma)/\gamma}. 
\end{equation*}
\end{corollary}
\begin{proof}
Since the KL divergence and the Wasserstein distance are always non-negative and we assume that $\gamma\leq \lambda$, we have by Theorem \ref{th:main_theorem_pipgla} that for all $n\in\mathbb{N}$
\begin{equation*}
W_2^2(\mathcal{L}({X_{n+1}}), \pi) \leq \frac{\lambda(1-\gamma\mu)}{\gamma} W_2^2(\mathcal{L}({X_n}), \pi) + \lambda(2 \gamma L_{g_1} d +  \lambda C).
\end{equation*}
Unrolling this recurrence we get 
\begin{align*}
 W_2^2(\mathcal{L}({X_{n}}), \pi) &\leq \frac{\lambda^n(1-\gamma\mu)^n}{\gamma^n} W_2^2(\mathcal{L}({X_0}), \pi) + \lambda(2 \gamma L_{g_1} d + \lambda C)\sum_{i=0}^{n-1}\frac{\lambda^i(1-\gamma\mu)^i}{\gamma^i}\\
 &\leq \frac{\lambda^n(1-\gamma\mu)^n}{\gamma^n} W_2^2(\mathcal{L}({X_0}), \pi) + \frac{\lambda(2 \gamma L_{g_1} d + \lambda C)}{1-\lambda(1-\gamma\mu)/\gamma},
\end{align*}
where we have used the assumption $\lambda\leq \gamma/(1-\mu\gamma)$.
\end{proof}

\subsubsection{Convergence and discretisation bounds}
We start by showing that the results established above can be applied to a proximal gradient scheme in which the noise is scaled by $\sqrt{N}$
\begin{align*}
V_{n+1/2} &= V_{n} -\gamma\nabla g_1(V_n)+\sqrt{\frac{2\gamma}{N}}\;\xi_{n},\numberthis\label{eq:simplified_pipgla_2}\\
V_{n+1} &= V_{n+1/2}-\lambda \nabla g_2^\lambda(V_{n+1/2}).
\end{align*}

\begin{corollary}[Rescaled noise]
\label{remark:pipgla_rewrite_N}
Let \textbf{D\ref{assumption_1_gla2}}--\textbf{D\ref{assumption_2_gla2}} hold and assume that $\gamma\leq 1/L_{g_1}$ and $\gamma\leq \lambda\leq\gamma/(1-\mu\gamma)$. Then,
\begin{equation*}
    W_2^2(\mathcal{L}({V_n}),\pi^N)\leq \frac{\lambda^n(1-\gamma\mu)^n}{\gamma^n}W_2^2(\mathcal{L}({V_0}),\pi^N) + \frac{\lambda(2 \gamma L_{g_1} d + \lambda N C)}{N\left(1-\lambda(1-\mu\gamma)/\gamma\right)}.
\end{equation*}
\em
\end{corollary}
\begin{proof}
Let $\Tilde{g}_1=Ng_1$ and $\Tilde{g}_2=Ng_2$.
It is easy to check that $\Tilde{g}_1$ is $(NL_{g_1})$--gradient Lipschitz and $(N\mu)$--strongly convex. 
In addition, we have that
\begin{align*}
    \prox_{g_2}^\lambda(x) = \argmin_{z\in\mathbb{R}^d} \frac{\Tilde{g}_2(x)}{N} + \frac{\Vert x-z\Vert^2}{2\lambda} =  \argmin_{z\in\mathbb{R}^d} \frac{1}{N}\left(\Tilde{g}_2(x) + \frac{\Vert x-z\Vert^2}{2\lambda/N}\right) = \prox_{\Tilde{g}_2}^{\lambda/N}(x),
\end{align*}
since the $\argmin$ does not change if the function is multiplied by a constant, which results in $\nabla g_2^\lambda=\nabla \Tilde{g}_2^{\lambda/N}/N$.
Thus, \eqref{eq:simplified_pipgla_2} can be rewritten as
\begin{align*}
V_{n+1/2} &= V_{n} -\Tilde{\gamma}\nabla \Tilde{g}_1(V_n)+\sqrt{2\Tilde{\gamma}}\;\xi_{n},\\
V_{n+1} &= V_{n+1/2}-\Tilde{\lambda} \nabla \Tilde{g}_2^{\Tilde{\lambda}}(V_{n+1/2}),
\end{align*}
where $\Tilde{\gamma}=\gamma/N$ and $\Tilde{\lambda}=\lambda/N$. Note that the subdifferential set satisfies $\partial \tilde{g}_2 = N \partial g_2$. Therefore, since $\Vert \nabla^0 g_2 (x)\Vert^2\leq C$ for all $x\in\mathbb{R}^d$ by \textbf{D\ref{assumption_2_gla2}}, it follows that $\Vert \nabla^0 \Tilde{g}_2(x)\Vert^2\leq N^2C$. 
Therefore, taking $\Tilde{\gamma}\leq 1/(NL_{g_1})$ which is equivalent to $\gamma\leq 1/L_{g_1}$, and applying Corollary \ref{col:pipgla} the result follows.
\end{proof}

In order to be able to use the bound obtained in Corollary \ref{remark:pipgla_rewrite_N}, we rewrite PIPGLA as the algorithm given in \eqref{eq:simplified_pipgla_2}.
To do so, define 
\begin{align*}
G_1(z_\theta, z_1,\dots, z_N) = \frac{1}{N}\sum_{i=1}^N  g_1(z_\theta, \sqrt{N} z_i), \\
G_2^\lambda(z_\theta, z_1,\dots, z_N) = \frac{1}{N}\sum_{i=1}^N  g_2^\lambda(z_\theta, \sqrt{N} z_i).
\end{align*}
Note that the gradients of these functions are given by
\begin{equation*}
    \nabla G_1(z_\theta, z_1,\dots, z_N)  = \left(N^{-1}\sum_{i=1}^N\nabla_\theta g_1(z_\theta, \sqrt{N}z_i),N^{-1/2}\nabla_{z_1} g_1(z_{\theta},\sqrt{N}z_1),\dots,N^{-1/2}\nabla_{z_N} g_1(z_{\theta}, \sqrt{N}z_N)\right)^{\intercal}
\end{equation*}
and similarly for $G_2^\lambda$.

Taking $Z_n = (\theta_n^N, N^{-1/2}X_n^{1, N},\dots, N^{-1/2}X_n^{N, N})$, PIPGLA can be expressed as 
\begin{align}\label{eq:intermediate_pipgla}
    Z_{n+1/2} &= Z_n -\gamma\nabla G_1(Z_n) + \sqrt{\frac{2\gamma}{N}} \xi_{n+1},\\
    Z_{n+1}&= Z_{n+1/2} -\lambda\nabla G_2^\lambda(Z_{n+1/2})\notag.
\end{align}

\begin{corollary}\label{cor:pipgla_final}
Let $Z_n=(\theta_n^N, N^{-1/2}X_n^{1,N}, \dots, N^{-1/2}X_n^{N,N})$ and $\pi^N\propto\exp(-N(G_1+G_2))$.
    Suppose that \textbf{A\ref{assumption_1}}, \textbf{A\ref{assumption_3}} and \textbf{C\ref{assumption_1_gla}} hold true and assume $\gamma\leq 1/L_{g_1}$ and $\gamma\leq \lambda\leq\gamma/(1-\mu\gamma)$. Then,
\begin{equation*}
    W_2^2(\mathcal{L}({Z_n}), \pi^N)\leq \frac{\lambda^n(1-\gamma\mu)^n}{\gamma^n}W_2^2(\mathcal{L}({Z_0}), \pi^N) + \frac{\lambda(2\gamma L_{g_1}(d_\theta +N d_x)+\lambda N C)}{N\left(1-\lambda(1-\mu\gamma)/\gamma\right)}.
\end{equation*}
\end{corollary}
\begin{proof}

Note that if \textbf{C\ref{assumption_1_gla}} holds then $\Vert\nabla^0 G_2(z)\Vert^2\leq C$ for every $z$. To see this note that
\begin{equation*}
     \partial G_2(z) = \partial G_2(z_\theta, z_1,\dots, z_N) = \frac{1}{N}\sum_{i=1}^N\partial g_2(z_\theta, \sqrt{N} z_i).
\end{equation*}
Therefore, using the fact that $(N^{-1}\sum_i a_i)^2\leq N^{-1}\sum_i a_i^2$, we get that the minimal section satisfies
\begin{equation*}
    \Vert\nabla^0 G_2(z)\Vert^2 \leq \frac{1}{N}\sum_{i=1}^N\left\Vert\nabla^0 g_2(z_\theta, \sqrt{N}z_i)\right\Vert^2\leq C.
\end{equation*}
In addition, observe that \textbf{A\ref{assumption_1}}, \textbf{A\ref{assumption_3}}  imply that $G_1, G_2$ and $ G_2^\lambda$ are convex since $g_1, g_2$ and $g_2^\lambda$ are convex, and $G_1$ is also $\mu$-strongly convex and $L_{g_1}$-gradient Lipschitz. The proof then follows from Corollary \ref{remark:pipgla_rewrite_N}.
\end{proof}

Before proving our final result for $W_2(\mathcal{L}(\theta_n^N), \delta_{\bar{\theta}_\star})$, we provide a result adapted from \citet[Lemma A.8]{sinho_lemma_A8} to our non-differentiable setting that will be useful to bound the first term of \eqref{eq:wasserstein_distance_decomposition}.
\begin{lemma}\label{lemma:aux_first_term_pipgla}
    Suppose that the distribution $\pi\propto \exp(-f)$ on $\mathbb{R}^d$ is $\alpha$-strongly log-concave, almost everywhere differentiable and that $x^\star$ is the minimiser of $f$. Then,
    \begin{equation*}
        \mathbb{E}_{X\sim\pi}[\Vert X-x^\star\Vert^2]\leq d/\alpha.
    \end{equation*}
\end{lemma}
\begin{proof}
Let $\Omega\subset \mathbb{R}^d$ denote the set of differentiable points of $f$, note that when $f$ is convex and differentiable at $x\in\Omega$, then $\partial f(x) = \{\nabla f(x)\}$, that is, its gradient is its only subgradient.
Recall also that $f$ is strongly convex, so for every $x, y\in\mathbb{R}^d$ we have that
    \begin{equation*}
        \langle\partial f(x), x-y\rangle \geq \alpha \Vert x-y\Vert^2.
    \end{equation*}
     Integration by parts shows that for any smooth function $\phi : \mathbb{R}^d \to \mathbb{R}$ of controlled growth, it holds that 
    \begin{equation}\label{eq:aux_lemma_a8}
        0 = \int_{\Omega}\left(\Delta\phi - \langle\nabla f, \nabla\phi\rangle\right)\md \pi 
 = \mathbb{E}_{X\sim\pi}[\Delta\phi-\langle\nabla f, \nabla\phi\rangle].
    \end{equation}
    Applying \eqref{eq:aux_lemma_a8} to the function $\phi(x) := \Vert x-x^\star\Vert^2/2$, for which $\nabla\phi(x) = x-x^\star$ and $\Delta\phi = d$, together with the strong convexity of $f$, the result follows. 
\end{proof}

To conclude we present the following theorem that provides a convergence bound for PIPGLA in terms of $N,\gamma, n, \lambda$ and the convexity properties of $U$.
\begin{theorem}\label{th:pipgla_appendix}[Theorem~\ref{th:pipgla_main_text} restated]
    Let \textbf{A\ref{assumption_1}}, \textbf{A\ref{assumption_2}}, \textbf{A\ref{assumption_3}} and \textbf{C\ref{assumption_1_gla}}
 hold. Then for $\gamma\leq1/L_{g_1}$ and $\gamma\leq \lambda\leq\gamma/(1-\mu\gamma)$,
  PIPGLA satisfies 
\begin{align*}
    W_2(\mathcal{L}({\theta^N_n}), \delta_{\Bar{\theta}_\star})\leq& \sqrt{\frac{d_{\theta}}{N\mu}}  +\frac{\lambda^{n/2}(1-\gamma\mu)^{n/2}}{\gamma^{n/2}} W_2(\mathcal{L}({Z_0^N}), \pi^N) 
    +\left(\frac{\lambda(2\gamma L_{g_1}(d_\theta +N d_x)+\lambda N C)}{N\left(1-\lambda(1-\mu\gamma)/\gamma\right)}\right)^{1/2}
\end{align*} 
for all $n\in\mathbb{N}$, with $Z_0^N$ given in \textbf{A\ref{assumption_2}} and $C>0$ given in \textbf{C\ref{assumption_1_gla}} and independent of $t,n,N,\gamma,d_{\theta}, d_x$.
\end{theorem}
\begin{proof}
Using a form of the Prékopa-Leindler inequality for strong convexity \citep[Theorem 3.8]{saumard2014log}, $\pi_{\Theta}$ is $\mu$-strongly log-concave. Therefore, $\pi_{\Theta}^N$ is $N\mu$-strongly log-concave and satisfies all the assumptions of Lemma \ref{lemma:aux_first_term_pipgla}. So, we have that
\begin{equation*}
    W_2(\delta_{\bar{\theta}_*}, \pi_{\Theta}^N)^2\leq \frac{d_\theta}{N\mu}.
\end{equation*}

On the other hand, note that $\pi^N(z)\propto \exp(-N(G_1(z)+G_2(z))) = \exp(-\sum_i U(z_\theta, \sqrt{N}z_i))$.
By Corollary \ref{cor:pipgla_final} it follows that
\begin{equation*}
    W_2(\mathcal{L}({\theta_n^N}), \pi_{\Theta}^N)\leq W_2(\mathcal{L}({Z_n}), \pi^N)\leq \sqrt{ \frac{\lambda^n(1-\gamma\mu)^n}{\gamma^n}W_2^2(\mathcal{L}({Z_0^N}), \pi^N) + \frac{\lambda(2\gamma L_{g_1}(d_\theta +N d_x)+\lambda N C)}{N\left(1-\lambda(1-\mu\gamma)/\gamma\right)}}.
\end{equation*}
Using that $\sqrt{x+y}\leq\sqrt{x} + \sqrt{y}$, we have
\begin{equation*}
    W_2(\mathcal{L}({\theta_n^N}), \pi_{\Theta}^N)\leq  \frac{\lambda^{n/2}(1-\gamma\mu)^{n/2}}{\gamma^{n/2}}W_2(\mathcal{L}({Z_0^N}), \pi^N) + \left(\frac{\lambda(2\gamma L_{g_1}(d_\theta +N d_x)+\lambda N C)}{N\left(1-\lambda(1-\mu\gamma)/\gamma\right)}\right)^{1/2}.
\end{equation*}
    The proof then follows from \eqref{eq:wasserstein_distance_decomposition} and the above.
\end{proof}}

\section{Theoretical Analysis of Proximal Particle Gradient Descent}\label{app:proximal_particle_gradient_descent_methods}
\subsection{Background on Particle Gradient Descent}\label{app:intro_pgd}
{\normalsize
The PGD algorithm \citep{pmlr-v206-kuntz23a} relies on the perspective that the MMLE problem can be solved by minimising the free energy 
\begin{equation}\label{eq:free-energy}
    F(\theta, q)=\int\log\big(q(x)\big) q(x) \md x +\int U(\theta, x)q(x)\md x
\end{equation}
for all $(\theta, q)\in\Theta\times\mathcal{P}(\mathbb{R}^{d_x})$, where $\Theta$ denotes the parameter space 
and $U(\theta, x)\coloneqq -\log p_{\theta}(x, y)$. 
\citet{pmlr-v206-kuntz23a} propose a discretisation of a gradient flow associated with \eqref{eq:free-energy}, where they endow $\Theta$ with the Euclidean geometry and $\mathcal{P}(\mathbb{R}^{d_x})$ with the 2-Wasserstein one to take gradients.
This leads to the Euclidean-Wasserstein gradient flow of $F$
\begin{align}\label{eq:gradient_flow_1}
    \bm{\Dot{\theta}}_t &= -\nabla_{\theta} F(\bm{\theta}_t, q_t) = -\int \nabla_{\theta} U(\bm{\theta}_t, x) q_t(x) \md x,\\
    \Dot{q}_t &= -\nabla_{q} F(\bm{\theta}_t, q_t) = \nabla_x\cdot\Big[q_t\nabla_x \log\Big(\frac{q_t}{p_{\bm{\theta}_t}(\cdot,y)}\Big)\Big].\notag
\end{align}
\citet{pmlr-v206-kuntz23a} prove that the gradient $\nabla F(\theta, q)$ vanishes if and only if $\theta$ is a stationary point of $p_{\theta}(y)$ and $q$ is its corresponding posterior. 
Based on the observation that \eqref{eq:gradient_flow_1} is a Fokker-Planck equation satisfied by the law of a McKean-Vlasov SDE, and using a finite number of particles $(X_t^{i,N})_{i=1}^N$ to estimate $q_t$, they obtain the following approximation, for $t \geq 0$,
\begin{align}
\label{eq:pgd}
    \md\bm{\theta}_t^N &= -\frac{1}{N}\sum_{i=1}^N \nabla_{\theta} U(\bm{\theta}_t^N, \mathbf X_{t}^{i, N}) \md t,\\
    \md \mathbf X_{t}^{i,N}&=-\nabla_{x} U(\bm{\theta}_t^N, \mathbf X_t^{i, N})\md t +\sqrt{2} \md \mathbf B_t^{i, N}, \quad \quad i=1,\dots,N,\notag
\end{align}
where $(\mathbf B_t^{i, N})_{t \geq 0}$ for $i = 0, \ldots, N$ are $d_x$-dimensional Brownian motions. Using a simple Euler--Maruyama discretisation with step size $\gamma > 0$ of~\eqref{eq:pgd} one obtains the particle gradient descent (PGD) algorithm \citep{pmlr-v206-kuntz23a}
\begin{align*}
    \theta_{n+1} &= \theta_n - \frac{\gamma}{N}\sum_{j=1}^N \nabla_{\theta} U(\theta_n, X_n^{j, N}),\\
    X_{n+1}^{i, N} &= X_n^{i, N} - \gamma \nabla_x U(\theta_n, X_n^{i, N}) + \sqrt{2\gamma} \xi_{n+1}^{i, N}, \quad \quad i = 1, \ldots, N,
\end{align*}
where $(\xi_{n})$ for $n\geq 0$ are $d_x$-dimensional i.i.d. standard Gaussians.}

\subsection{Proximal Particle Gradient Descent}
{\normalsize
Similar to the approach we have taken in the main text, we can also provide a proximal version of the PGD algorithm. As mentioned in the main text, if we remove the noise term in the dynamics of $\theta$, we obtain
\begin{align}
        \md \bm{\theta}_{t}^N&=-\frac{1}{N}\sum_{i=1}^N \nabla_{\theta} U^{\lambda}(\bm{\theta}_{t}^N, \mathbf{X}_{t}^{i, N})\md t,\label{eq:app:particles_proximal_11}\\
        \label{eq:app:particles_proximal_22}
    \md \mathbf{X}_{t}^{i,N}&=-\nabla_{x} U^{\lambda}(\bm{\theta}_{t}^N, \mathbf{X}_{t}^{i, N})\md t +\sqrt{2} \md \mathbf B_t^{i, N}.
\end{align}
We can then provide an algorithm which is a discretisation of \eqref{eq:app:particles_proximal_11}-\eqref{eq:app:particles_proximal_22}, termed Moreau-Yosida Particle Gradient Descent (MYPGD), analogous to MYIPLA. The algorithm is given in Algorithm \ref{alg:MYPGD}.

\begin{algorithm}[t]
\caption{Moreau-Yosida Particle Gradient Descent (MYPGD)}\label{alg:MYPGD}
\begin{algorithmic}
\Require $N, K, \lambda, \gamma, \pi_{\text{init}}\in\mathcal{P}(\mathbb{R}^{d_{\theta}})\times \mathcal{P}((\mathbb{R}^{d_{x}})^N)$
\State Draw $(\theta_0, \{X_0^{i, N}\}_{i=1}^N)$ from $\pi_{\text{init}}$
\For{$n=0:K$} 
 \begin{align*}
    \theta_{n+1}^N &= \Big(1-\frac{\gamma}{\lambda}\Big)\theta_{n}^N  + \frac{\gamma}{N}\sum_{i=1}^N \Big(-\nabla_{\theta} g_1(\theta_n^N, X_n^{i, N}) + \frac{1}{\lambda} \prox_{g_2}^{\lambda}(\theta_n^N, X_n^{i, N})_{\theta}\Big)\\
    X_{n+1}^{i, N} &= \Big(1-\frac{\gamma}{\lambda}\Big)X_{n}^{i, N} -\gamma \nabla_{X} g_1(\theta_n^N, X_n^{i, N}) + \frac{\gamma}{\lambda}\prox_{g_2}^{\lambda}(\theta_n^N, X_n^{i, N})_X+\sqrt{2\gamma}\;\xi_{n+1}^{i, N}
\end{align*}
\textbf{end for}
\vspace{5pt}
\EndFor
\Return $\theta_{K+1}^N$
\end{algorithmic}
\end{algorithm}

We extend the results of \cite{caprio2024error} to provide a nonasymptotic bound for MYPGD. To do so, we consider the following metric on $\mathbb{R}^{d_\theta}\times \mathcal{P}_2(\mathbb{R}^{d_x})$
\begin{equation*}
    \boldsymbol{\md}((\theta, q), (\theta', q'))= \sqrt{\Vert\theta-\theta'\Vert^2 + W_2^2(q, q')}.
\end{equation*}
Under similar assumptions to those used in Theorem~\ref{th:myipla_main_text} we obtain the following result.
\begin{theorem}[MYPGD]
\label{th:mypgd}
    Let \textbf{A\ref{assumption_1}}--\textbf{A\ref{assumption_4}} 
    hold. If $X_0^{1},\dots, X_0^{N}$ are drawn independently from a distribution $q_0$ in $\mathcal{P}_2(\mathbb{R}^{d_x})$ and $\lambda>0$, $\gamma \leq 1/(L_{g_1} + \lambda^{-1} + \mu)$, then
    \begin{align*}
        \mathbb{E}[\Vert \theta_n^N - \Bar{\theta}_\star\Vert^2]^{1/2} \leq & \frac{\lambda}{\mu} \Big(\frac{\Vert g_2\Vert_{\textnormal{Lip}}^2}{2} A  +  B\Big) + \frac{(L_{g_1} + \lambda^{-1})\sqrt{2}}{\mu \sqrt{N}}\sqrt{B_0+\frac{2 d_x}{\mu}}\\& + \boldsymbol{\md}((\theta_0, q_0), (\Bar{\theta}_{\star,\lambda}, \pi_{\star,\lambda})) e^{-\mu n\gamma} + A_{0,\gamma, \lambda}\gamma^{1/2}+\mathcal{O}(\lambda^2) 
    \end{align*}
for all $n\in\mathbb{N}$; where $B_0=\Vert \theta_0\Vert^2 + \sup_{i\in 1, \dots, N}\mathbb{E}[\Vert X_0^{i, N}\Vert^2]<\infty$ and 
\begin{equation*}
    A_{0,\gamma, \lambda}=\sqrt{\frac{4\gamma + 4/a}{a} 220 (L_{g_1}+\lambda^{-1})^2\Big(\gamma(L_{g_1}+\lambda^{-1})^2\Big[B_0 + \frac{2 d_x}{\mu}\Big] + d_x\Big)}, \quad\; a =\frac{2 (L_{g_1}+\lambda^{-1})\mu }{L_{g_1}+\lambda^{-1} +\mu}.
\end{equation*}
\end{theorem}
\begin{proof}
Let us denote by $(\theta_n^N, Q_n^{N,\gamma})$ the MYPGD output after $n$ iterations using a discretisation step $\gamma$ and by $Q_{\star,\lambda}^N$ the empirical distribution of $N$ i.i.d. particles drawn from $\pi_{\Bar{\theta}_{\star,\lambda}}$. Using the triangular inequality, we have
\begin{equation*}
    \mathbb{E}[\Vert \theta_n^N - \Bar{\theta}_\star\Vert^2]^{1/2} \leq \Vert \Bar{\theta}_{\star} - \Bar{\theta}_{\star, \lambda}\Vert +\mathbb{E}[\Vert \theta_n^N - \Bar{\theta}_{\star, \lambda}\Vert^2]^{1/2} \leq \Vert \Bar{\theta}_{\star} - \Bar{\theta}_{\star, \lambda}\Vert + \boldsymbol{\md}((\theta_n^N,Q_n^{N,\gamma}), (\Bar{\theta}_{\star, \lambda}, Q_{\star,\lambda}^N)). 
\end{equation*}
The term $\Vert \Bar{\theta}_{\star} - \Bar{\theta}_{\star, \lambda}\Vert$ can be upper bounded by $\frac{\lambda}{\mu} \Big(\frac{\Vert g_2\Vert_{\textnormal{Lip}}^2}{2} A  +  B\Big) + \mathcal{O}(\lambda^2)$ using Proposition \ref{prop:theta_convergence_lambda}, while a bound for the second term $\boldsymbol{\md}((\theta_n^N,Q_n^{N,\gamma}), (\Bar{\theta}_{\star, \lambda}, Q_{\star,\lambda}^N))$ is derived in \citet[Theorem 7]{caprio2024error}, which gives the desired result.
\end{proof}

Selecting $\gamma = \lambda$ and $g_1=0$ in MYPGD we obtain an extension of PGD corresponding to the PIPULA algorithm introduced in Section~\ref{sec:myula}, that we termed Proximal PGD (PPGD).
Obtaining a rigorous bound like that in Theorem \ref{th:mypgd} for this algorithm is more challenging due to the presence of $\gamma$ both as time discretisation parameter and in the Lipschitz constant of $\nabla U^\gamma$.
In particular, while under \textbf{A\ref{assumption_1}}--\textbf{A\ref{assumption_1prime}}, \textbf{A\ref{assumption_4}} and \textbf{B\ref{assumption_22}} \citet[Lemma 10 and 11]{caprio2024error} hold with $\lambda = \gamma$, establishing a result controlling the time discretisation error like that in \citet[Lemma 12]{caprio2024error} is not straightforward.}

\section{Convergence to Wasserstein gradient flow}
\label{app:gf}
{\normalsize
We now show that the continuous time interacting particle system introduced in~\eqref{eq:particles_proximal_1}--\eqref{eq:particles_proximal_2} converges in the large $N$ limit (i.e. $N\to \infty$) to a McKean--Vlasov SDE with a solution whose law satisfies the Euclidean-Wasserstein gradient flow
\begin{align*}
    \bm{\Dot{\theta}}_{\lambda, t} &= -\nabla_{\theta} F(\bm{\theta}_{\lambda, t}, q_{\lambda, t}) = -\int \nabla_{\theta} U^\lambda(\bm{\theta}_{\lambda, t}, x) q_{\lambda, t}(x) \md x,\\
    \Dot{q}_{\lambda, t} &= -\nabla_{q} F(\bm{\theta}_{\lambda, t}, q_{\lambda, t}) = \nabla_x\cdot\Big[q_{\lambda, t}\nabla_x \log\Big(\frac{q_{\lambda, t}}{p_{\bm{\theta}_{\lambda, t}}^\lambda(\cdot,y)}\Big)\Big],
\end{align*}
where $p_{\bm{\theta}_{\lambda, t}}^\lambda$ denotes the Moreau-Yosida envelope of $p_{\bm{\theta}_{\lambda, t}}$.
This result is classical in the study of McKean--Vlasov SDEs, where is referred to as \emph{propagation of chaos} (e.g. \citet[Theorem 1.4]{sznitman1991topics}).
    
We start by proving the following auxiliary result.
Let us denote, for any $\theta\in\real^{d_\theta}$ and $\nu\in\cP(\real^{d_x})$, $g(\theta, \nu):= \int_{\real^{d_x}} \nabla_{\theta} U^\lambda(\theta, x^\prime)\nu(x^\prime)\md x^\prime$.
  
\begin{lemma}
\label{lem:b_lip}
The function $g:\real^{d_\theta} \times \cP(\real^{d_x}) \to \mathbb{R}^{d_\theta} $ is Lipschitz continuous in both arguments, i.e.,
\begin{align*}
  \|g(\theta_1, \nu_1) - g(\theta_2, \nu_2)\| \leq \lambda^{-1}\left( \|\theta_1 - \theta_2\| +W_1(\nu_1, \nu_2)\right).
\end{align*}
\end{lemma}
\begin{proof}
    \citet[Proposition 12.19]{Rockafellar_B._1998} shows that $\nabla U^\lambda$ is Lipschitz continuous with Lipschitz constant $\lambda^{-1}$. Then the result follows from \citet[Lemma 5]{akyildiz2023interacting}.
\end{proof}

We can now show convergence of~\eqref{eq:particles_proximal_1}--\eqref{eq:particles_proximal_2} to the following McKean--Vlasov SDE
\begin{align}
\label{eq:gf_lambda}
    \md \bm{\theta}_{\lambda, t}&=-\left[\int \nabla_{\theta} U^{\lambda}(\bm{\theta}_{\lambda, t}, x)q_{\lambda, t}(x)\md x \right]\md t\\
    \md \mathbf{X}_{\lambda, t}&=-\nabla_{x} U^{\lambda}(\bm{\theta}_{\lambda, t}, \mathbf{X}_{\lambda, t})\md t +\sqrt{2} \md \mathbf B_t. \notag
\end{align}
  
\begin{proposition}[Propagation of chaos]
\label{prop:poc}
For any (exchangeable) initial condition $(\theta_0^{N}, X_0^{1:N})$ such that $(\theta_0^{N}, X_0^{j, N}) = (\theta_0, X_0)$ for $j=1,\dots, N$ with $\mathbb{E}\left[\vert \theta_0\vert^2+\vert X_0\vert^2\right]<\infty$, we have for any $T \geq 0$
\begin{equation}
\label{eq:poc}
    \mathbb{E}\left[\sup_{t \in [0,T]} \left(\norm{\bm{\theta}_{\lambda, t} - \bm{\theta}^N_t}+\norm{\mathbf{X}_{\lambda, t} - \mathbf{X}_t^{j,N}}\right)\right] \leq \frac{\sqrt{2}(\sqrt{C_T}\lambda^{-1}+\sqrt{T})e^{2T\lambda^{-1}}}{N^{1/2}}
\end{equation}
where $C_T:= \sup_{t\leq T}\mathbb{E}\left[\vert \bm{\theta}^N_t\vert^2+\vert \mathbf{X}_t^{j, N}\vert^2\right]<\infty$, for any $j=1, \dots, N$.
\end{proposition}
\begin{proof}
The proof exploits the Lipschitz continuity of $\nabla U^\lambda$ and of $g$ established in Lemma~\ref{lem:b_lip}. The argument is classical and omitted, see \citet[Proposition 8]{akyildiz2023interacting} for the proof in a similar context.
\end{proof}

We can further show that~\eqref{eq:gf_lambda} converges to the following MKVSDE
\begin{align}
    \label{eq:gf_sde}
    \md \bm{\theta}_{ t}&=-\left[\int \nabla_{\theta} U(\bm{\theta}_{t}, x)q_{ t}(x)\md x \right]\md t\\
    \md \mathbf{X}_{ t}&=-\nabla_{x} U(\bm{\theta}_{ t}, \mathbf{X}_{ t})\md t +\sqrt{2} \md \mathbf B_t, \notag
\end{align}
associated with the gradient flow~\eqref{eq:gradient_flow_1}, when $\lambda\to 0$.

\begin{proposition}
Assume that $U$ is gradient Lipschitz with constant $\norm{\nabla U}_{\textnormal{Lip}}$. For any initial condition $(\theta_0, X_0)$ such that  $\mathbb{E}\left[\vert \theta_0\vert^2+\vert X_0\vert^2\right]<\infty$, we have for any $T \geq 0$
    \begin{equation*}
        \mathbb{E}\left[\sup_{t \in [0,T]} \left(\Vert \bm{\theta}_{\lambda, t} -\bm{\theta}_{t}\Vert^2+\norm{\mathbf{X}_{\lambda, t}-\mathbf{X}_{t}}^2\right)\right] \leq \left(\lambda^2\norm{\nabla U}_{\textnormal{Lip}}^4 C_T+\norm{\nabla U}_{\textnormal{Lip}}^2\mathcal{O}(\lambda^4)\right)T \exp(2\norm{\nabla U}_{\textnormal{Lip}}^2T),
    \end{equation*}
    where $C_T:= \sup_{t\leq T}\mathbb{E}\left[\vert \bm{\theta}_{\lambda, t}\vert^2+\vert \mathbf{X}_{\lambda, t}\vert^2\right]<\infty$. It follows that, as $\lambda\to 0$, \eqref{eq:gf_lambda} converges to~\eqref{eq:gf_sde} in $\mathbb{L}^2$.
\end{proposition}
\begin{proof}
For any $t\geq 0$, we have
\begin{align*}
  \bm{\theta}_t &= \theta_0+\int_0^t \left[-\int \nabla_\theta U(\bm{\theta}_s, x)q_s(x)\md x\right]\md s,\\
  \mathbf{X}_t &= X_0-\int_0^t \nabla_x U(\bm{\theta}_s, \mathbf{X}_s)\md s +\sqrt{2} \bm{B}_t,
\end{align*}
and equivalently for $\bm{\theta}_{\lambda, t}, {X}_{\lambda, t}$. We first observe that
\begin{align*}
    \Vert \bm{\theta}_{\lambda, t} -\bm{\theta}_{t}\Vert^2 &=\Big\Vert{\int_0^t\left[\int \nabla_{\theta} U^{\lambda}(\bm{\theta}_{\lambda, s}, x)q_{\lambda, s}(x)\md x  - \int \nabla_{\theta} U(\bm{\theta}_{s}, x)q_{ s}(x)\md x \right]\md s}\Big\Vert^2\\
   &=\Big\Vert{\int_0^t\mathbb{E}\left[ \nabla_{\theta} U^{\lambda}(\bm{\theta}_{\lambda, s},\mathbf{X}_{\lambda, s})  - \nabla_{\theta} U(\bm{\theta}_{s}, \mathbf{X}_s)\right]\md s}\Big\Vert^2\\
   &\leq\mathbb{E}\left[\Big\Vert{\int_0^t [\nabla_{\theta} U^{\lambda}(\bm{\theta}_{\lambda, s},\mathbf{X}_{\lambda, s})  - \nabla_{\theta} U(\bm{\theta}_{s}, \mathbf{X}_s)]\md s}\Big\Vert^2\right],
\end{align*}
and
\begin{align*}
   \mathbb{E}\left[ \norm{\mathbf{X}_{\lambda, t}-\mathbf{X}_{t}}^2\right] &= \mathbb{E}\left[ \Big\Vert{\int_0^t [\nabla_x U(\bm{\theta}_s, \mathbf{X}_s)-\nabla_x U^\lambda(\bm{\theta}_{\lambda, s},\mathbf{X}_{\lambda, s})]\md s}\Big\Vert^2\right].
\end{align*}
Combining the above we obtain 
\begin{align*}
\mathbb{E}\bigg[\sup_{s\in [0, t]}    \Vert \bm{\theta}_{\lambda, s} &-\bm{\theta}_{s}\Vert^2+\mathbb{E}[\norm{\mathbf{X}_{\lambda, s}-\mathbf{X}_{s}}^2]\bigg] \leq \mathbb{E}\left[ \int_0^t \norm{\nabla U(\bm{\theta}_s, \mathbf{X}_s)-\nabla U^\lambda(\bm{\theta}_{\lambda, s},\mathbf{X}_{\lambda, s})}^2\md s\right]\\
    &\leq 2\mathbb{E}\left[ \int_0^t \norm{\nabla U(\bm{\theta}_s, \mathbf{X}_s)-\nabla U(\bm{\theta}_{\lambda, s},\mathbf{X}_{\lambda, s})}^2\md s\right]\\
    &+2\mathbb{E}\left[ \int_0^t \norm{\nabla U(\bm{\theta}_{\lambda, s},\mathbf{X}_{\lambda, s})-\nabla U^\lambda(\bm{\theta}_{\lambda, s},\mathbf{X}_{\lambda, s})}^2\md s\right]\\
    &\leq 2\norm{\nabla U}_{\textnormal{Lip}}^2 \int_0^t \left(\Vert \bm{\theta}_{\lambda, s} -\bm{\theta}_{s}\Vert^2+\mathbb{E}[\norm{\mathbf{X}_{\lambda, s}-\mathbf{X}_{s}}^2]\right)\md s\\
    &+2\mathbb{E}\left[ \int_0^t \norm{\nabla U(\bm{\theta}_{\lambda, s},\mathbf{X}_{\lambda, s})-\nabla U^\lambda(\bm{\theta}_{\lambda, s},\mathbf{X}_{\lambda, s})}^2\md s\right].
\end{align*}
In the case in which $\nabla U$ is Lipschitz continuous, we further have that $\nabla U^\lambda(v) = \nabla U(\prox_U^\lambda(v))$ for $v=(\theta, x)$ \cite[Section 2]{pereyra2016proximal}, and we have
\begin{align*}
    \norm{\nabla U(v) - \nabla U^\lambda(v)} = \norm{\nabla U(v) - \nabla U(\prox_U^\lambda(v))}\leq \norm{\nabla U}_{\textnormal{Lip}}\norm{v - \prox_U^\lambda(v)}.
\end{align*}
Recalling that, since $U$ is gradient Lipschitz, $\prox_{U}^{\lambda}(v) = v - \lambda\nabla U(v) + \mathcal{O}(\lambda^2)$ \citep[Section 3.3]{proximal_algorithms_neal}, we further have that
\begin{align}
\label{eq:u_lip}
    \norm{\nabla U(v) - \nabla U^\lambda(v)} &\leq \norm{\nabla U}_{\textnormal{Lip}}(\lambda\norm{\nabla U(v)}+\mathcal{O}(\lambda^2))\nonumber\\
    &\leq \norm{\nabla U}_{\textnormal{Lip}}(\lambda\norm{\nabla U}_{\textnormal{Lip}}(1+\norm{v})+\mathcal{O}(\lambda^2)),
\end{align}
and we can bound
\begin{align*}
    \mathbb{E} &\left[ \int_0^t \norm{\nabla U(\bm{\theta}_{\lambda, s},\mathbf{X}_{\lambda, s})-\nabla U^\lambda(\bm{\theta}_{\lambda, s},\mathbf{X}_{\lambda, s})}^2\md s\right]\\
    &\leq  \lambda^2\norm{\nabla U}_{\textnormal{Lip}}^4\int_0^t\mathbb{E}\left[1+\norm{(\bm{\theta}_{\lambda, s},\mathbf{X}_{\lambda, s})}^2\right]\md s+\norm{\nabla U}_{\textnormal{Lip}}^2\mathcal{O}(\lambda^4)t\\
&\leq  \lambda^2\norm{\nabla U}_{\textnormal{Lip}}^4 C_Tt+\norm{\nabla U}_{\textnormal{Lip}}^2\mathcal{O}(\lambda^4)t
\end{align*}
with $C_T$ given in the statement of the result.

Let us denote by $h(t) =\sup_{s\in[0, t]} \Vert \bm{\theta}_{\lambda, s} -\bm{\theta}_{s}\Vert^2+\mathbb{E}[\Vert{\mathbf{X}_{\lambda, s}-\mathbf{X}_{s}}\Vert^2]$. Then, using the bounds above we have that 
\begin{equation*}
    h(t) \leq 2\norm{\nabla U}_{\textnormal{Lip}}^2\int_0^t h(s) \md s + \left(\lambda^2\norm{\nabla U}_{\textnormal{Lip}}^4 C_T+\norm{\nabla U}_{\textnormal{Lip}}^2\mathcal{O}(\lambda^4)\right)t.
\end{equation*}
Using Gronwall's inequality we obtain
\begin{align*}
    h(t)&\leq \left(\lambda^2\norm{\nabla U}_{\textnormal{Lip}}^4 C_T+\norm{\nabla U}_{\textnormal{Lip}}^2\mathcal{O}(\lambda^4)\right)t \exp(2\norm{\nabla U}_{\textnormal{Lip}}^2t),
\end{align*}
from which the result follows.
\end{proof}}

\section{Algorithm Comparison}\label{app:algorithm_comparison}
We provide further details on computing the complexity estimates of Section \ref{subsec:algorithm_comparison}. 
For convenience, we summarise the complexity estimates for $\lambda$, the number of particles $N$, $\gamma$, and the number of steps $n$ to achieve an error $\mathbb{E}\left[\Vert \theta_n^N-\bar{\theta}^\star \Vert^2\right]^{1/2}=\mathcal{O}(\varepsilon)$ from Table \ref{tab:comparison_table}. The values are provided in terms of the key parameters $d_\theta, d_x$ and $\delta>0$ is any small positive constant. 
\begin{center}
\begin{tabular}{ |c ||c |c |c |c|}
\hline
 & $\lambda$& $N$ & $\gamma$ & $n$  \\ 
 \hline \hline
MYIPLA  & $\mathcal{O}(\varepsilon)$& $\mathcal{O}(d_\theta \varepsilon^{-2})$ & $\mathcal{O}(d_x^{-1}\varepsilon^2)$ & $\mathcal{O}(d_x \varepsilon^{-2-\delta})$
\\
 \hline
 PIPGLA & $\mathcal{O}(\varepsilon^2)$ &$\mathcal{O}(d_\theta \varepsilon^{-2})$ & $\mathcal{O}(d_x^{-1}\varepsilon^2)$ & $\mathcal{O}(\log \varepsilon^2/\log d_x)$\\
 \hline
  MYPGD & $\mathcal{O}(\varepsilon)$& $\mathcal{O}(d_x \varepsilon^{-2})$& $\mathcal{O}(d_x^{-1}\varepsilon^2)$ & $\mathcal{O}(d_x \varepsilon^{-2-\delta})$\\
 \hline
\end{tabular}
\end{center}
The bound for MYIPLA follows from first choosing $\lambda$ so that the first term in Theorem~\ref{th:myipla_main_text} is $\mathcal{O}(\varepsilon)$, then choosing $N$ so that the second term is $\mathcal{O}(\varepsilon)$ and $\gamma$ sufficiently small to counteract the dependence on $d_x$ in the fourth term. Finally, since for every $p\in \mathbb{N}$ one has $e^x\geq x^p/p!$ for $x>0$, for every $\delta>0$ (by choosing $p \in \mathbb{N}$ large enough) one has $e^{-\varepsilon^\delta}\leq C\varepsilon$. Therefore, as long as $n$ is chosen sufficiently large that $\mu n \gamma = \mathcal{O}(\varepsilon^{-\delta})$, the exponential decay is strong enough so that the middle term is of order $\mathcal{O}(\varepsilon)$. A similar approach based on the bound in Theorem~\ref{th:mypgd} provides the bounds for MYPGD.

On the other hand, the bound for PIPGLA follows from first choosing $N$ so that the first term in Theorem \ref{th:pipgla_main_text} is $\mathcal{O}(\varepsilon)$, then $\lambda$ and $\gamma$ to counteract the dependence of $d_x$ on the third term. Finally, considering the values of $\lambda$ and $\gamma$, we select $n$ to ensure that the second term is $\mathcal{O}(\varepsilon)$.

On the other hand, we also compared the algorithms in terms of their computational requirements. We account for the computational cost of running each algorithm for $n$ iterations with $N$ particles and time discretisation step $\gamma$, while guaranteeing an $\mathcal{O}(\varepsilon)$ error, in terms of component-wise evaluations of $\nabla g_1$ and $\prox_{g_2}^\lambda$, and number of independent standard $1$-dimensional Gaussian samples.

For every step of MYIPLA, PIPGLA and MYPGD one requires $N(d_\theta+d_x)$ evaluations of $\nabla g_1$ component-wise and $N(d_\theta+d_x)$ evaluations of $\prox_{g_2}^\lambda$ component-wise.
In the case of MYIPLA and PIPGLA, we need $d_\theta+Nd_x$ independent standard $1$-dimensional Gaussians for each iteration; since MYPGD does not have a noise in the $\theta$-component this reduces to $Nd_x$.

\begin{center}
\begin{tabular}{ |c ||c |c |c|}
\hline
 &  Evaluations of $\nabla g_1$ &  Evaluations of $\prox_{g_2}^\lambda$ & Indep. $1$d Gaussians \\ 
 \hline \hline
 MYIPLA  & $\mathcal{O}(d_\theta d_x(d_\theta+d_x)\varepsilon^{-4-\delta})$ & $\mathcal{O}(d_\theta d_x(d_\theta+d_x)\varepsilon^{-4-\delta})$ &  $\mathcal{O}(d_\theta d_x^2\varepsilon^{-4-\delta})$\\
 \hline
 PIPGLA & $\mathcal{O}(d_\theta (d_\theta+d_x)\varepsilon^{-2}\frac{\log \varepsilon^{2}}{\log d_x})$  & $\mathcal{O}(d_\theta (d_\theta+d_x)\varepsilon^{-2}\frac{\log \varepsilon^{2}}{\log d_x})$ &  $\mathcal{O}(d_\theta d_x\varepsilon^{-2}\frac{\log \varepsilon^{2}}{\log d_x})$
 \\ 
 \hline
  MYPGD & $\mathcal{O}(d_x^2(d_\theta+d_x)\varepsilon^{-4-\delta})$ & $\mathcal{O}(d_x^2(d_\theta+d_x)\varepsilon^{-4-\delta})$ &  $\mathcal{O}( d_x^3\varepsilon^{-4-\delta})$\\ 
 \hline
\end{tabular}
\end{center}

Finally, Table \ref{tab:method_comparison} summarises the key differences and advantages of each algorithm.
We recall that $L_{g_1}$ and $\mu$ denote the Lipschitz continuity and strong-convexity parameters of $g_1$, introduced in assumptions \textbf{A\ref{assumption_1}} and \textbf{A\ref{assumption_3}}, respectively.

\begin{table}[h]
\smaller
\caption{Comparison of convergence assumptions, parameter constraints, and advantages of each algorithm.}
\label{tab:method_comparison}
\centering
\begin{tabular}{llllllll}
\toprule
\multirow{2}{*}{} & \textbf{Assumptions for} & \multirow{2}{*}{\textbf{Constraints on $\lambda$}} & \multirow{2}{*}{\textbf{Constraints on $\gamma$}} & \multirow{2}{*}{\textbf{Advantages}} \\
& \textbf{convergence} & & & \\
\midrule
\multirow{2}{*}{MYIPLA} & \multirow{2}{*}{\textbf{A\ref{assumption_1}}--\textbf{A\ref{assumption_4}} }& \multirow{2}{*}{$\lambda \geq 0$} & \multirow{2}{*}{$\gamma < \min\left\{(L_{g_1} + \lambda^{-1})^{-1},\, 2\mu^{-1} \right\}$} & Noise in the $\theta$-dynamics \\
& & & & helps escape local minima. \\
\midrule
\multirow{2}{*}{PIPGLA} & \multirow{2}{*}{\textbf{A\ref{assumption_1}}, \textbf{A\ref{assumption_2}}, \textbf{A\ref{assumption_3}} and \textbf{C\ref{assumption_1_gla}}}& \multirow{2}{*}{$\gamma \leq \lambda \leq \gamma/(1 - \mu\gamma)$} & \multirow{2}{*}{$\gamma \leq 1/L_{g_1}$} & Ensures $\theta$ estimates remain  \\
& & & &within the support of the distribution.  \\
\midrule
\multirow{2}{*}{MYPGD} & \multirow{2}{*}{\textbf{A\ref{assumption_1}}--\textbf{A\ref{assumption_4}}} & \multirow{2}{*}{$\lambda \geq 0$} &\multirow{2}{*}{$\gamma < (L_{g_1} + \lambda^{-1} + \mu)^{-1}$}& Produces lower-variance estimates \\
& & & &for MMLE in the strongly convex setting. \\
\bottomrule
\end{tabular}
\end{table}
\normalsize

\section{Numerical Experiments}\label{ap:experiments}
\subsection{Derivation of the proximal operators}
{\normalsize
\subsubsection{Laplace Prior with Unknown Mean $\theta$}\label{ap:proximal_op_laplace_mean}

We recall that using a Laplace prior $g_2(\theta, x) = \sum_{i=1}^{d_x}|x_i-\theta|$.
\begin{equation*}
    \prox_{g_2}^{\lambda}(\theta,x)=\argmin_{(u_0,u)} h(u_0, u)= \argmin_{(u_0,u)} \{g_2(u_0, u) + \Vert (u_0, u)-(\theta, x)\Vert^2/{(2\lambda)}\}.
\end{equation*}
The first order optimality condition is given by
\begin{equation*}
    0\in\partial g_2(u_0, u) + \nabla\big(\Vert(u_0, u)-(\theta, x)\Vert^2/(2\lambda)\big).
\end{equation*}

We recall that $\phi\in\real^d$ is a subdifferential of the $\ell^1$-norm at $x\in\real^d$ if and only if $\phi_i(x) =  \text{sign}(x_i)$ if $x_i\neq 0$ and $\vert \phi_i(x)\vert \leq 1$ otherwise \citep{proximal_algorithms_neal}.

Let us define the set $D=\{i\in\{1,\dots, d_x\}| u_i-u_0 = 0\}$. Then, the first order optimality condition becomes 
\begin{align*}
        0&\in 
       \left\lbrace- \sum_{i\notin D} t_i  -\sum_{i\notin D}\text{sign}(u_{i}-u_0) + (u_0-\theta)/\lambda\ |\ \vert t_i\vert \leq 1\right\rbrace\\
        \ &\begin{cases}
        0\in \left\lbrace t_i + \frac{u_i-x_i}{\lambda}\ |\ \vert t_i\vert \leq 1\right\rbrace \quad\quad \text{if $i\in D$}\\
        0 = \text{sign}(u_{i}-u_0)+ (u_i-x_i)/\lambda \quad\text{if $i\not\in D$}
    \end{cases}.
\end{align*}

Reordering terms, we get
\begin{align}
    u_0&\in\left\lbrace 
       \theta + \lambda\bigg(\sum_{i\notin D} t_i  -\sum_{i\notin D}\text{sign}(u_{i}-u_0)\bigg)\ |\ \vert t_i\vert \leq 1\right\rbrace,\label{eq:iterative_regression_11}
\\
    \ & \begin{cases}
       u_i\in \left\lbrace x_i - \lambda t_i\ |\  \vert t_i\vert \leq 1\right\rbrace  \quad\quad \text{if $i\in D$},\\
        u_i = x_i - \lambda\; \text{sign}(u_{i}-u_0)\quad\text{if $i\not\in D$}.
    \end{cases}\label{eq:iterative_regression_22}
\end{align}
Assuming that $D = \emptyset$, the previous system of equations can be solved iteratively using a fixed point algorithm. 

Alternatively, for a lower computational cost we can obtain an approximate solution by setting $u_0=\theta$ in (\ref{eq:iterative_regression_22})
\begin{equation*}
     \begin{cases}
       u_i\in \left\lbrace x_i - \lambda t_i\ | \ \vert t_i\vert \leq 1\right\rbrace  \quad\quad \text{if $i\in D$},\\
        u_i = x_i - \lambda\; \text{sign}(u_{i}-\theta)\quad\text{if $i\not\in D$}.
    \end{cases}
\end{equation*}
which is solved applying the soft-thresholding operator
\begin{equation*}
    u_i = \theta + [x_i -\theta-\lambda\;\text{sign}(x_i-\theta)]\mathds{1}\{\vert x_i-\theta\vert \geq\lambda \}.
\end{equation*}
Using these $u_i$'s, taking $u_0=\theta$ in the right-hand side of (\ref{eq:iterative_regression_11}) and assuming $D = \emptyset$, we obtain
\begin{equation*}
    u_0 = \theta + \lambda\sum_{i=1}^{d_x} \text{sign}(u_i-\theta).
\end{equation*}

\subsubsection{Laplace Prior with Unknown Scale $e^{2\theta}$}\label{ap:proximal_approx_bnn}
For the Bayesian neural network experiment  we consider a Laplace prior with zero mean and unknown scale parameterised by $e^{2\theta}$ (which ensures that the scale is positive), we have $g_2(\theta, x) = d_x\alpha + \sum_i |x_i|e^{-2\alpha}$.
Its proximal operator is given by
\begin{equation*}
    \prox_{g_2}^{\lambda}(\theta, x) =\argmin_{(u_0, u)} h(u_0, u), \quad h(u_0, u) = u_0d_x+ \sum_i |u_i| e^{-2u_0} + \Vert (u_0, u)-(\theta, x)\Vert^2/{(2\lambda)}.
\end{equation*}
The optimality condition is given by
\begin{equation*}
    0\in \partial\big(u_0d_x+ \sum_i |u_i| e^{-2u_0}\big)+\nabla\big(\Vert(u_0, u)-(\theta, x)\Vert^2/(2\lambda)\big),
\end{equation*}
which provides the following system of equations
\begin{align*}
    0 &= d_x-2e^{-2u_0}\sum_{i=1}^{d_x} |u_i| + \frac{1}{\lambda}(u_0-\theta),\\
    \ &\begin{cases}
        0\in \left\lbrace e^{-2u_0}t_i + (u_i-x_i)/\lambda\ |\ \vert t_i\vert \leq 1\right\rbrace\quad\quad \text{if $u_i=0$},\\
        0 = e^{-2u_0}\;\text{sign}(u_{i}) + (u_i-x_i)/\lambda \quad \text{if $u_i\neq 0$}.
    \end{cases} 
\end{align*}
Reordering terms, we get
\begin{align}
    u_0 &= \theta -\lambda d_x +2\lambda e^{-2u_0}\sum_i\vert u_i \vert\label{eq:grad_theta}\\
    \ & \begin{cases}
u_i\in \{x_i - \lambda e^{-2u_0}t_i \ | \ \vert t_i\vert \leq 1\}\quad \quad \text{if $u_i=0$},\\
        u_i=x_i - \lambda e^{-2u_0}\;\text{sign}(u_i)\quad \;\text{if $u_i\neq 0$}. 
    \end{cases}\label{eq:grad_w}
\end{align}
This system of equations can be solved using an iterative solver, however this will incur in a high computational cost. Therefore, we opt for the following approximation of (\ref{eq:grad_w}), where we set $u_0=\theta$,
\begin{equation}
    \begin{cases}
u_i\in \{x_i - \lambda e^{-2\theta}t_i \ | \ \vert t_i\vert \leq 1\}\quad \quad \text{if $u_i=0$},\\
        u_i=x_i - \lambda e^{-2\theta}\;\text{sign}(u_i)\quad \;\text{if $u_i\neq 0$}. 
    \end{cases}\label{eq:grad_w_approx}
\end{equation}
The solution of \eqref{eq:grad_w_approx} is 
\begin{equation*}
    u_i \approx [x_i-\lambda e^{-2\theta}\;\text{sign}(x_i)]\mathds{1}\{\vert x_i\vert \geq\lambda e^{-2\theta}\}. 
\end{equation*}
Using these $u_i$'s together with the Lambert $W$ function, the solution of \eqref{eq:grad_theta} is given by
\begin{equation*}
    u_0 \approx \theta -\lambda d_x + \frac{1}{2} W\left(4\lambda e^{-2\theta}\sum_i\vert u_i \vert\right).
\end{equation*}

\subsubsection{Uniform Prior}\label{ap:proximal_op_uniform_mean}

We recall that using a uniform prior 
\begin{equation*}
    g_2(\theta, x) = d_x\log(2\theta) + \sum_{i=1}^{d_x} \imath_{[-\theta,\theta]}(x_i),
\end{equation*}
where $\imath_{\mathcal{K}}$ is the convex indicator of $\mathcal{K}$ defined by $\imath_{\mathcal{K}}(x)=0$ if $x\in\mathcal{K}$ and $\imath_{\mathcal{K}}(x)=\infty$ otherwise. 
In this case, the proximal operator satisfies
\begin{align*}
    \prox_{g_2}^{\lambda}(\theta,x)&= \argmin_{(u_0,u)} \{g_2(u_0, u) + \Vert (u_0, u)-(\theta, x)\Vert^2/{(2\lambda)}\} \\
    &= \argmin_{\substack{(u_0,u)\\|u_i|\leq u_0}} \{d_x\log(2u_0) + \Vert (u_0, u)-(\theta, x)\Vert^2/{(2\lambda)}\}.
\end{align*}
We can obtain an approximate solution by deriving the first order conditions for $u_i$ with $i=0, 1, \dots, d_x$ and combining them with the constraint $|u_i|\leq u_0$:
\begin{align*}
u_0 &= \begin{cases}
    \frac{\theta + \sqrt{\theta^2-4\lambda d_x}}{2} \;\;\text{if}\;\;\theta^2\geq4\lambda d_x,\\
    \max_i |x_i|\;\; \text{otherwise},
\end{cases}\\
u_i &= \text{sign}(x_i)\cdot\min\{|x_i|, |u_0|\}. 
\end{align*}

\subsubsection{Approximation for PIPULA and PPGD}\label{ap:proximal_operator_whole}
In PIPULA and PPGD, we need to compute the proximal operator of $U = g_1 + g_2$ which is usually not available in closed form. Since $\gamma$ is normally set to a small enough value, we follow \cite{pereyra2016proximal} and approximate the proximity map of $U$ as
\begin{align*}
    \prox_{U}^{\gamma}(v) &= \argmin_{v'} \{g_1(v') + g_2(v') + \Vert v'- v \Vert^2/(2\gamma)\}\\
    &\approx\argmin_{v'} \{g_1(v) + (v'-v)^{\intercal}\nabla g_1(v) + g_2(v') + \Vert v'- v \Vert^2/(2\gamma)\}\\
    &\approx\argmin_{v'} \{g_2(v') + \Vert v'- v + 2\gamma\nabla g_1^{\intercal}(v)  \Vert^2/(2\gamma)\}\\
    &\approx \prox_{g_2}^{\gamma}(v + 2\gamma\nabla g_1^{\intercal}(v)),
\end{align*}
where $v = (\theta, x), v' = (\theta',x')$.}
\subsection{Bayesian logistic regression}\label{ap:experiments_logistic}
In the case of the Laplace prior, the negative log joint likelihood is given by 
\begin{equation*}
  -\log p_{\theta}(x,y) = \underbrace{\sum_{i=1}^{d_x}|x_i-\theta|}_{g_2(\theta, x)} + \underbrace{d_x\log2 -\log p(y|x)}_{g_1(\theta, x)};
\end{equation*}
and for the uniform prior, we obtain
\begin{equation*}
  -\log p_{\theta}(x,y) = \underbrace{d_x\log(2\theta) + \sum_{i=1}^{d_x} \imath_{[-\theta,\theta]}(x_i)}_{g_2(\theta, x)} - \underbrace{\log p(y|x)}_{g_1(\theta, x)},
\end{equation*}
where $g_1$ is differentiable and $g_2$ is lower semi-continuous, and $\imath_{\mathcal{K}}$ is the convex indicator of $\mathcal{K}$ defined by $\imath_{\mathcal{K}}(x)=0$ if $x\in\mathcal{K}$ and $\imath_{\mathcal{K}}(x)=\infty$ otherwise.

In both cases we have that 
\begin{align*}
    g_1(\theta, x) = \sum_{j=1}^{d_y}\left(y_j\log(s(v_j^Tx))+(1-y_j)\log(s(-v_j^Tx))\right)+C
\end{align*}
where $C$ is a constant. As shown in \citet[Section 6.1.1]{akyildiz2023interacting}, the function $g_1$ is gradient Lipschitz and strictly convex but not strongly convex.
The function $g_2$ satisfies \textbf{A\ref{assumption_1}} for both the Laplace and the uniform prior, as observed in \cite{pereyra2016proximal}, in the case of the Laplace prior $g_2$ also satisfies \textbf{A\ref{assumption_1prime}} while the uniform prior does not lead to a Lipschitz $g_2$. Since $g_1$ does not depend on $\theta$, \textbf{A\ref{assumption_4}} holds for the Laplace prior. 

{\normalsize
\paragraph{Dataset.} We create a synthetic dataset by first fixing the value of $\theta$ and sampling the latent variable $x\in \mathbb{R}^{50}$ from the corresponding prior. We then sample the 900 observations from a Bernoulli distribution with parameter $s(v_j^T x)$, where $s$ is the logistic function and the entries of the covariates $v_j$ are drawn from a uniform distribution $\mathcal{U}(-1,1)$. 
The true value of $\theta$ is set to $\bar{\theta}_\star = -4$ for the Laplace prior and $\bar{\theta}_\star = 1.5$ for the uniform one.

\paragraph{Implementation details.}
The $x$-gradients of $g_1$ can be computed analytically. 
To choose the optimal values of $\gamma$ and $\lambda$ for the different implementations, we perform a grid search in the range $[5\times10^{-4}, 0.5]$. The selected optimal values are displayed in Table \ref{table-logistic-hyperparameters}. 
We note that in PIPGLA the optimal values for $\lambda, \gamma$ turn out  to be when $\lambda = \gamma$.

\begin{table}[h!]
  \caption{{\normalsize Optimal hyperparameters for Bayesian logistic regression example. Recall that for PPGD and PIPULA we only have the $\gamma$ parameter since we set $\lambda=\gamma$.}}
  \label{table-logistic-hyperparameters}
  \centering
  \begin{tabular}{llllll}
    \toprule
    Algorithm     & Approx./Iterative  & \multicolumn{2}{c}{$\gamma$} & \multicolumn{2}{c}{$\lambda$} \\
    \cmidrule(r){3-4}
    \cmidrule(r){5-6}
         &   &  Laplace & Unif & Laplace & Unif \\
    \midrule
    \multirow{2}{*}{PPGD} & Approx  & $0.1$  &   $0.03$ & $-$& $-$\\
    & Iterative  & $0.06$ & $-$ &$-$ & $-$ \\
     \midrule
    \multirow{2}{*}{PIPULA} & Approx  & $0.06$  &   $0.03$ & $-$&$-$\\
    & Iterative  & $0.06$ & $-$ & $-$ & $-$\\
     \midrule
   \multirow{2}{*}{MYPGD} & Approx  & \multirow{1}{*}{$0.05$}  &   \multirow{1}{*}{$0.001$}& \multirow{1}{*}{$0.25$} & \multirow{1}{*}{$0.01$}  \\
       & Iterative  &   $0.05$&   $-$& $0.005$ & $-$  \\
        \midrule
    \multirow{2}{*}{MYIPLA} & Approx  & $0.05$& $0.001$& $0.35$& $0.01$\\
       & Iterative  &   $0.05$&   $-$& $0.005$ & $-$  \\
 \midrule
        \multirow{2}{*}{PIPGLA} & Approx  & $0.01$ & $0.02$ & $0.01$ & $0.02$\\
     & Iterative  &   $0.01$&   $-$& $0.01$ & $-$  \\
    \bottomrule
  \end{tabular}
\end{table}

\begin{table}[t!]
\caption{Bayesian logistic regression for Laplace and uniform priors. Normalised MSE (NMSE) for $\theta$ for different algorithm when run 500 times using 50 particles, 5000 steps and different starting points. Computation times and NMSEs are averaged over the 500 replicates. The second column indicates whether the proximal map is calculated approximately or iteratively, using 40 steps in each iteration. For the uniform prior case we have not implemented the iterative method.}
\label{table-logistic-comparison-extended}
\centering
\begin{tabular}{llllll}
  \toprule
  Algorithm     & Approx./Iterative  & \multicolumn{2}{c}{NMSE (\%)} & \multicolumn{2}{c}{Times (s)} \\
  \cmidrule(r){3-4}
  \cmidrule(r){5-6}
       &   &  Laplace & Unif & Laplace & Unif \\
  \midrule
  \multirow{2}{*}{PPGD} & Approx  & $14.70\pm 4.42$  &   $3.63\pm 4.93$ & $102.6\pm 5.1$& $107.9\pm 5.5$\\
  & Iterative  & $19.04\pm 1.34$ & $-$ &$122.3\pm 5.1$ & $-$ \\
   \midrule
  \multirow{2}{*}{PIPULA} & Approx  & $12.18\pm1.62$  &   $4.71\pm 6.02$ & $98.8\pm 5.7$ & $101.0\pm4.0$ \\
  & Iterative  & $19.22\pm1.28$ & $-$ & $126.2\pm 3.8$ & $-$\\
   \midrule
 \multirow{2}{*}{MYPGD} & Approx  & $6.09\pm 0.34$  &   $\mathbf{0.60\pm 0.23}$& $91.9\pm4.8$ &$109.3\pm4.6$  \\
 & Iterative  & $4.44\pm1.40$  &   $-$& $129.7\pm 15.8$ & $-$  \\
  \midrule
  \multirow{2}{*}{MYIPLA} & Approx  & $4.42\pm 1.32$  & $15.26\pm 4.44$ & ${89.9\pm4.2}$& ${97.0\pm4.2}$\\
 & Iterative  &  $4.67\pm1.60$ &   $-$& $120.5\pm10.1$ & $-$  \\
 \midrule
  \multirow{2}{*}{PIPGLA} & Approx  & $2.30\pm0.58$  & $6.83\pm 3.97$ & $116.5\pm5.5$& $103.1\pm8.0$\\
   & Iterative  & $\mathbf{2.02\pm0.54}$  & $-$ & $122.9\pm6.9$& $-$\\
\midrule
  \multirow{1}{*}{IPLA} & --  & $7.76\pm3.39$  & $20.12\pm 2.88$ & $\mathbf{81.1\pm3.0}$& $\mathbf{82.9\pm4.9}$\\
  \bottomrule
\end{tabular}
\end{table}

\paragraph{Results.} 
Table \ref{table-logistic-comparison-extended} extends the results in Table \ref{table-logistic-comparison} by also including the results for PPGD, PIPULA and IPLA (as a benchmark).
Figure~\ref{fig:results_logistic_regression2} shows the $\theta$-iterates obtained with MYIPLA and PIPGLA starting from 7 different initial values $\theta_0$ and using the approximate solver for $\prox_{g_2}^\lambda$ with $g_2(\theta, x)=\sum_{i=1}^{d_x}|x_i-\theta|$ and an iterative procedure using 40 iterations in each step.
We observe that the iterative solver results in a slightly slower convergence to stationarity, but overall the two sets of algorithms converge to the same true value of $\theta$.
We also observe that the convergence to stationarity for PIPGLA is much slower compared to MYIPLA.
However, if we increase the value of $\gamma$ in the hope of faster convergence, the iterates either do not converge to the true value or the standard deviation is significantly larger.
For all algorithms considered, approximate solvers are 25\% faster than iterative solvers (see Table \ref{table-logistic-comparison-extended}).

We also compare the results for the uniform prior, in this case we only use the approximate proximity map (Figure \ref{fig:logistic_uniform}), as the iterative approach is not numerically stable.

\begin{figure}
  \centering
    \begin{subfigure}[b]{0.32\textwidth}
        \centering
        \includegraphics[width=\textwidth]{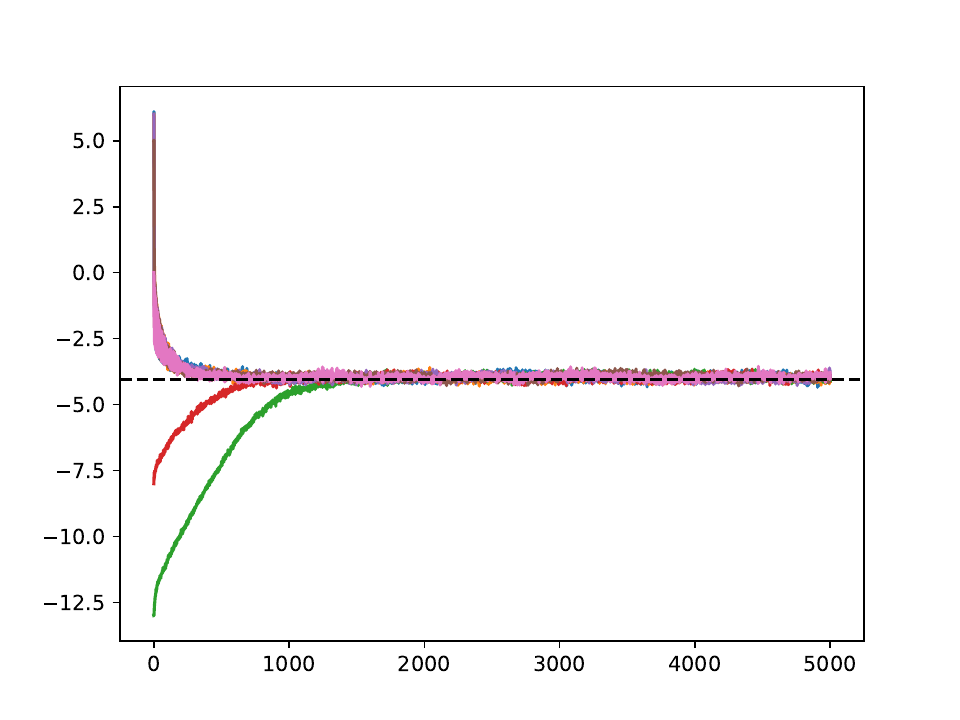} 
        \caption{MYIPLA, approximated $\prox_{g_2}^{\lambda}$.}
        \label{fig:logistic_myula}
    \end{subfigure}
    \begin{subfigure}[b]{0.32\textwidth}
        \centering
        \includegraphics[width=\textwidth]{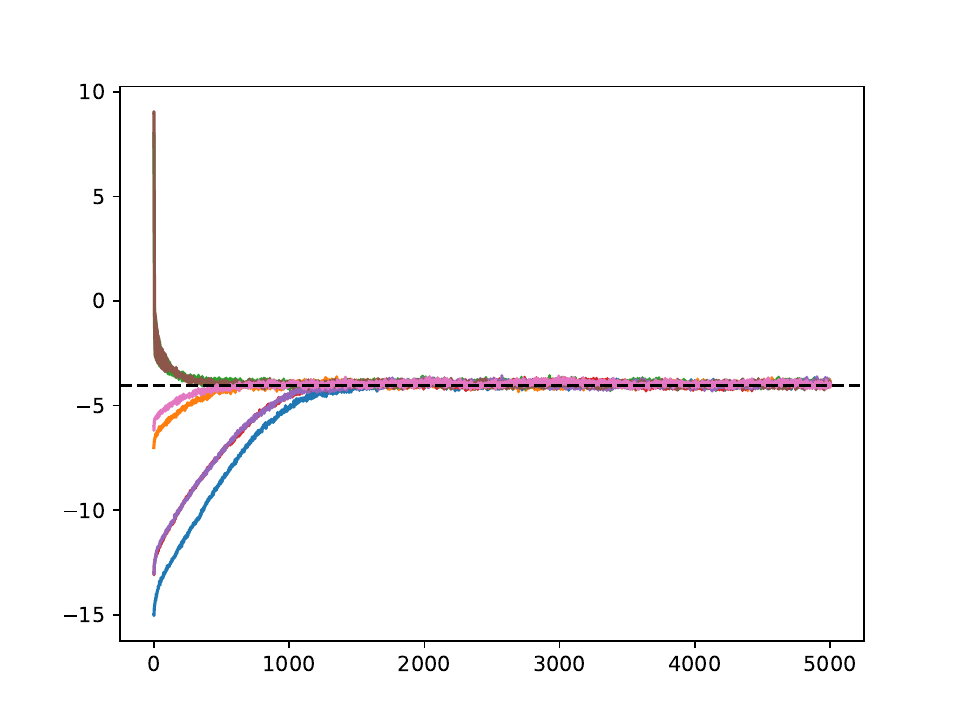} 
        \caption{MYIPLA, iterative solver for $\prox_{g_2}^{\lambda}$.}
        \label{fig:logistic_myula_iterative}
    \end{subfigure}
    \vfill
    \begin{subfigure}[b]{0.32\textwidth}
        \centering
        \includegraphics[width=\textwidth]{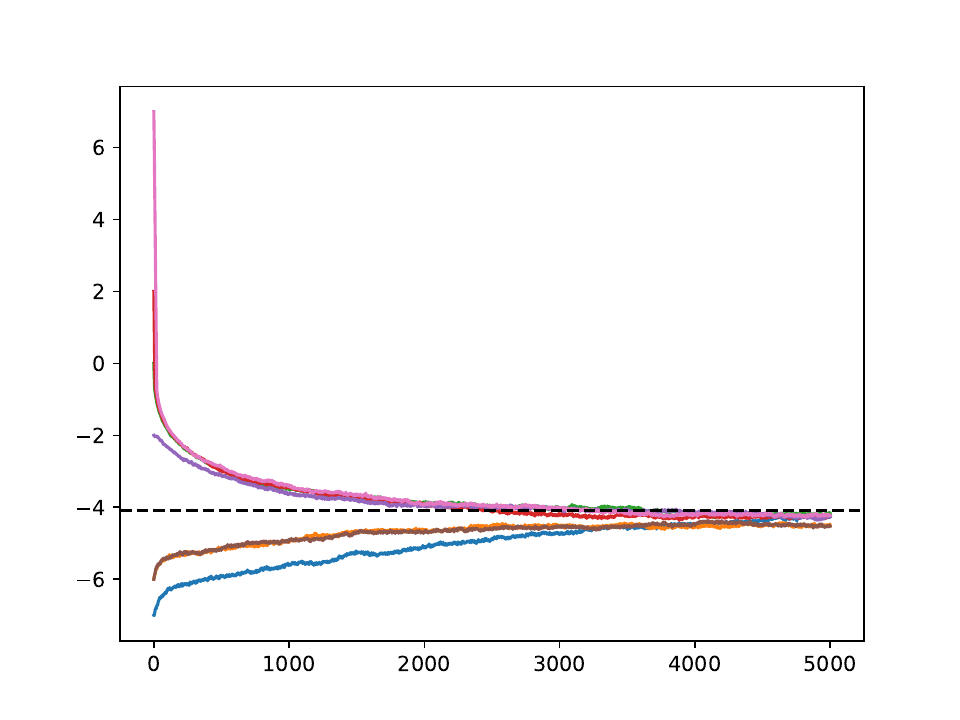} 
        \caption{PIPGLA, approximated $\prox_{g_2}^{\lambda}$.}
        \label{fig:logistic_pipgla_same}
    \end{subfigure}
    \begin{subfigure}[b]{0.32\textwidth}
        \centering
        \includegraphics[width=\textwidth]{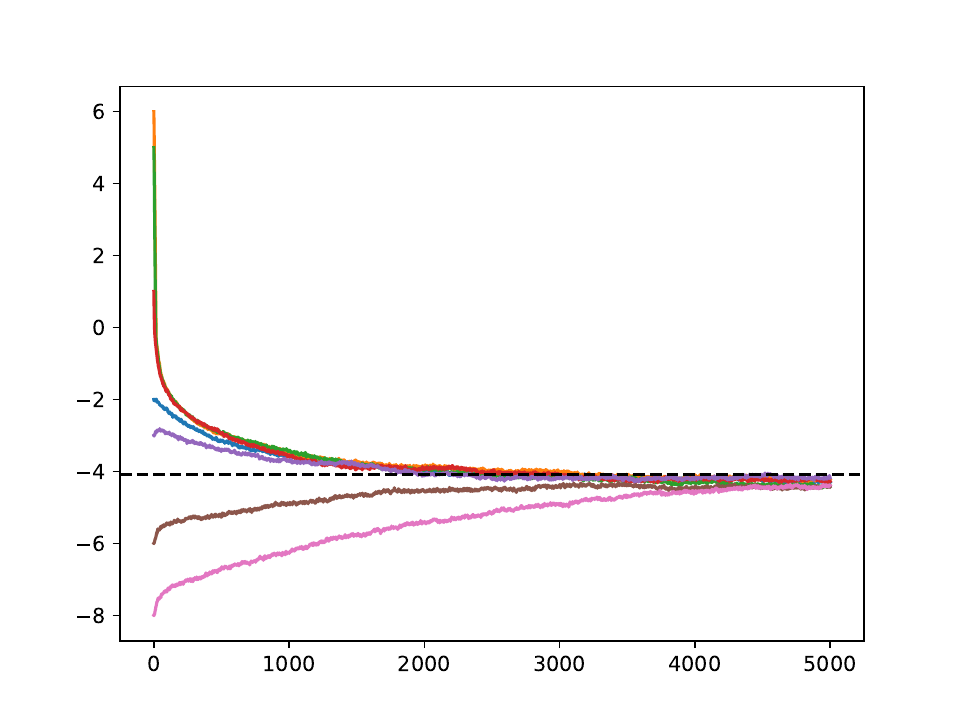} 
        \caption{PIPGLA, iterative solver for $\prox_{g_2}^{\lambda}$.}
        \label{fig:logistic_pipgla_same_iterative}
    \end{subfigure}
    \hfill
  \caption{{\normalsize Bayesian logistic regression with isotropic Laplace priors on the regression weights $\prod_i \text{Laplace}(x_i|\theta,1)$, with true $\theta=-4$. Each plot shows the $\theta$-iterates for 7 different starting points.}}  
  \label{fig:results_logistic_regression2}
\end{figure}

\begin{figure}
  \centering
      \begin{subfigure}[b]{0.32\textwidth}
        \centering
        \includegraphics[width=\textwidth]{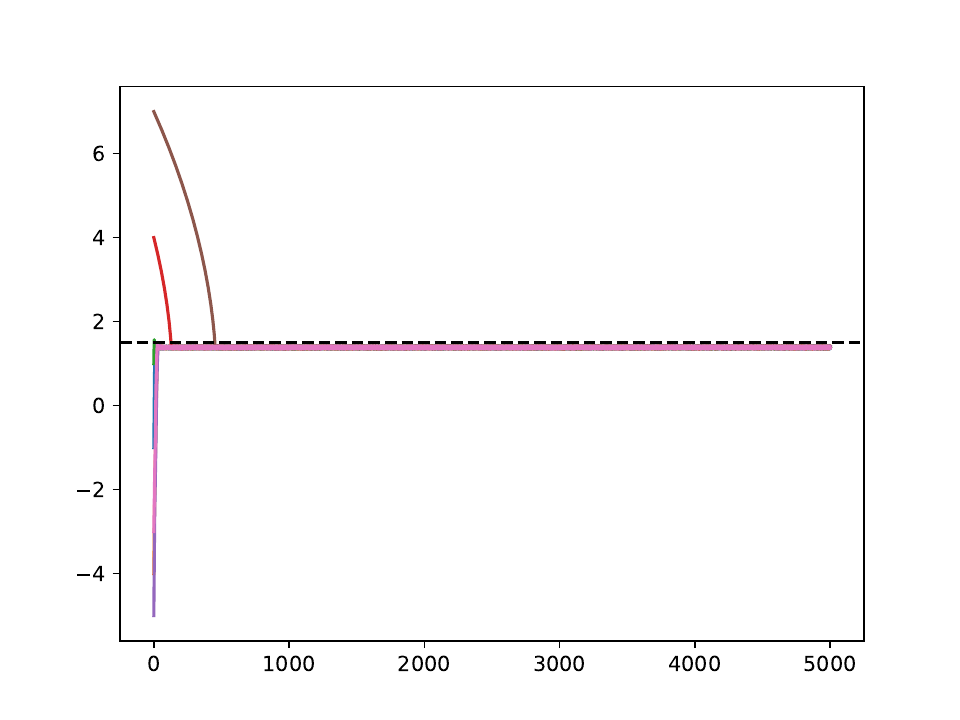}
        \caption{MYPGD.}        \label{fig:logistic_uniform_pgd}
    \end{subfigure}
    \hfill
        \begin{subfigure}[b]{0.32\textwidth}
        \centering
        \includegraphics[width=\textwidth]{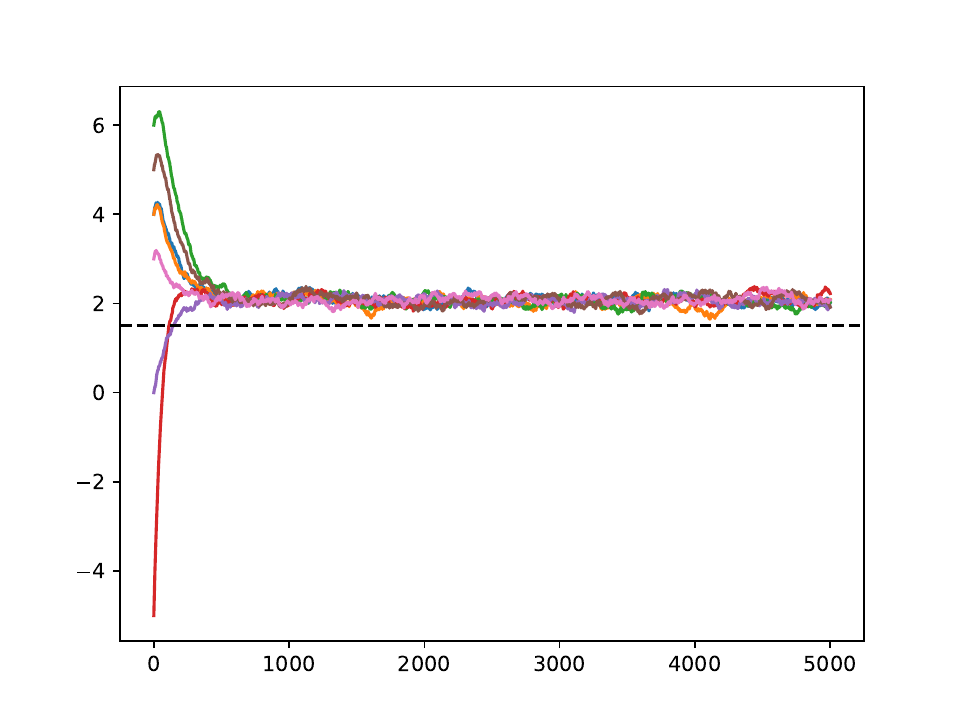}
        \caption{MYIPLA.}     \label{fig:logistic_uniform_ipla}
    \end{subfigure}
    \hfill
        \hfill
        \begin{subfigure}[b]{0.32\textwidth}
        \centering
        \includegraphics[width=\textwidth]{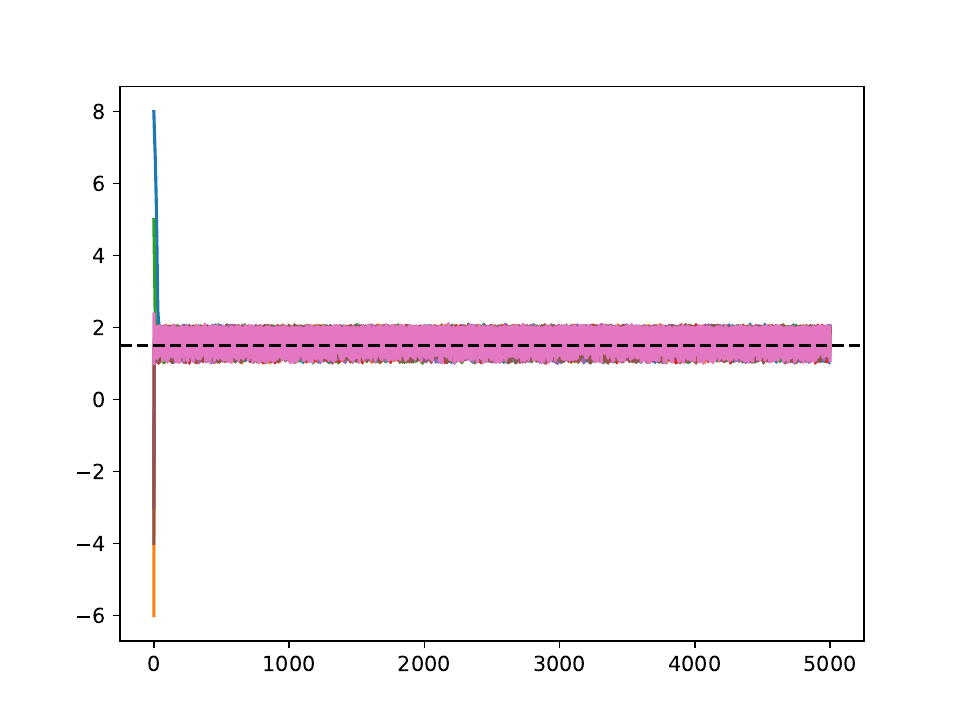}
        \caption{PIPGLA.}   \label{fig:logistic_uniform_PIPGLA}
    \end{subfigure}
    \hfill
  \caption{{\normalsize Bayesian logistic regression with isotropic uniform priors on the regression weights $\prod_i \mathcal{U}(x_i|-\theta,\theta)$, with true $\theta=1.5$. The plot displays the $\theta$-iterates for 7 randomly chosen starting points.}} 
  \label{fig:logistic_uniform}
\end{figure}

Since all the algorithms considered aim at estimating the parameter $\theta$ by sampling from a distribution which concentrates around $\bar{\theta}_\star$, we compare the estimators of $\bar{\theta}_\star$ obtained by using only the last iterate $\theta_{K+1}^N$ and averaging over a number of iterates.
We compare the normalised MSE (NMSE) for $\theta$ for the estimator obtained by averaging 
the $\theta$-iterates after discarding a burn-in of 1500 samples (column named \textit{avg}) against using the last $\theta$ of the chain (column \textit{last}).
The results are in agreement, with the NMSE for the averaged estimator having lower variance in most settings (Table~\ref{tab:blr_avg}).

\begin{table}
\caption{{\normalsize Bayesian logistic regression for Laplace and uniform priors. Normalised MSE (NMSE) for the last iterate of $\theta$ (\textit{last}) and the posterior mean after discarding a burn-in of 1500 samples (\textit{avg}). Each different algorithm is run 500 times for different starting points using 50 particles and 5000 steps. NMSEs are averaged over the 500 replicates. The second column indicates whether the proximal map is calculated approximately or iteratively, using 40 steps in each iteration. For the uniform prior case, we did not implement the iterative method due to numerical instabilities.}}
\label{table-logistic-comparison-avg}
\centering
\begin{tabular}{llllll}
  \toprule
  Algorithm     & Approx/  & \multicolumn{2}{c}{Laplace} & \multicolumn{2}{c}{Uniform} \\
  \cmidrule(r){3-4}
  \cmidrule(r){5-6}
       &  Iterative &  {\small NMSE last(\%)} & {\small NMSE avg(\%)} & {\small NMSE last(\%)} & {\small NMSE avg(\%)}\\
  \midrule
  \multirow{2}{*}{PPGD} & Approx  & $14.70\pm 4.42$  &   $16.73\pm 0.83$ & $3.63\pm 4.93$& $\mathbf{0.11\pm 0.04}$\\
  & Iterative  & $19.04\pm 1.34$ & $18.66\pm0.60 $ &$-$ & $-$ \\
   \midrule
  \multirow{2}{*}{PIPULA} & Approx  & $12.18\pm1.62$  &   $12.34\pm0.82$ & $4.71\pm 6.02$ & $0.12\pm 0.01$ \\
  & Iterative  & $19.22\pm1.28$ & $18.63\pm 0.79$ & $-$ & $-$\\
   \midrule
 \multirow{2}{*}{MYPGD} & Approx  & $6.09\pm 0.34$  &   $4.94\pm 0.51$& $\mathbf{0.60\pm 0.23}$ &$0.60\pm 0.02$  \\
 & Iterative  & $4.44\pm1.40$  &   $4.33\pm 0.59$& $-$ & $-$  \\
  \midrule
  \multirow{2}{*}{MYIPLA} & Approx  & $4.42\pm 1.32$  & $4.31\pm 0.67$ & $15.26\pm 4.44$& $16.01\pm 2.01$\\
 & Iterative  &  $4.67\pm1.60$ &   $4.45\pm 0.42$& $-$ & $-$  \\
 \midrule
  \multirow{2}{*}{PIPGLA} & Approx  & $2.30\pm0.58$  & $2.45\pm 0.94$ & $6.83\pm 3.97$& $4.22\pm 0.07$\\
   & Iterative  & $\mathbf{2.02\pm0.54}$  & $\mathbf{2.03\pm0.88}$ & $-$& $-$\\
  \bottomrule
\end{tabular}
\label{tab:blr_avg}
\end{table}}

\subsection{Bayesian neural network}\label{ap:experiments_bnn}
\subsubsection{Sparsity Inducing Prior: MNIST}\label{ap:experiments_bnn_sparsity_inducing_prior}
Our setting is equivalent to assuming that the datapoints' labels $l$ are conditionally independent given the features $f$ and network weights $x= (w,v)$, and therefore have the following probability density
\begin{equation*}
    p(l|f,x)\propto \exp\bigg(\sum_{j=1}^{40} v_{lj} \tanh\Big(\sum_{i=1}^{784}w_{ji}f_i\Big)\bigg).
\end{equation*}
We assign priors $p_{\alpha}(w)=\prod_i \text{Laplace}(w_i|0, e^{2\alpha})$ and $p_{\beta}(v)=\prod_i \text{Laplace}(v_i|0, e^{2\beta})$ to the input and output layer's weights, respectively, and learn $\theta = (\alpha, \beta)$ from the data. The model's  density is given by
\begin{equation*}
    p_{\theta}(x,\mathcal{Y}_{\text{train}}) =\prod_i \text{Laplace}(w_i|0, e^{2\alpha}) \prod_j \text{Laplace}(v_j|0, e^{2\beta})\prod_{(f,l)\in\mathcal{Y}_{\text{train}}}p(l|f,x),
\end{equation*}
where $x$ denotes the weight matrices, i.e. $x = (w, v)$.
We note that the log density can be decomposed as
\begin{equation*}
    -\log p_{\theta}(x,\mathcal{Y}_{\text{train}}) = \underbrace{d_w\alpha + \sum_i |w_i|e^{-2\alpha} + d_v\beta+ \sum_j |v_j|e^{-2\beta}}_{g_2(\theta, x)} \underbrace{- \sum_{(f,l)\in\mathcal{Y}_{\text{train}}}\log p(l|f,x)}_{g_1(\theta, x)},
\end{equation*}
where $d_w$ and $d_v$ denote the dimensions of the weights $w$ and $v$, respectively, $g_1$ is differentiable and does not depend on $\theta$ and $g_2$ is proper, convex and lower semi-continuous.
We have derived an approximation to the proximity map of $g_2$ in \ref{ap:proximal_approx_bnn}.
{\normalsize
\paragraph{Dataset.} We use the MNIST dataset. Features are normalised so that each pixel has mean zero and unit standard deviation across the dataset. We split the dataset into 80/20 training and test sets.

\paragraph{Proximal operator of $g_2$.} 

As $g_2$ can be expressed as $g_2(\theta, x) = g_2(\alpha, w) + g_2(\beta, v)$, we can compute their proximal operators separately. It is sufficient to calculate the proximal operator for $g_2(w,\alpha)$ since it is equivalent to that of $g_2(v,\beta)$. 
To do so, we have that
\begin{equation*}
    \prox_{g_2}^{\lambda}(\alpha, w) =\argmin_{(u_0, u)} h(u_0, u), \quad h(u_0, u) = u_0d_w+ \sum_i |u_i| e^{-2u_0} + \Vert (u_0, u)-(\alpha, w)\Vert^2/{(2\lambda)},
\end{equation*}
whose approximate solution is calculated in Section \ref{ap:proximal_approx_bnn}.

\paragraph{Implementation details.}
For the $x$-gradients of $g_1$, we use JAX's grad function (implementing a version of autograd). Plugging the expressions above in the corresponding equations, 
we can implement the proposed algorithms. 
However, due to the high dimensionality of the latent variables, we stabilise the algorithm using the heuristics discussed in Section 2 of \citet{pmlr-v206-kuntz23a}. 
This simply entails dividing the gradients and proximal mapping terms of the updates of $\alpha$ and $\beta$ by $d_w$ and $d_v$. 
We then set $\gamma = 0.05$ and $\lambda = 0.5$ (after performing a grid search) which ensures that the algorithms are not close to losing stability. 
In addition, the weights of the network are initialised according to the assumed prior. 
This is done by setting each weight to $\pm\; a\log u$ where $u\sim \mathcal{U}(0, 1)$, the sign is chosen uniformly at random and $a>0$ is interpreted as the average initial size of the weights. 
\citet{10.1162/neco.1995.7.1.117} suggests setting $a=1/\sqrt{2m}$ for $w$ and $a=1.6/\sqrt{2m}$ for $v$, where $m$ is the fan-in of the destination unit.

\paragraph{Predictive performance metrics.} To allow comparison, we use the same performance metrics as in \citet{pmlr-v206-kuntz23a}. We include their presentation of this metrics for completeness.

Given a new feature vector $\hat{f}$, the posterior predictive distribution for a label $\hat{l}$ associated with the marginal likelihood maximiser $\Bar{\theta}_{\star}$ is given by
\begin{equation*}
    p_{\Bar{\theta}_{\star}}(\hat{l}|\hat{f}, \mathcal{Y}_{\text{train}}) = \int p(\hat{l}|\hat{f},x) p_{\Bar{\theta}_{\star}}(x|\mathcal{Y}_{\text{train}}) \md x.
\end{equation*}
As $p_{\Bar{\theta}_{\star}}(x|\mathcal{Y}_\text{train})$ is unknown, we approximate it with the empirical distribution of the final particle cloud $q = N^{-1}\sum_{i=1}^N \delta_{X_K^i}$, leading to
\begin{equation*}
     p_{\Bar{\theta}_{\star}}(\hat{l}|\hat{f}, \mathcal{Y}_{\text{train}}) \approx \int p(\hat{l}|\hat{f},x) q(\md x) = \frac{1}{N}\sum_{i=1}^N p(\hat{l}|\hat{f},X_K^i) =: g(\hat{l}|\hat{f}).
\end{equation*}
The metrics considered to evaluate the approximation of the predictive power are the average classification error over the test set $\mathcal{Y}_{\text{test}}$, i.e. 
\begin{equation*}
    \text{Error}:=\frac{1}{\vert\mathcal{Y}_{\text{test}}\vert}\sum_{(f,l)\in \mathcal{Y}_{\text{test}}} \mathds{1}\{l=\hat{l}(f)\},\quad \text{where} \;\; \hat{l}(f):= \argmax_{\,\hat{l}} \;g(\hat{l}|\hat{f}),
\end{equation*}
and the log pointwise predictive density (LPPD, \citet{RefWorks:RefID:76-vehtari2017practical})
\begin{equation*}
     \text{LPPD}:= \frac{1}{\vert\mathcal{Y}_{\text{test}}\vert} \sum_{(f,l)\in \mathcal{Y}_{\text{test}}} \log(g(l|f)).
\end{equation*}
Under the assumption that data is drawn independently from $p(l, f)$, we have the following approximation for large test data sets,
\begin{align*}
    \text{LPPD}&\approx \int \log(g(l|f)) p(\md l, \md f) = \int \bigg[\int \log\Big(\frac{g(l|f)}{p(l|f)}\Big) p(\md l| \md f)\bigg] p(\md f) + \int \log(p(l|f)) p(\md l, \md f) \\
    &= -\int \text{KL}(g(\cdot|f)\Vert p(\cdot|f)) p(\md f) + \int \log(p(l|f)) p(\md l, \md f).
\end{align*}
This means that the larger the LPPD is, the smaller the mean KL divergence between our classifier $g(l|f)$ and the optimal classifier $p(l|f)$.

\paragraph{Results.} First, it is important to discuss whether the Laplace prior is more appropriate in this setting than the Normal one.
\citet{jaynes_prior} provides two reasons why the Laplace prior is particularly suitable for Bayesian neural network models. Firstly, for any feedforward network there is a functionally equivalent network in which the weight of a non-direct connection has the same size but opposite sign, therefore consistency demands that the prior for a given weight $w$ is a function of $\vert w\vert$ alone. Secondly, if it is assumed that all that is known about $\vert w\vert$ is its scale, and that the scale of a positive quantity is determined by its mean rather than some higher order moment, then the maximum entropy distribution for a positive quantity constrained to a given mean is the exponential distribution. It would follow that the signed weight $w$ has a Laplace density \citep{10.1162/neco.1995.7.1.117}.

We have examined the sparsity-inducing nature of the Laplace prior versus a normal one in Figure \ref{fig:histogram_weights} and Table \ref{table-bnn-comparison}.
As mentioned in the main text, the sparse representation of our experiment also has the advantage of producing models that are smaller in terms of memory usage when small weights are zeroed out.
To investigate this, we set to zero all weights below a certain threshold and analyse the performance of the compressed weight matrices. 
We consider two cases, averaging the particles of the final cloud $X_{500}^1, \dots, X_{500}^{100}$, applying the threshold and then calculating the performance, and secondly, setting to zero small values of each particle of the cloud and averaging the performance of each particle. 
We compare the results for the Bayesian neural networks with Laplace and Normal priors (Table \ref{table-bnn-sparsity}). 
It is important to note that when applying the same threshold to both cases, the Laplace prior leads to a very compressed weight matrix compared to the Normal prior, i.e. there is a significant difference in the percentage of weights set to zero. We observe that when setting the same proportion of weights to zero in both layers, the performance of the BNN with Laplace priors is better in terms of the log pointwise predictive density than that of the BNN with Normal priors, especially when averaging the final cloud of particles before computing the performance.

\begin{table}
\caption{Bayesian neural network. Performance of BNN with Laplace (implemented using MYIPLA) and Normal priors (implementation with PGD) when setting weights from the final particle cloud below a certain threshold to zero. The second column refers to whether the particles are averaged before ($\checkmark$) or after (\scalebox{0.75}{$\ballotx$}) calculating the performance.}
\label{table-bnn-sparsity}
\centering
\begin{tabular}{lcllllll}
  \toprule
  Prior   & Average over  & \multicolumn{2}{c}{$\%$ of zero weights} & \multicolumn{2}{c}{Thresholds} & Error ($\%$)     &  LPPD \\
  \cmidrule(r){3-4}
  \cmidrule(r){5-6}
     & particles? & Layer 1& Layer 2 & Layer 1& Layer 2 &      &   \\
  \midrule
  \multirow{2}{*}{Laplace} & \checkmark & $74$ & $48$  & $0.2$ & $0.2$ &$\mathbf{7.0}$& $\mathbf{-0.23}$\\
  &\scalebox{0.75}{$\ballotx$}  & $56$ & $35$ &   $1$ & $1$& $\mathbf{1.5}$ & $\mathbf{-0.07}$\\
  \midrule  
  \multirow{4}{*}{Normal} & \checkmark  & $74$ & $48$  &   $0.5$ & $1.1$& $15$& $-0.74$\\
      &   \checkmark &$16$ & $15$ &  $0.2$ & $0.2$ & $16$ &$-0.78$\\
  &   \scalebox{0.75}{$\ballotx$} &$56$ & $35$ &  $7$ & $4$ & $2.0$ &$-0.11$\\
  &   \scalebox{0.75}{$\ballotx$} &$8.6$ & $7.1$ &  $1$ & $1$ & $1.5$ &$-0.10$\\    
  \bottomrule
\end{tabular}
\end{table}
Figure \ref{fig:weight_analysis} shows how the performance metrics evolve when weights below a certain threshold are set to zero, when particles are averaged before (\ref{fig:weight_analysis_average_over_particles}) or after (\ref{fig:weight_analysis_not_averaged}) computing the performance for MYIPLA.

Once we have set the weights of the matrix below a certain threshold to zero, it is necessary to explore the dead units. These are hidden units all of whose input or output weights are zero \citep{10.1162/neco.1995.7.1.117}. In both cases, the unit is redundant and it can be eliminated to obtain a functionally equivalent network architecture, we will called this new effective weight matrix $w_{\text{pruned}}$.
The occupancy ratio of a weight matrix $w$ \citep{MARINO2023152} is defined as $\psi=\text{size}(w_{\text{pruned}})/\text{size}(w)$, where $\text{size}$ denotes the memory size. The inverse of $\psi$ is the compression ratio.
We compute the occupancy ratio of the weight matrix for both the hidden and output layer for different values of the pruning threshold. We do this for each particle of the final cloud and obtain the average as well as for the averaged final particle cloud, results are shown in Figure \ref{fig:memory_analysis}.
\vspace{30pt}

\begin{figure}[h]
  \centering
    \begin{subfigure}[b]{0.25
    \textwidth}
        \centering
        \includegraphics[width=\textwidth]{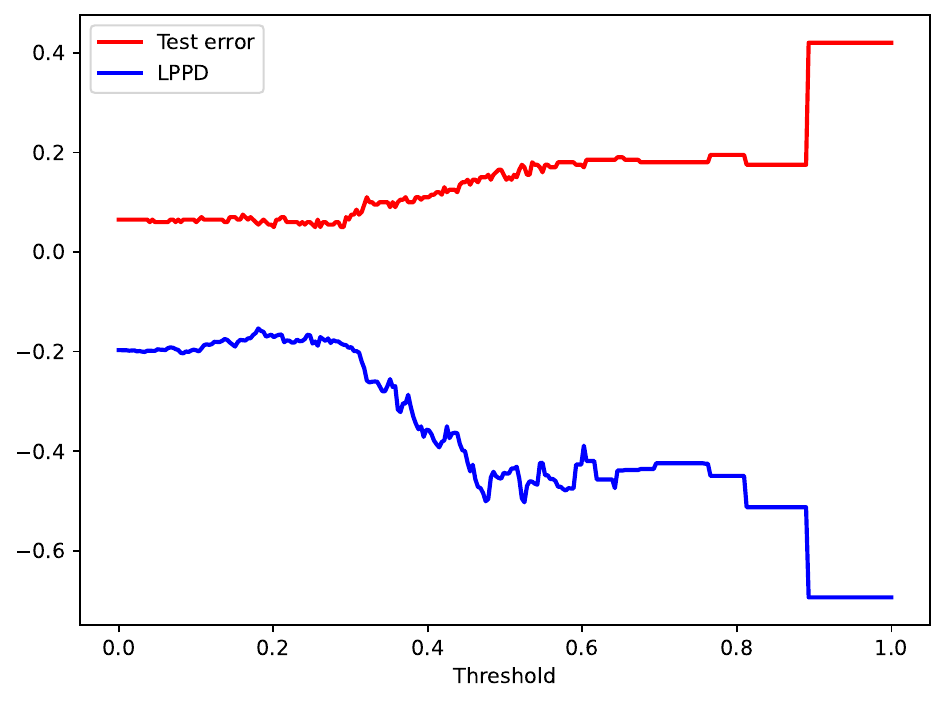}
        \caption{}
        \label{fig:weight_analysis_average_over_particles}
    \end{subfigure}
    \hspace{30pt}
        \begin{subfigure}[b]{0.25\textwidth}
        \centering
        \includegraphics[width=\textwidth]{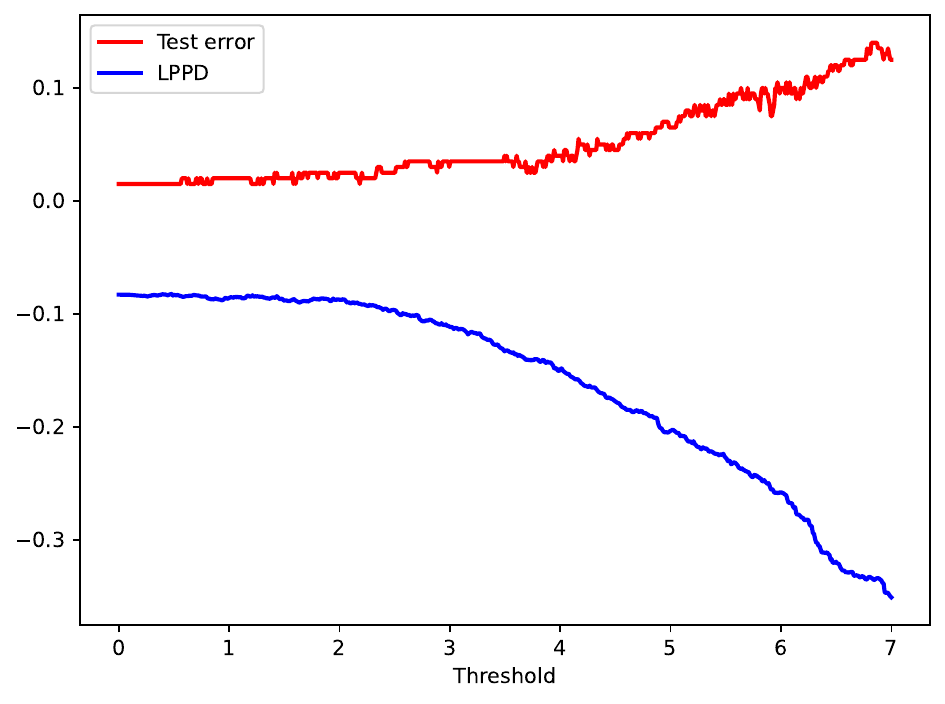} 
        \caption{}
        \label{fig:weight_analysis_not_averaged}
    \end{subfigure}
    \hfill
  \caption{{\normalsize Evolution of the performance metrics when weights below a certain threshold are set to zero, when particles are averaged before (a) or after (b) computing the performance.}}  
  \label{fig:weight_analysis}
\end{figure}

\begin{figure}[h!]
  \centering
    \begin{subfigure}[b]{0.25\textwidth}
        \centering
        \includegraphics[width=\textwidth]{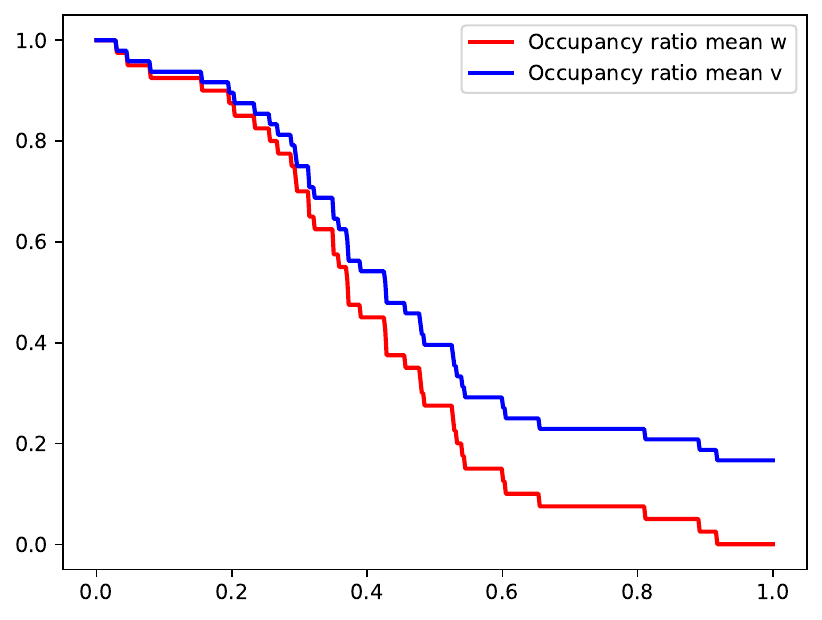}
        \caption{}
        \label{fig:memory_analysis_average_over_particles}
    \end{subfigure}
    \hspace{30pt}
        \begin{subfigure}[b]{0.25\textwidth}
        \centering
        \includegraphics[width=\textwidth]{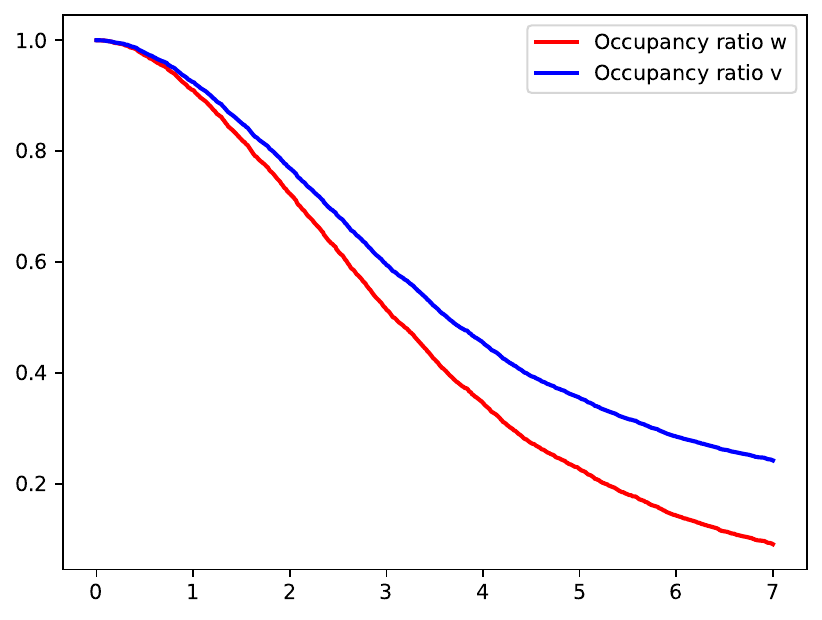} 
        \caption{}
        \label{fig:memory_analysis_not_averaged}
    \end{subfigure}
    \hfill
  \caption{{\normalsize Occupancy ratio for the weights matrices of the hidden and output layers as a function of the pruning threshold, when particles are averaged before (a) or after (b) computing the occupancy ratio.}}  
  \label{fig:memory_analysis}
\end{figure}}

\subsubsection{Sparsity Inducing Prior: CIFAR10}\label{ap:experiments_bnn_addtional_dataset}
We further evaluate our methods on  a classification task using a more complex dataset: CIFAR10.
As with the MNIST dataset, to reduce the cost of computing the gradients on a big dataset, we subsample 5000 data points with labels \emph{plane}, \emph{car}, \emph{ship} and \emph{truck}. 

Given that the data consists of colour images, we employ a convolutional neural network (CNN) architecture. Specifically, we use a combination of convolutional layers, max pooling layers, and linear layers with non-linear activation functions. For simplicity, we apply a sparsity-inducing prior only to the linear layers, and not to the convolutional ones. The sparsity inducing prior for each layer with weight matrix $w$ is given by $p_{\alpha}(w)=\prod_i \text{Laplace}(w_i|0, e^{2\alpha})$ where $\alpha$ is learn from the data. The network structure and layer dimensions are as follows.

\begin{itemize}
    \item  Convolutional layer (Deterministic): Conv2d(3, 6, 5)
\item Max pooling layer (Deterministic): MaxPool2d(2, 2)
\item Convolutional layer (Deterministic): Conv2d(6, 16, 5)
\item Linear layer with sparsity inducing prior + SELU activation function: Linear(16 $\times$ 5 $\times$ 5, 512)
\item Linear layer with sparsity inducing prior + SELU activation function: Linear(512, 256)
\item Linear layer with sparsity inducing prior + SELU activation function: Linear(256, 128)
\item Linear layer with sparsity inducing prior: Linear(128, 4)
\end{itemize}

Table \ref{tab:classification_results_cifar10} presents quantitative results for the variance of the weights and error metrics. The last column provides a measure of the sparsity-inducing effect of the Laplace prior on the linear layers.

\begin{table}[h]
\caption{Bayesian neural network on CIFAR10 dataset. Test errors and log pointwise predictive density (LPPD) achieved using the final particle cloud with $N = 50$. Computation times and standard deviation of the empirical distribution of the weight matrix $w$ for linear layers are also provided.}
\centering
\begin{tabular}{lcccc}
\toprule
{Algorithm} & {Error (\%)} & {LPPD ($\times 10^{-1}$)} & {Time (s)} & {Std. $w$} \\
\midrule
{MYPGD}   & $5.27\pm0.95$      & $-4.41\pm0.38$      & $201$     & $3.10 $      \\
{MYIPLA}  & $\mathbf{5.23\pm1.31}$ & $-5.05\pm0.45$      & $199$     & $3.22$       \\
{PIPGLA}  & $5.39\pm1.02$      & $\mathbf{-4.32\pm0.37}$ & $295$     & $\mathbf{2.85}$ \\
{PGD}     & $6.01\pm1.15$      & $-5.73\pm0.40$      & $\mathbf{178}$ & $11.51$      \\
{SOUL}    & $9.11\pm2.03$      & $-7.68\pm1.56$     & $433$     & $15.68$     \\
{IPLA}    & $5.40\pm1.33$      & $-5.90\pm0.75$      & $181$     & $15.73$     \\
\bottomrule
\end{tabular}
\label{tab:classification_results_cifar10}
\end{table}

\subsubsection{Non-Differentiable Activation Functions}\label{ap:experiments_bnn_activation_function}
During gradient checking in neural network training, a potential source of inaccuracy arises from the presence of non-differentiable points in the objective function \citep{kumar2024gddoesntmakecut}. These non-smooth points often result from the use of activation functions such as the Rectified Linear Unit (ReLU), defined as $\max(0, x)$, as well as from the hinge loss in support vector machines, maxout neurons, among others. To give a concrete example, consider the ReLU activation function and $x<0$ but very close to $0$. The analytic gradient evaluated at $x$ is equal to $0$. However, the numerical gradient can be non-zero when using a finite difference approximation in case $x+h>0$. 

Our method provides a principled way of dealing with these non-differentiable points. 
To illustrate this, we present a simple example similar to the one in the previous section. Here, we consider a Bayesian neural network with a Normal prior distribution on the weights to classify MNIST digits 1 and 7, instead of 4 and 9. 
Additionally, we use a linear approximation of $tanh$ as the activation function to mitigate the dying neuron problem associated with ReLU \citep{dying_neuron_2020}, while noting that it still remains non-differentiable. This linear approximation is defined as
\begin{equation*}
    h(x) = \begin{cases}
    -1 \quad  \text{if} \; x<-1,\\
     x \quad \ \ \;\text{if} \; x\in[-1, 1],\\
    1 \quad \ \;\;\text{if} \; x>1.
    \end{cases}
\end{equation*}
Furthermore, we can compute the proximal mapping of $h$ which is given by
\begin{equation}\label{eq:proximity_map_leaky_relu}
    \prox_{h}^\lambda(x) = \begin{cases}
        x &\text{if} \; x<-1,\\
        -1  &\text{if} \; x\in[-1, -1+\lambda],\\
        x-\lambda   & \text{for} \; x\in[-1+\lambda,1-\lambda],\\
        1 & \text{for} \; x\in[1-\lambda, 1],\\
        x&\text{if} \; x> 1,
    \end{cases}
\end{equation}
where we have applied the first order optimality condition and used the subgradient of the function at $x=-1$ and $1$ which is given by the sets $[0, 1]$ and $[-1, 0]$, respectively.
Therefore, in this setting we have the following likelihood
\begin{equation*}
    p(l|f,x)\propto \exp\bigg(\sum_{j=1}^{40} v_{lj} h_{LR}\Big(\sum_{i=1}^{784}w_{ji}f_i\Big)\bigg).
\end{equation*}
We assign priors $p_{\alpha}(w)=\prod_i \mathcal{N}(w_i|0, e^{2\alpha})$ and $p_{\beta}(v)=\prod_i \mathcal{N}(v_i|0, e^{2\beta})$ to the input and output layer's weights, respectively, and learn $\theta = (\alpha, \beta)$ from the data. Hence, model's  density is given by
\begin{equation*}
    p_{\theta}(x,\mathcal{Y}_{\text{train}}) =\prod_i \mathcal{N}(w_i|0, e^{2\alpha}) \prod_j \mathcal{N}(v_j|0, e^{2\beta})\prod_{(f,l)\in\mathcal{Y}_{\text{train}}}p(l|f,x),
\end{equation*}
where $x$ denotes the weight matrices, i.e. $x = (w, v)$.
We note that the log density can be decomposed as
\begin{equation*}
    -\log p_{\theta}(x,\mathcal{Y}_{\text{train}}) = \underbrace{d_w\alpha + \frac{1}{2}\sum_i |w_i|^2e^{-2\alpha} + d_v\beta+ \frac{1}{2}\sum_j |v_j|^2e^{-2\beta}}_{g_1(\theta, x)} \underbrace{- \sum_{(f,l)\in\mathcal{Y}_{\text{train}}}\log p(l|f,x)}_{g_2(\theta, x)},
\end{equation*}
where $d_w$ and $d_v$ denote the dimensions of the weights $w$ and $v$, respectively, $g_1$ is differentiable and depends on $\theta$ and $x$, while $g_2$ is proper, convex and lower semi-continuous and only depends on $x$, that is, $g_2(\theta, x) = g_2(x)$. 
As a result, the non-differentiability affects only the latent variables $x$. In this case, we can compute the proximity map of $g_2$ by using the expression for the proximity map of the activation function, provided in \eqref{eq:proximity_map_leaky_relu}.

We follow the same implementation details as outlined in the previous section and use the same performance metrics: average classification error over a test set and log pointwise predictive density.
The results for the proposed proximal algorithms are provided in Table \ref{table-bnn-comparison-nondiff-activation} together with the computation times for $N=50$ and $500$ iterations. In addition, plots of the evolution of the different performance metrics for different number of particles are shown in Figure \ref{fig:bnn_non_diff_activation_performance}.
We observe that the standard deviation of the LPPD across runs decreases as the number of particles increases. Moreover, PIPGLA exhibits a lower standard deviation compared to the other methods.

\begin{table}[t]
  \caption{Bayesian neural network with non-differentiable activation function. Test errors and log pointwise predictive density (LPPD) achieved using the final particle cloud with $N = 50$ and $500$ iterations.}
  \label{table-bnn-comparison-nondiff-activation}
  \centering
  \begin{tabular}{llll}
    \toprule
    Algorithm     & Error $\left(\%\right)$     &  LPPD $\left(\times 10^{-2}\right)$ & Times (s)\\
    \midrule
    MYPGD & $0.75\pm 0.68$  & $\mathbf{-3.36\pm1.18}$  &   $\mathbf{40}$\\
    MYIPLA & $\mathbf{0.70\pm 0.50}$  & $-4.28\pm 2.86$ & $\mathbf{40}$\\
        PIPGLA & $0.90\pm 0.49$  & $-3.76\pm0.96$  &   $68$\\
    \bottomrule
  \end{tabular}
\end{table}
\normalsize

\begin{figure}[t]
    \centering

    \begin{subfigure}[b]{0.33\textwidth}
        \centering
        \includegraphics[width=\textwidth]{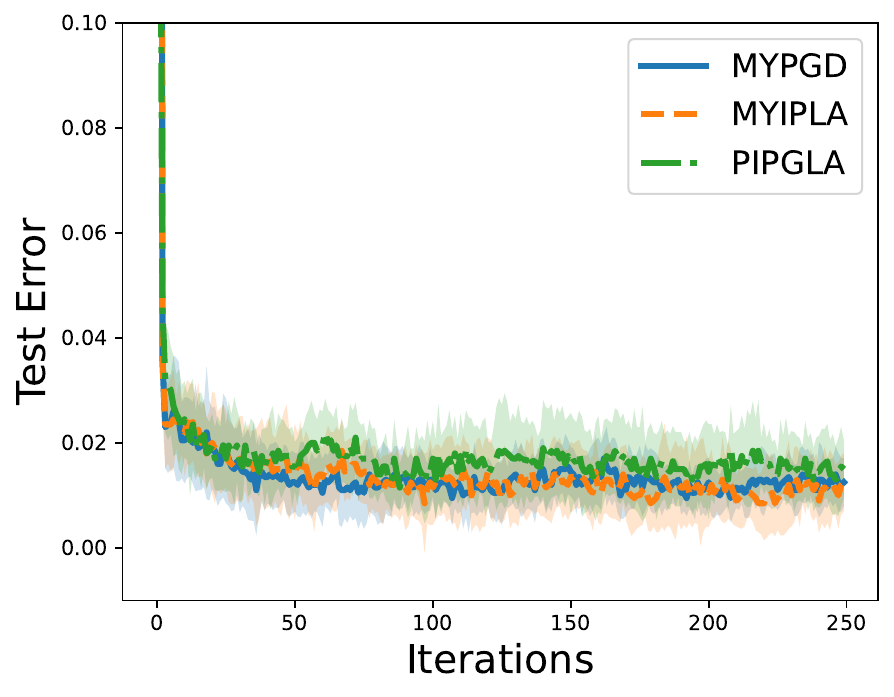}
        \caption{Test error. $N=5$}
        \label{fig:subfig_11_error_label}
    \end{subfigure}
        \begin{subfigure}[b]{0.33\textwidth}
        \centering
        \includegraphics[width=\textwidth]{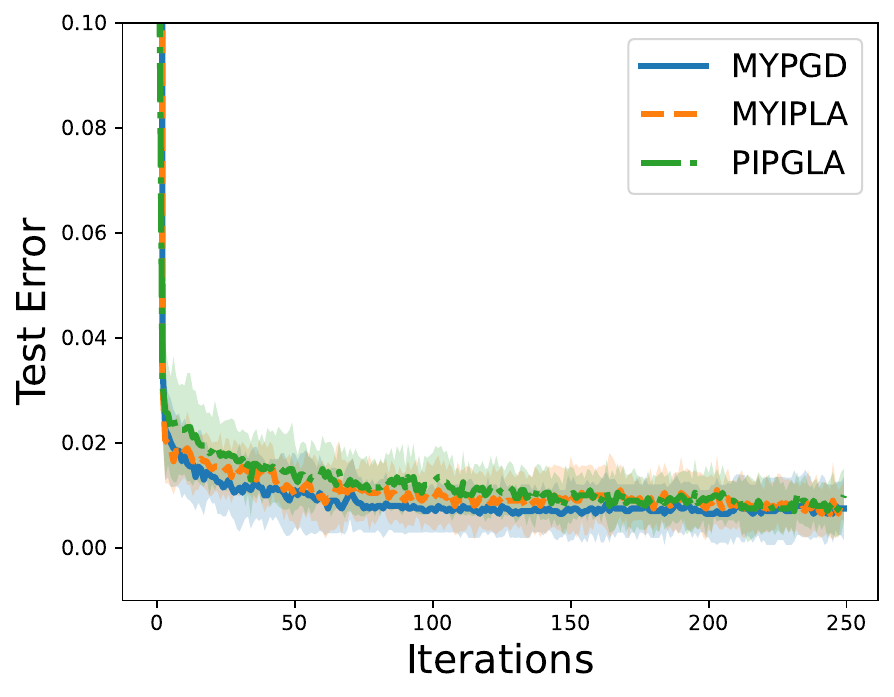}
        \caption{Test error. $N=50$}
        \label{fig:subfig_33_error_label}
    \end{subfigure}
        \begin{subfigure}[b]{0.33\textwidth}
        \centering
        \includegraphics[width=\textwidth]{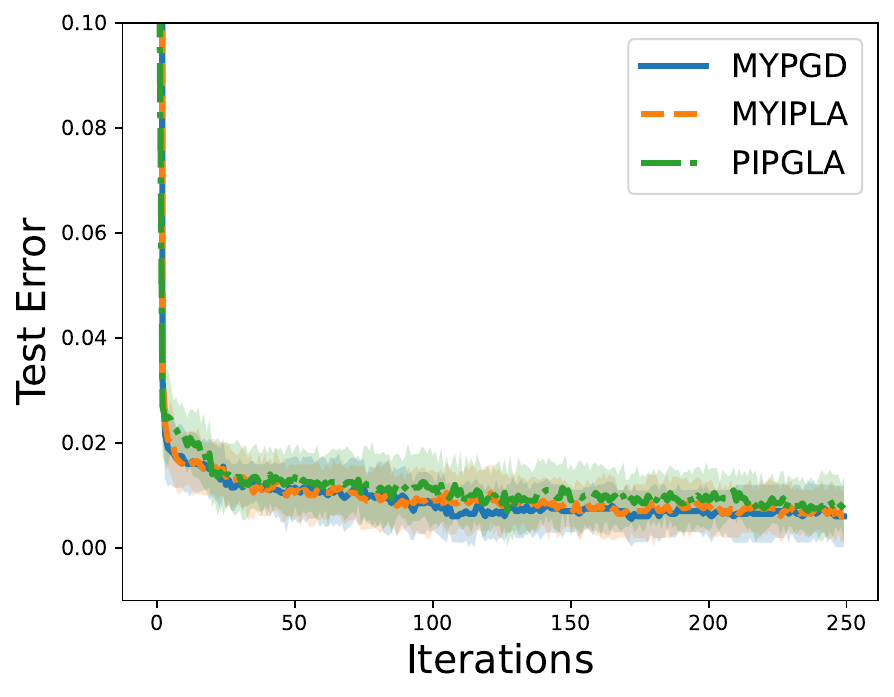}
        \caption{Test error. $N=100$}
        \label{fig:subfig_33_error_label_100}
    \end{subfigure}
    \vfill
    \begin{subfigure}[b]{0.33\textwidth}
        \centering
        \includegraphics[width=\textwidth]{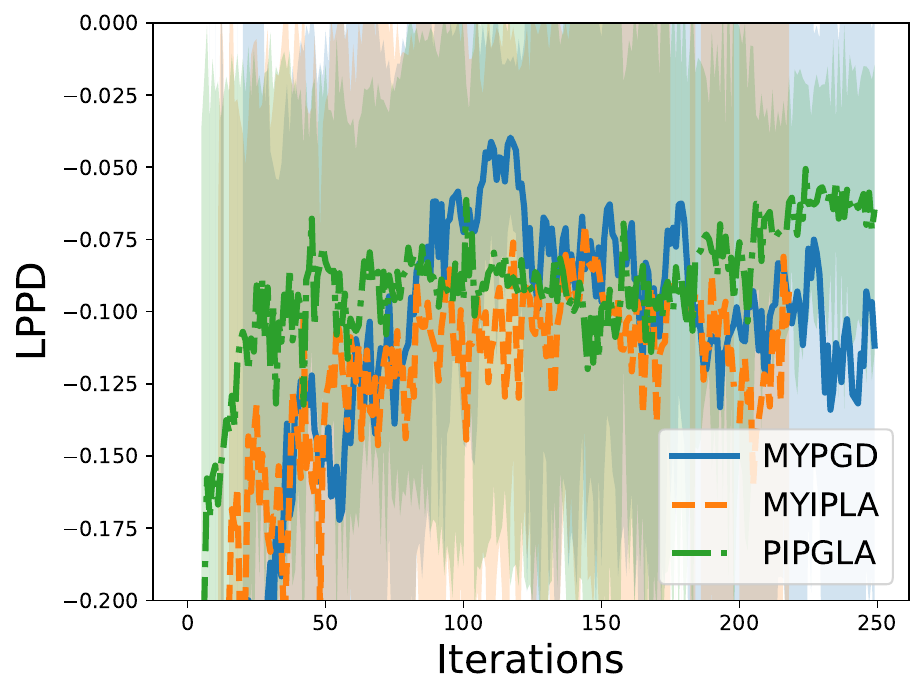}
        \caption{LPPD. $N=5$}
        \label{fig:subfig_22_lppd}
    \end{subfigure}
    \begin{subfigure}[b]{0.33\textwidth}
        \centering
        \includegraphics[width=\textwidth]{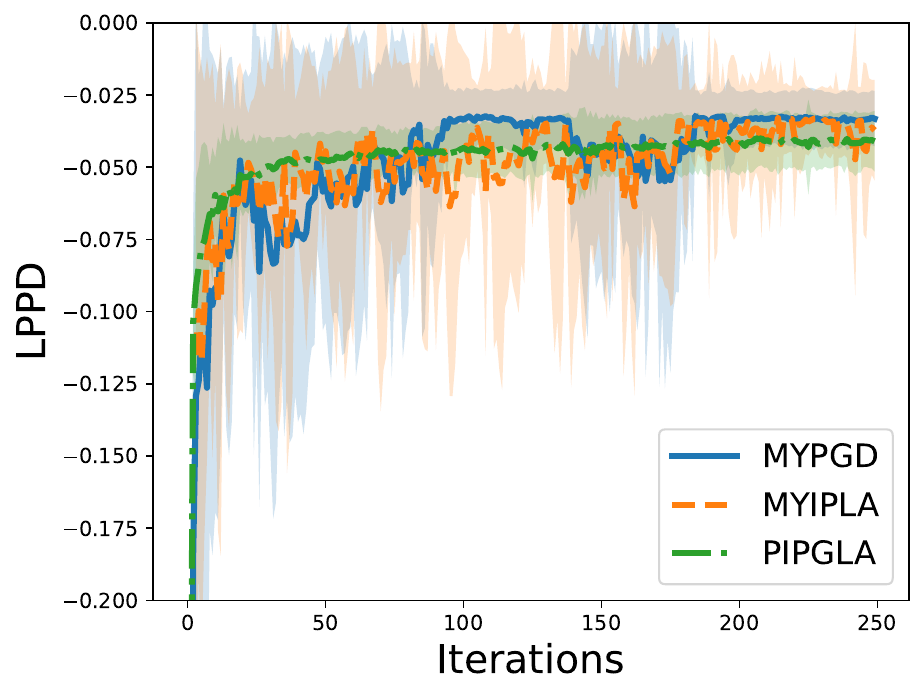}
        \caption{LPPD. $N=50$}
        \label{fig:subfig_44_lppd}
    \end{subfigure}
    \begin{subfigure}[b]{0.33\textwidth}
        \centering
        \includegraphics[width=\textwidth]{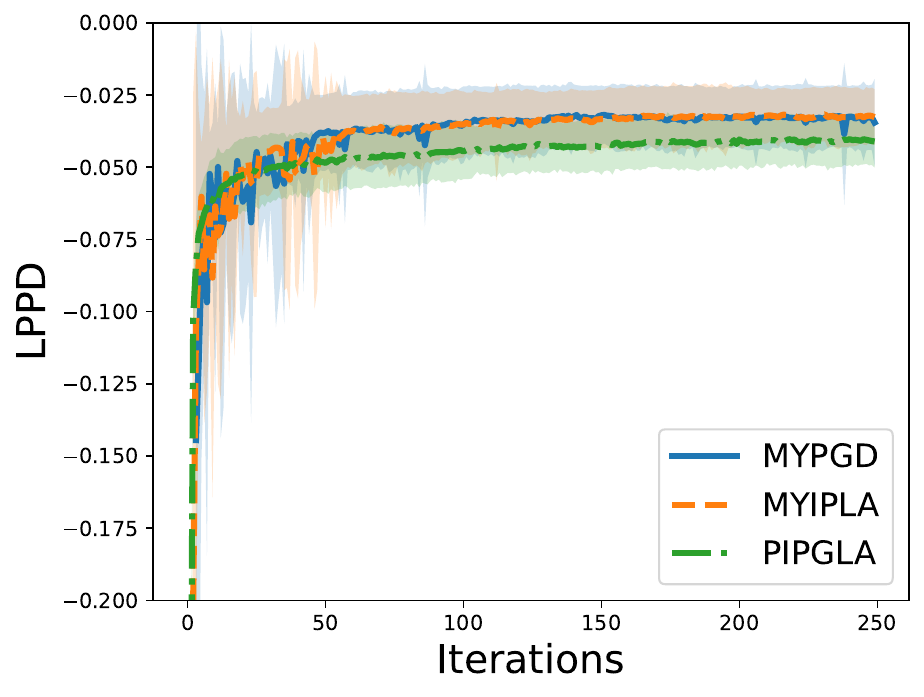}
        \caption{LPPD. $N=100$}
        \label{fig:subfig_44_lppd_100}
    \end{subfigure}

    \caption{Evolution of the classification error on a test set (top) and the log pointwise predictive density (LPPD) (bottom) over iterations in the BNN experiment with non-differentiable activation function, using $250$ iterations. Values averaged over 100 runs.}
    \label{fig:bnn_non_diff_activation_performance}
\end{figure}

\subsection{Image Deblurring}\label{app:image_deblurring}
We consider the problem of recovering a high-quality image from a blurred and noisy observation $y = Hx +\varepsilon$, where $H$ is a blurring operator that blurs a pixel $x_{i,j}$ uniformly with its closest neighbours (10 $\times$ 10 patch), and $\varepsilon\sim\mathcal{N}(0, \sigma^2 I)$. 
The log prior is proportional to the total variation defined as $ TV(x)=\Vert \nabla_d x\Vert_{1}$, where $\Vert \cdot\Vert_{1}$ is the $\ell_1$ norm and $\nabla_d$ is the two-dimensional discrete gradient operator, which is non-differentiable. 
The proportionality parameter, $e^\theta$, which controls the strength of this log prior, typically requires manual tuning.
Instead of fixing this parameter manually, we estimate its optimal value using our proposed algorithms. Note that we exponentiate $\theta$ to ensure its positivity.

The posterior distribution for the model takes the form
\begin{equation*}
    p_\theta(y|x)\propto\exp\left(-\Vert y- Hx\Vert^2/(2\sigma^2)-e^\theta TV(x) + \log C(\theta)\right),
\end{equation*}
where $C(\theta)$ is proportional to the normalising constant of the prior distribution. 
To compute $C(\theta)$, we start by considering the case when $\theta = 0$. In this case, the total variation prior is given by
\begin{equation*}
    p(x)= C \exp(-TV(x)),
\end{equation*}
where $C$ is constant. For $\theta\neq 0$, the prior $p_\theta(x)$ can be expressed using the pushforward measure as
\begin{equation*}
    p_\theta(x) = T_{e^\theta}\#p(x) = e^{d_x\theta}p(e^\theta x),
\end{equation*}
where $T_{e^\theta}\#$ denotes the pushforward operator and $d_x$ is the dimension of $x$. Due to the linearity of the total variation norm, it follows that 
\begin{equation*}
    p_\theta(x)=C e^{d_x\theta} \exp\left(-e^\theta TV(x)\right).
\end{equation*}
Thus, we obtain that theta $C(\theta) = e^{d_x\theta}$.
For the experiments, we employ the algorithms proposed by \citep{douglas_56} and \citep{chambolle2004algorithm} to efficiently compute the proximal operator of the total variation norm.
Due to the difficulty of computing the joint proximal operator over the parameter $\theta$ and latent variables $x$, we have consider hybrid versions of the algorithms, which use standard gradient-based updates for the parameters and proximal updates for the particles. 
That is, the updates for the hybrid MYIPLA algorithm are given by
\begin{align*}
    \theta_{n+1}^N =& \theta_{n}^N  - \frac{\gamma}{d_xN}\sum_{i=1}^N e^{\theta_n^N} TV( X_n^{i, N})+\gamma+\sqrt{\frac{2\gamma}{N}}\xi_{n+1}^{0, N},\\
    X_{n+1}^{i, N} =& \Big(1-\frac{\gamma}{\lambda}\Big)X_{n}^{i, N} -\gamma \frac{H^\intercal(HX_n^{i, N} -y)}{\sigma^2} + \frac{\gamma}{\lambda}\prox_{e^{\theta_n^N} TV}^{\lambda}(X_n^{i, N})+\sqrt{2\gamma}\;\xi_{n+1}^{i, N}.
\end{align*}
Note that as in the Bayesian neural network example, we apply the heuristic of dividing the gradient term in the $\theta$ updates by $d_x$ for numerical stability.
For PIPGLA, the update for the parameter $\theta$ remains the same, while the updates for the particles are of the form
\begin{align*}
    X_{n+1/2}^{i, N} &= X_{n}^{i, N} -\gamma \frac{H^\intercal(HX_n^{i, N} -y)}{\sigma^2} +\sqrt{2\gamma}\;\xi_{n+1}^{i, N},\\
    X_{n+1}^{i, N} &= \prox_{e^{\theta_{n+1}^ {N}}TV}^{\lambda} \big(X_{n+1/2}^{i, N}\big).
\end{align*}
Analogous forms are defined for the proximal PGD algorithms.

\paragraph{Dataset.} We use black and white images with pixels values ranging from 0 to 255. The dimensions of the acoustic guitar image are $d_x = n_1\times n_2 = 584\times 238$, while the dimensions of the boat image (a standard benchmark in the image reconstruction literature) are $d_x = 512\times512$.

\paragraph{Implementation details.} 
We implement the proximal operator of the total variation using the \textit{proxTV} Python package \citep{2011_Barbero11_icml, 2018_barbero_jmlr}. Specifically, we employ the Douglas-Rachford method introduced by \citet{douglas_56} and the Chambolle-Pock method \citep{chambolle2004algorithm}. The Douglas-Rachford method is significantly faster than the Chambolle-Pock method.
It is important to note that increasing the precision of the Moreau-Yosida envelope significantly slows down the computation of the proximal operator when using these numerical schemes.
We set $\gamma = 0.01$, and $\lambda = 0.4$ for MYPGD and MYIPLA and $\lambda = 0.001$ for PIPGLA (after performing a grid search) which ensures that the algorithms are not close to losing stability. 
In addition, the pixels of the initial particles are drawn from a normal distribution with mean $\mu = 50$ and scale parameter $10$, while the initial parameter estimate $\theta_0$ is sampled from a uniform distribution over $[-15, 10]$.

\paragraph{Performance metrics.} 
To evaluate the performance of our algorithms in image reconstruction, we evaluate the mean squared error (MSE) and the structural similarity index (SSIM) between the particle cloud and the ground-truth image. The SSIM quantifies image quality by comparing luminance, contrast, and structural details.

\begin{figure}[t]
    \centering

    \begin{subfigure}[b]{0.18\textwidth}
        \centering
        \includegraphics[width=\textwidth]{plots/original_reshaped.jpg}
        \caption{Original}
        \label{fig:subfig_11}
    \end{subfigure}
    \hfill
    \begin{subfigure}[b]{0.18\textwidth}
        \centering
        \includegraphics[width=\textwidth]{plots/blurred.jpg}
        \caption{Blurred}
        \label{fig:subfig_22}
    \end{subfigure}
    \hfill
    \begin{subfigure}[b]{0.18\textwidth}
        \centering
        \includegraphics[width=\textwidth]{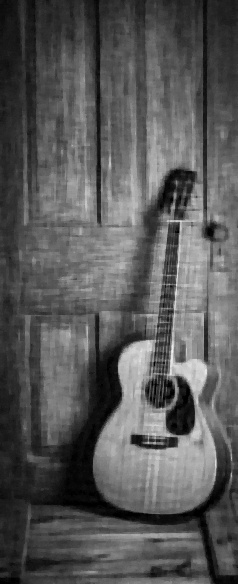}
        \caption{MYPGD}
        \label{fig:subfig_33}
    \end{subfigure}
    \hfill
    \begin{subfigure}[b]{0.18\textwidth}
        \centering
        \includegraphics[width=\textwidth]{plots/myipla_reconstructed_2.jpg}
        \caption{MYIPLA}
        \label{fig:subfig_44}
    \end{subfigure}
    \hfill
    \begin{subfigure}[b]{0.18\textwidth}
        \centering
        \includegraphics[width=\textwidth]{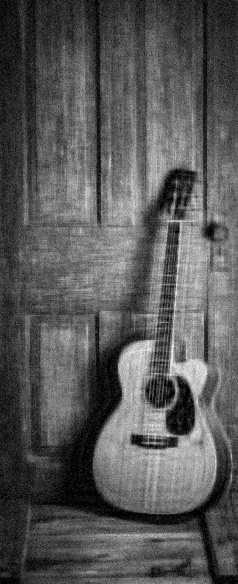}
        \caption{PIPGLA}
        \label{fig:subfig_55}
    \end{subfigure}

    \caption{Image deblurring experiment. All the algorithms use $N=10$ particles and are run for 3000 iterations with a burn-in of 100 iterations.}
    \label{fig:image_deconvolution_methods}
\end{figure}

\begin{figure}[h!]
    \centering

    \begin{subfigure}[b]{0.33\textwidth}
        \centering
        \includegraphics[width=\textwidth]{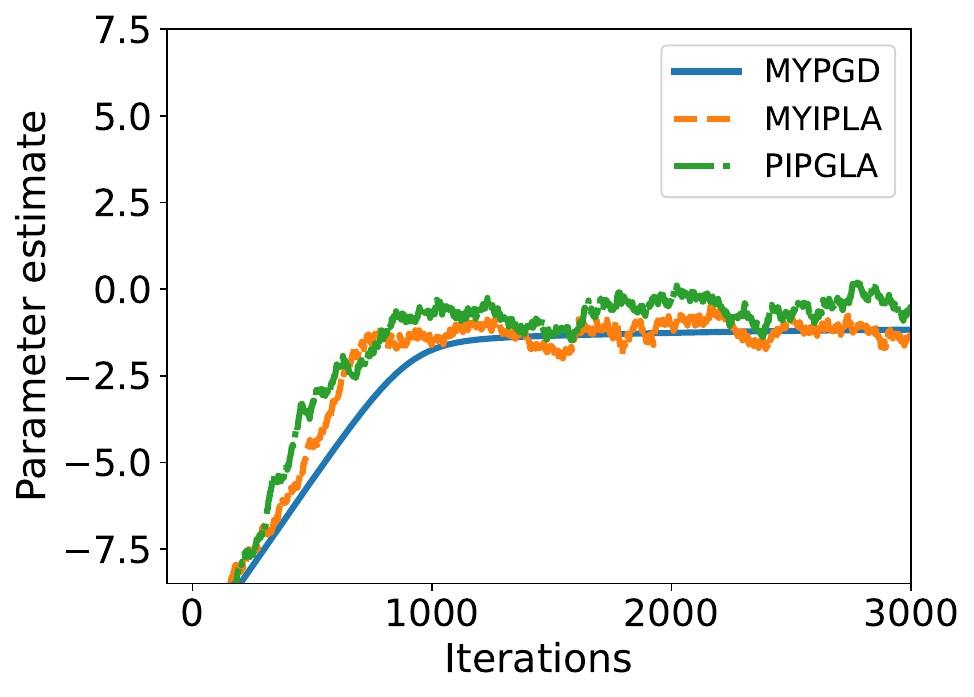}
        \caption{$\theta$ estimates}
        \label{fig:subfig_11_estimates}
    \end{subfigure}
    \begin{subfigure}[b]{0.33\textwidth}
        \centering
        \includegraphics[width=\textwidth]{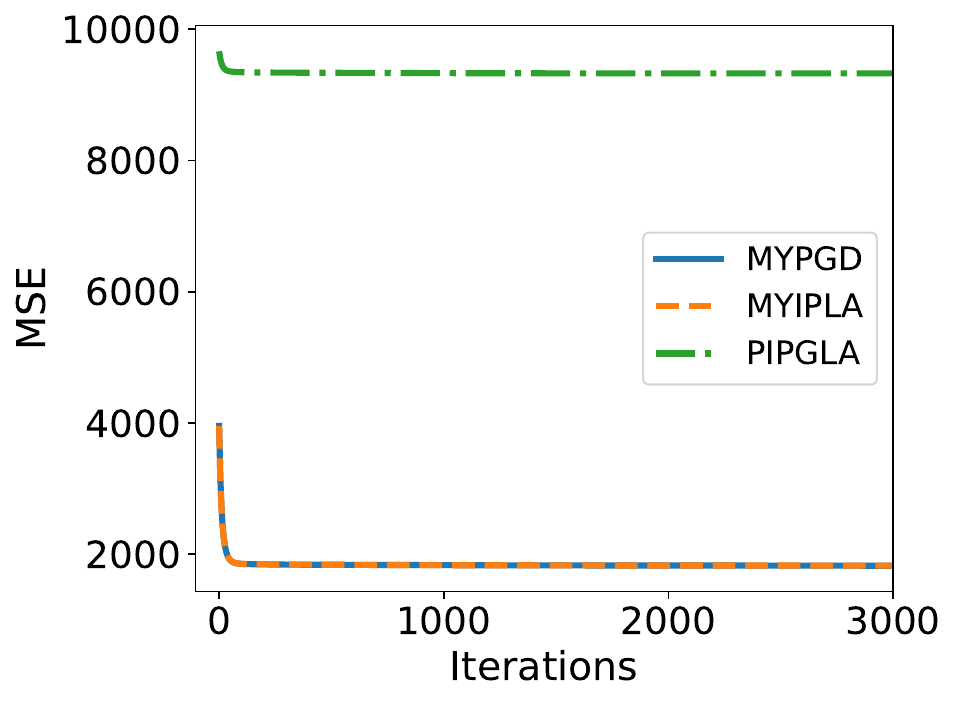}
        \caption{MSE}
        \label{fig:subfig_22_nmse}
    \end{subfigure}
    \begin{subfigure}[b]{0.33\textwidth}
        \centering
        \includegraphics[width=\textwidth]{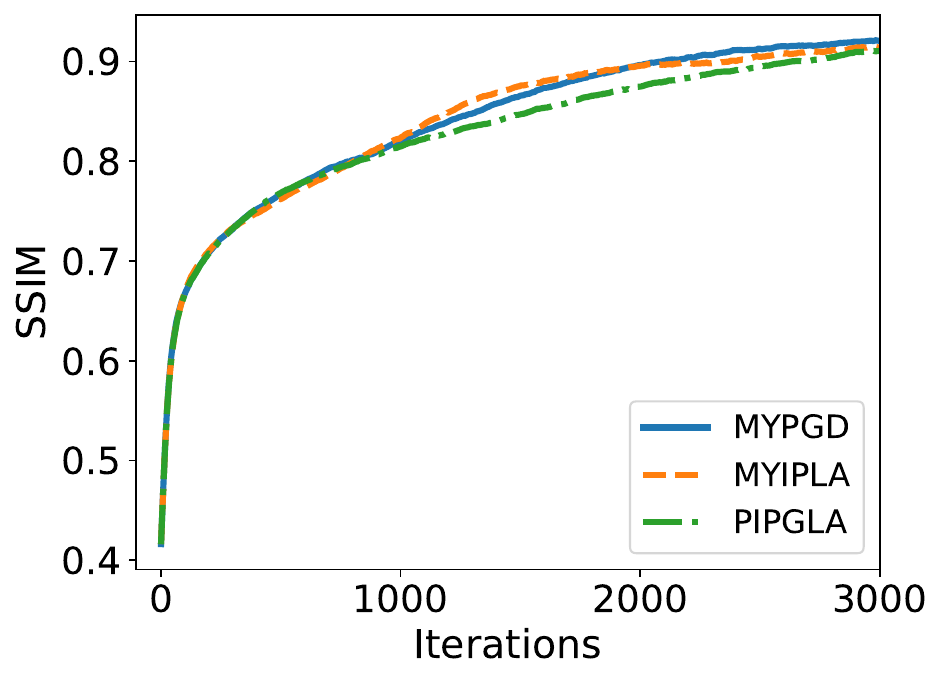}
        \caption{Structural similarity index (SSIM)}
        \label{fig:subfig_22_ssim}
    \end{subfigure}

    \caption{Evolution of different quantities over iterations in the image deblurring experiment with $N=10$ particles for the acoustic guitar image. The plots are shown after discarding a burn-in period of 100 iterations and the initial parameter is $\theta_0 = -10.5$.}
    \label{fig:image_deconvolution_methods_analysis}
\end{figure}

\begin{figure}[t]
    \centering

    \begin{subfigure}[b]{0.18\textwidth}
        \centering
        \includegraphics[width=\textwidth]{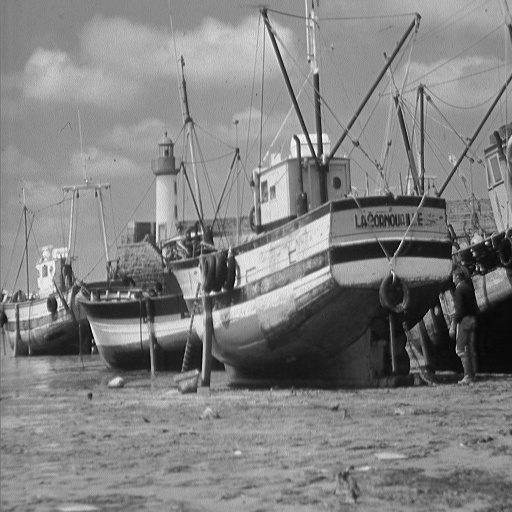}
        \caption{Original}
        \label{fig:subfig_11_boat}
    \end{subfigure}
    \hfill
    \begin{subfigure}[b]{0.18\textwidth}
        \centering
        \includegraphics[width=\textwidth]{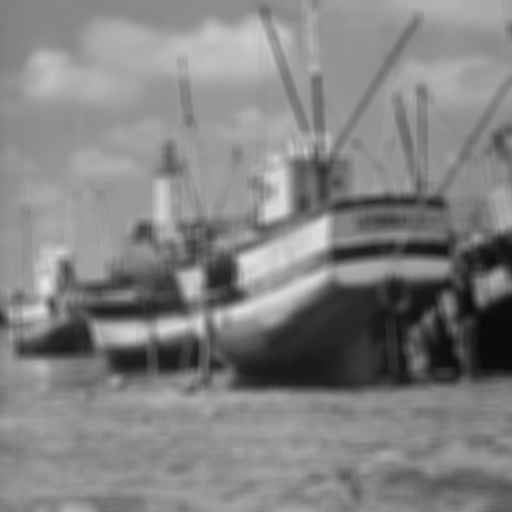}
        \caption{Blurred}
        \label{fig:subfig_22_boat}
    \end{subfigure}
    \hfill
    \begin{subfigure}[b]{0.18\textwidth}
        \centering
        \includegraphics[width=\textwidth]{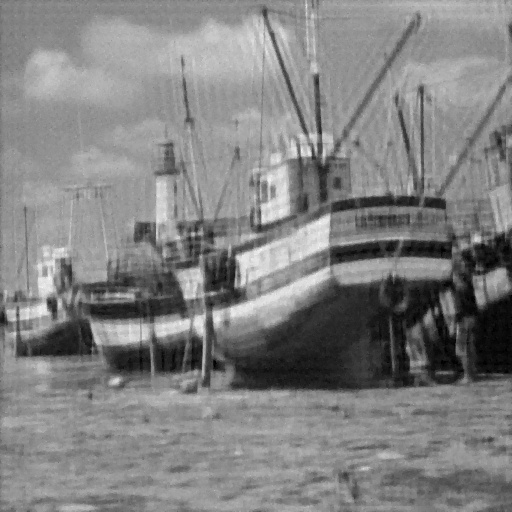}
        \caption{MYPGD}
        \label{fig:subfig_33_boat}
    \end{subfigure}
    \hfill
    \begin{subfigure}[b]{0.18\textwidth}
        \centering
        \includegraphics[width=\textwidth]{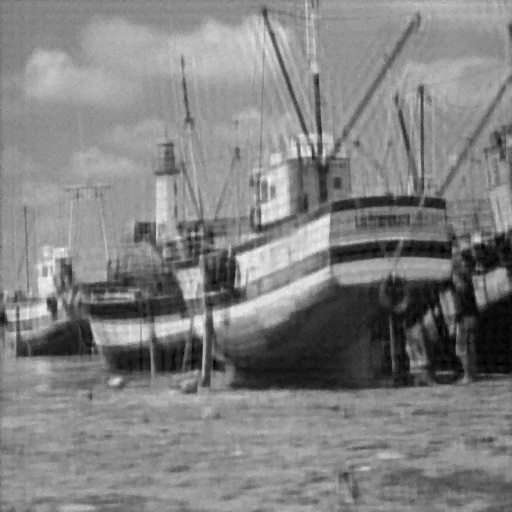}
        \caption{MYIPLA}
        \label{fig:subfig_44_boat}
    \end{subfigure}
    \hfill
    \begin{subfigure}[b]{0.18\textwidth}
        \centering
        \includegraphics[width=\textwidth]{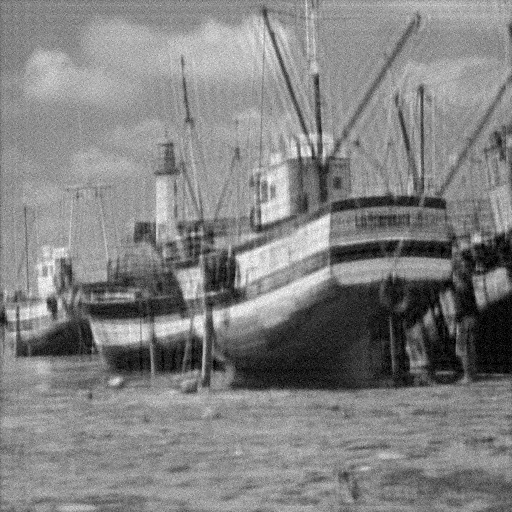}
        \caption{PIPGLA}
        \label{fig:subfig_55_boat}
    \end{subfigure}

    \caption{Image deblurring experiment. All the algorithms use $N=10$ particles and are run for 3000 iterations with a burn-in of 100 iterations.}
    \label{fig:image_deconvolution_methods_boat}
\end{figure}

\begin{figure}[h!]
    \centering

    \begin{subfigure}[b]{0.33\textwidth}
        \centering
        \includegraphics[width=\textwidth]{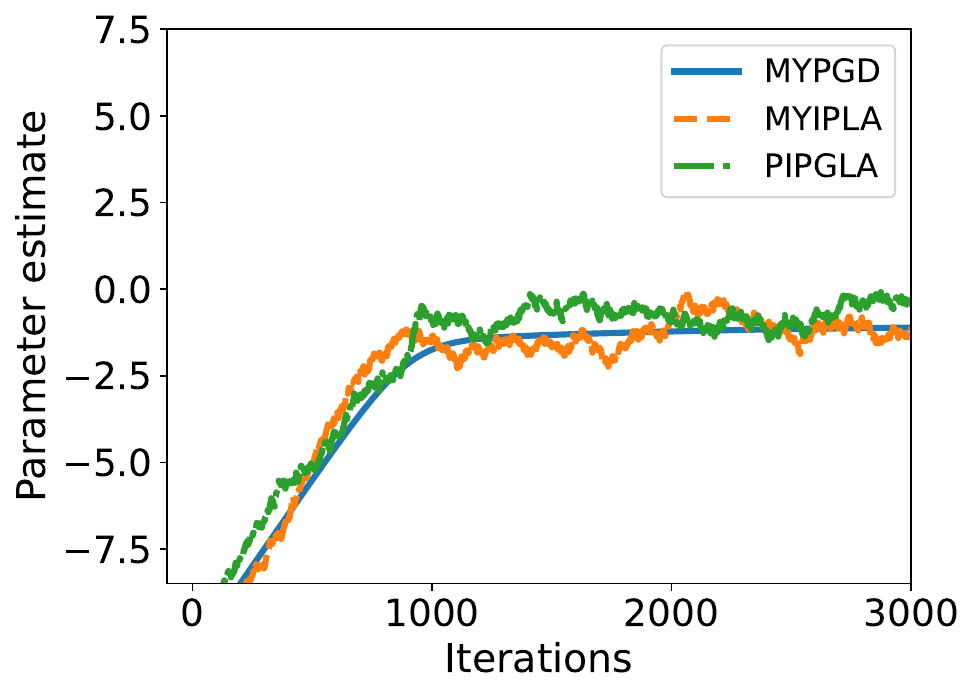}
        \caption{$\theta$ estimates}
        \label{fig:subfig_11_estimates_boat}
    \end{subfigure}
    \begin{subfigure}[b]{0.33\textwidth}
        \centering
        \includegraphics[width=\textwidth]{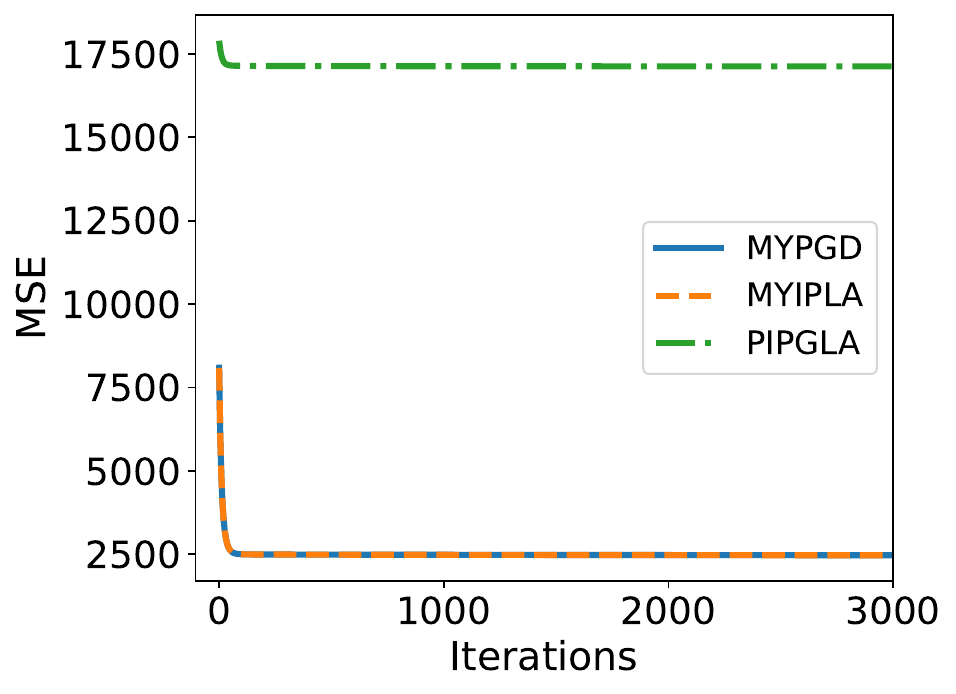}
        \caption{MSE}
        \label{fig:subfig_22_nmse_boat}
    \end{subfigure}
    \begin{subfigure}[b]{0.33\textwidth}
        \centering
        \includegraphics[width=\textwidth]{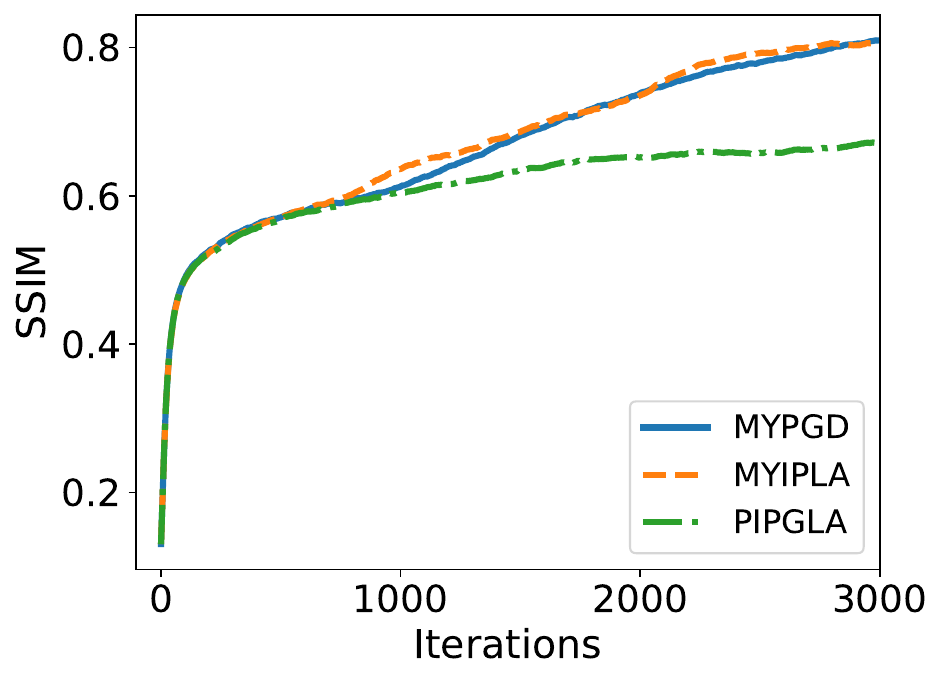}
        \caption{Structural similarity index (SSIM)}
        \label{fig:subfig_22_ssim_boat}
    \end{subfigure}

    \caption{Evolution of different quantities over iterations in the image deblurring experiment with $N=10$ particles for the boat image. The plots are shown after discarding a burn-in period of 100 iterations and the initial parameter is $\theta_0 = -8.1$.}
    \label{fig:image_deconvolution_methods_analysis_boat}
\end{figure}

\paragraph{Results.}
Figures \ref{fig:image_deconvolution_methods} and \ref{fig:image_deconvolution_methods_boat}  display the original and blurred images alongside the reconstructed images obtained using our different proximal algorithms. The methods are run for 3000 iterations (with a burn-in of 100 iterations) and $N=10$ particles, employing the Douglas-Rachford method to numerically evaluate the proximal operator of the total variation norm. Figures \ref{fig:image_deconvolution_methods_analysis} and \ref{fig:image_deconvolution_methods_analysis_boat} illustrate the evolution of the parameter estimates $\theta$, the mean squared error and the SSIM, after discarding a burn-in period of 100 iterations.
The high MSE for PIPGLA (Figures \ref{fig:subfig_22_nmse} and \ref{fig:subfig_22_nmse_boat}) arises from the difference in the shades of grey between the reconstructed and the original images, remaining large regardless of the choice of the proximal parameter $\lambda$.
Besides, the optimal value for the strength of the total variation prior achieved by the algorithms, $e^\theta\approx 0.35$ (for both test images) is close to the value set manually in similar works for image reconstruction (e.g. \citet{pereyra2016proximal, durmus_proximal,goldman2022gradient}).

\subsection{Nuclear-norm models for low rank matrix estimation}\label{app:low_rank_matrix}
In this section, we demonstrate another application of our methods: the problem of matrix completion. Matrix completion \citep{sparse_matrix_noise, sparse_matrix} focuses on recovering an intact matrix with low-rank property from incomplete data. Its application varies from
wireless communications \citep{mc_wireless}, traffic sensing \citep{Mardani2014EstimatingTA} to integrated radar and recommender systems \citep{GOGNA20155789}. 
The low-rank prior knowledge is incorporated in the model using the nuclear-norm of the matrix \citep{fazel_2002}. 
However, similar to the image deblurring example, the strength of this prior is a hyperparameter that must be set manually. Instead, we estimate the optimal value of this parameter, thereby extending the applicability of proximal methods that typically perform MLE, rather than MMLE, as in our algorithms.

We conduct a graphical posterior predictive check of the widely used nuclear norm model for low-rank matrices, similar to the example in \citet{pereyra2016proximal}, but in the context of matrix completion rather than matrix denoising. Let $x$ be an unknown low-rank matrix of size $n_1 \times n_2$. 
Consider a mask $M_\Omega$, where $\Omega$ is a set of indices from a matrix of size $n_1 \times n_2$. 
When the mask is applied to the matrix $x$, i.e., $M_\Omega X$, only the entries of the matrix corresponding to indices in  $\Omega$ are observed.
Furthermore, after the masking operation, we do not have direct access to the observed entries but instead observe a noisy version of them, where the observational noise has mean zero and covariance $\sigma^2 I$. Thus, our observations are given by $y = M_\Omega x + \sigma^2\varepsilon$, with $\varepsilon\sim\mathcal{N}(0, I)$. 
It is important to highlight, that we will also estimate with our algorithms the scale parameter $\sigma$, rather than requiring it to be fixed manually.

Our objective is to recover $x$ from $y$ under the prior knowledge that $x$ has low rank, that is, most of its singular values are zero. A convenient model for this type of problem is the nuclear
norm prior, which is a sparsity-inducing prior, given by
\begin{equation*}
    p_\theta(x) = C(\theta_1) e^{-e^{\theta_1}\Vert x\Vert_{\text{tr}}},
\end{equation*}
where $\Vert \cdot\Vert_{\text{tr}}$ is the trace (or nuclear) norm, which is a convex envelope of the rank function \citep{JMLR:v9:bach08a_trace_norm}, and is defined as
\begin{equation*}
  \Vert x\Vert_{\text{tr}} = \sum_{i=1}^r\sigma_i(x),
\end{equation*}
$r=\text{rank}(x)$ and $\sigma_1(x)\geq \dots\geq \sigma_r(x)\geq 0$ are the singular values.
Besides, the constant $C(\theta_1)$ can be computed using the pushforward argument, as in Section \ref{app:image_deblurring}, leveraging the linearity of the trace norm. Specifically, it is given by  $C(\theta_1) = C e^{d_x\theta_1}$, where $d_x$ denotes the dimension of the original matrix $x$ and $C$ is a constant.
The posterior distribution of our model can be written as
\begin{equation*}
    p_\theta(y|x)\propto \frac{e^{d_x\theta_1}}{e^{d_y\theta_2}}\exp\left({-\frac{\Vert M_\Omega x-y \Vert^2}{2e^{2\theta_2}}-e^{\theta_1}\Vert x\Vert_{\text{tr}}}\right).
\end{equation*}
Therefore, the negative log density can be decomposed as
\begin{equation*}
  U(\theta, x) =  \underbrace{-d_x\theta_1+d_y\theta_2 +\frac{\Vert M_\Omega x -y\Vert^ 2}{2 e^{2\theta_2}}}_{g_1(\theta, x)} +\underbrace{e^{\theta_1}\Vert x\Vert_{\text{tr}}}_{g_2(\theta, X)},
\end{equation*}
where $d_y$ denotes the number of observed entries.
Note that we exponentiate the parameters $\theta_1$ and $\theta_2$ to ensure their positivity.

\paragraph{Dataset.} We use the \textit{checkerboard} image of size $188\times 188$ and rank 2. We add Gaussian observational noise with variance $\sigma^2 = 0.1$ and mask $30\%$ of the pixels in the image.

\paragraph{Proximal operator of $g_2$.} Recall that $g_2(\theta, x)$ is of the form
\begin{equation*}
    g_2(\theta, x)= e^{\theta_1}\Vert x\Vert_{\text{tr}}.
\end{equation*}
To compute the proximal map, we first observe that if $\theta_1$ is known, then by \citet[Theorem 2.1]{svd_thresholding}, it follows that
\begin{equation*}
    \prox_{g_2}^{\lambda}(x) = \argmin_{z}\; \{e^{\theta_1} \Vert z\Vert_{\text{tr}} + \frac{1}{2\lambda} \Vert x-z\Vert_{F}^2 \}= S_{e^{\theta_1}\lambda}(x) := U\Sigma_{e^{\theta_1}\lambda} V^{T},
\end{equation*}
where $U\Sigma V^T$ is a singular value decomposition, and $\Sigma_{\beta}$ is diagonal with entries $(\Sigma_{\beta})_{ii} = \max \{\Sigma_{ii}-\beta, 0\}$. Based on this, we calculate
\begin{equation*}
    \prox_{g_2}^{\lambda}(\theta, x) = \argmin_{(\alpha, z)}\; \{e^{\alpha} \Vert z\Vert_{\text{tr}} + \frac{1}{2\lambda} \big(\Vert\theta_1-\alpha\Vert^2 + \Vert x-z\Vert_{F}^2 \big)\},
\end{equation*}
where $\Vert\cdot\Vert$ denotes the Frobenius norm.
The minimisers $(\alpha, z)$ satisfy the following system of equations
\begin{align}
        &\alpha = \theta_1 +\lambda e^{\alpha}\Vert S_{e^{\alpha}\lambda}(x)\Vert_{\text{tr}} \Longrightarrow (\alpha -\theta)e^{\theta_1 - \alpha} = \lambda e^{\theta_1}\Vert S_{e^{\alpha}\lambda}(x)\Vert_{\text{tr}},\label{eq:alpha_matrix}\\
    &z = S_{\lambda e^{\alpha}}(x).\label{eq:x_matrix}
\end{align}
Solving this system is complicated due to the dependence between $\alpha$ and $z$ and using an iterative solver can be computationally burdensome. Therefore, we have decided to approximate (\ref{eq:alpha_matrix}) by
\begin{equation*}
    (\alpha -\theta_1)e^{\theta_1-\alpha} \approx \lambda e^{\theta_1}\Vert S_{e^{\theta}\lambda}(x)\Vert_{\text{tr}} \Longrightarrow \alpha \approx \theta_1 + W(\lambda e^{\theta_1}\Vert S_{e^{\theta_1}\lambda}(X^l)\Vert_{\text{tr}}),
\end{equation*}
where $W$ is the Lambert $W$ function.
Substituting this value of $\alpha$ into (\ref{eq:x_matrix}), we obtain 
\begin{equation*}
    z \approx S_{e^{\alpha}\lambda}(x).
\end{equation*}

\paragraph{Implementation.} To stabilise the implementation of the algorithms, we divide the gradient and proximal mapping terms in the updates of $\theta_1$ and $\theta_2$ by the dimension of the the matrix $x$, $d_x$, and the number of observed entries in $y$, $d_y$, respectively. We then set $\gamma= 0.01$, and $\lambda=0.25$ for MYPGD and MYIPLA and $\lambda = 0.01$ for PIPGLA. 
The pixels of the initial particles are drawn from a normal distribution with mean $\mu = 50$ and scale parameter $10$, while the initial values of the parameters $\theta_1$ and $\theta_2$ are drawn from uniform distributions over $[-15, 5]$ and $[-10, 10]$, respectively.

\paragraph{Performance metrics.} 
To asses the performance of our algorithms for low-rank matrix completion, we analyse the normalised mean squared error (NMSE) for both the entire matrix and the missing entries.

\paragraph{Results.} Figure \ref{fig:matrix_completion_experiment} displays the original and observed matrices alongside the reconstructed matrices obtained using our different proximal algorithms. The methods are run for 3000 iterations (with a burn-in of 100 iterations) and $N=10$ particles. The NMSEs for the entire matrix and the missing entries for the final particle cloud are displayed in Table \ref{table-matrix-completion-experiment}, together with the computation times.

\begin{figure}[t]
    \centering

    \begin{subfigure}[b]{0.18\textwidth}
        \centering
        \includegraphics[width=\textwidth]{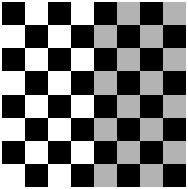}
        \caption{Original}
        \label{fig:subfig_11_matrix}
    \end{subfigure}
    \hfill
    \begin{subfigure}[b]{0.18\textwidth}
        \centering
        \includegraphics[width=\textwidth]{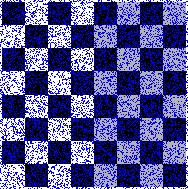}
        \caption{Observed}
        \label{fig:subfig_22_matrix}
    \end{subfigure}
    \hfill
    \begin{subfigure}[b]{0.18\textwidth}
        \centering
        \includegraphics[width=\textwidth]{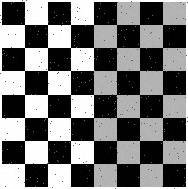}
        \caption{MYPGD}
        \label{fig:subfig_33_matrix}
    \end{subfigure}
    \hfill
    \begin{subfigure}[b]{0.18\textwidth}
        \centering
        \includegraphics[width=\textwidth]{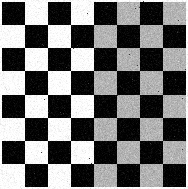}
        \caption{MYIPLA}
        \label{fig:subfig_44_matrix}
    \end{subfigure}
    \hfill
    \begin{subfigure}[b]{0.18\textwidth}
        \centering
        \includegraphics[width=\textwidth]{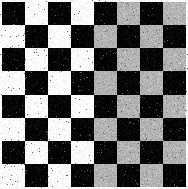}
        \caption{PIPGLA}
        \label{fig:subfig_55_matrix}
    \end{subfigure}

    \caption{Low-rank matrix completion. All the algorithms use $N=10$ particles and are run for 3000 iterations. The blue pixels in (b) represent the mask.}
    \label{fig:matrix_completion_experiment}
\end{figure}

\begin{table}[t]
  \caption{Low-rank matrix completion. Normalised mean squared errors (NMSE) for the entire matrix and the missing entries achieved using the final particle cloud with $N = 10$ and $3000$ iterations.}
  \label{table-matrix-completion-experiment}
  \centering
  \begin{tabular}{llll}
    \toprule
    Algorithm     & NMSE entire $\left(\%\right)$     &   NMSE missing $\left(\%\right)$ & Times (min)\\
    \midrule
    MYPGD & $1.21\pm 0.49$  & $1.67\pm0.52$  &   $\mathbf{4.1}$\\
    MYIPLA & $\mathbf{1.13\pm 0.48}$  & $\mathbf{1.53\pm 0.44}$ & $4.7$\\
        PIPGLA & $2.02\pm 0.29$  & $2.11\pm 0.31$  &   $5.5$\\
    \bottomrule
  \end{tabular}
\end{table}
\normalsize

\subsection{Ablation Study}
In this section, we analyse how the choice of the regularisation parameter $\lambda$ in the Moreau–Yosida approximation affects the performance and stability of the algorithm.

Choosing an appropriate value for $\lambda$ is a challenging task as this parameter controls both the level of regularisation and the closeness to the target, and is closely tied to the step size parameter $\gamma$.
\citet{durmus_proximal} provides some empirical guidance on the choice of $\gamma, \lambda$ for sampling tasks. 
\citet{crucinio2023optimal} shows that $\lambda \leq \gamma$ generally leads to better results in the case of grad Lipschitz potentials, while one should choose $\lambda \geq \gamma$ for light tail distributions. Adaptive strategies to choose $\lambda$ have been considered in the optimisation literature (see \citet{oikonomidis24a} and references therein) but equivalent results for sampling have not been obtained yet.

We conduct additional experiments to analyse the impact of the regularisation parameter $\lambda$ in the Bayesian logistic regression task with Laplace prior. In Figure \ref{fig:regularisation_parameter_analysis}, we report the performance (measured by NMSE) of MYIPLA, MYPGD and PIPGLA algorithms using approximate proximity maps, evaluated over a fine grid of $\lambda$ values.
The step size parameters used are those listed in Table~\ref{table-logistic-hyperparameters}: $\gamma = 0.05$ for MYIPLA and MYPGD, and $\gamma = 0.01$ for PIPGLA.
Each configuration is run with 100 different random seeds to compute confidence intervals.  
We observe that our algorithms exhibit stable performance across a broad range of $\lambda$ values.

\begin{figure}[h!]
    \centering

    \begin{subfigure}[b]{0.6\textwidth}
        \centering
        \includegraphics[width=\textwidth]{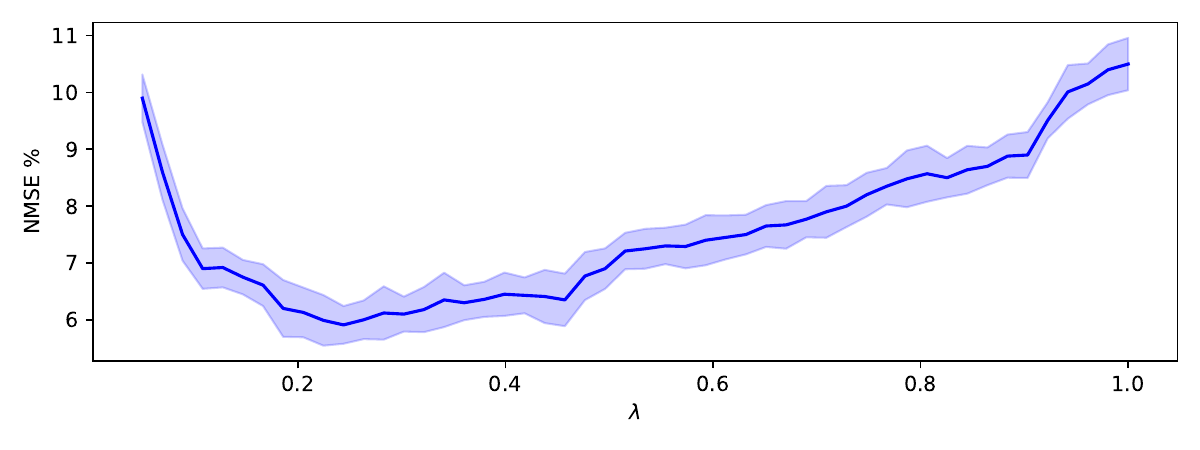}
        \caption{MYPGD, $\gamma = 0.05$}
        \label{fig:subfig_11_lambda_mypgd}
    \end{subfigure}
    \begin{subfigure}[b]{0.6\textwidth}
        \centering
        \includegraphics[width=\textwidth]{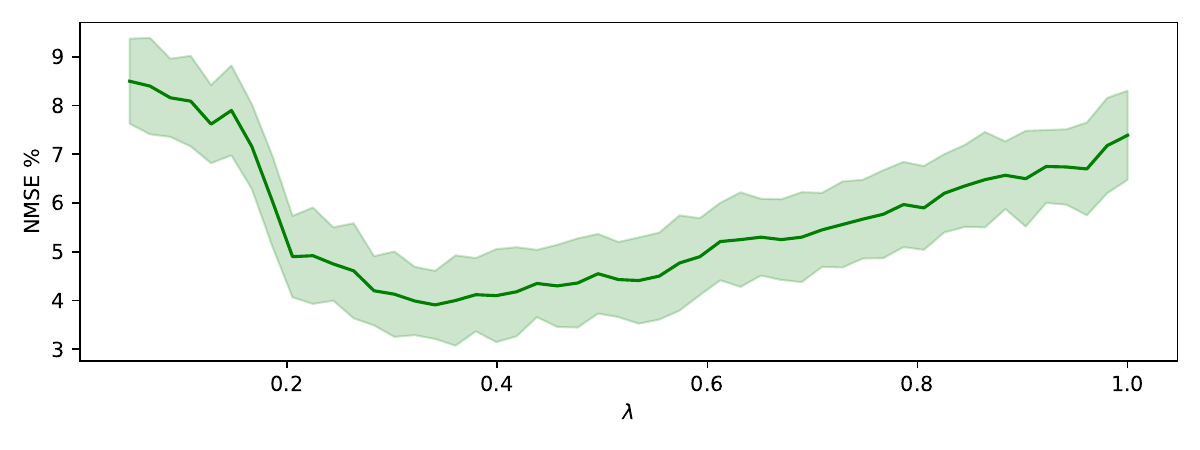}
        \caption{MYIPLA,  $\gamma = 0.05$}
        \label{fig:subfig_22_lambda_myipla}
    \end{subfigure}
    \begin{subfigure}[b]{0.6\textwidth}
        \centering
        \includegraphics[width=\textwidth]{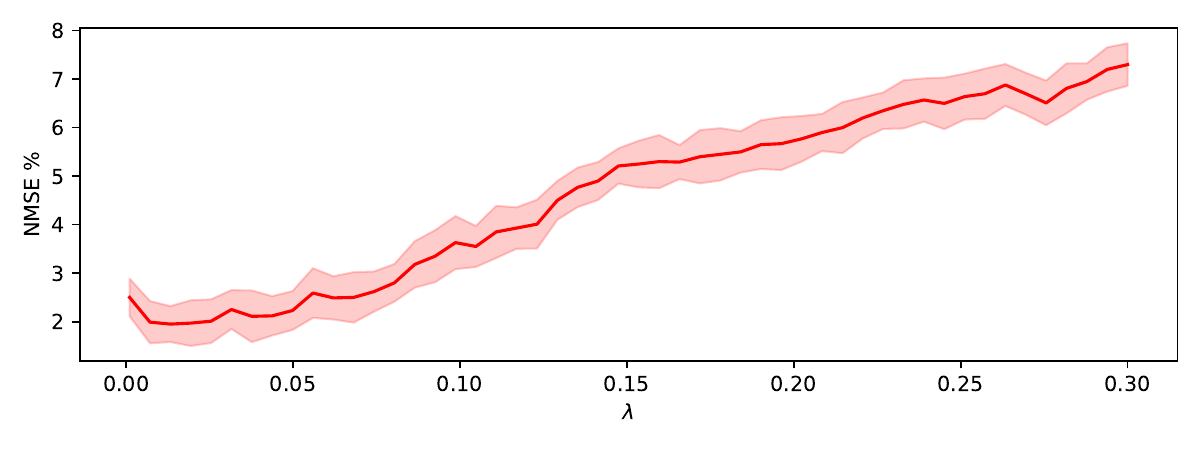}
        \caption{PIPGLA,  $\gamma = 0.01$}
        \label{fig:subfig_22_lambda_pipgla}
    \end{subfigure}

    \caption{Normalised MSE (\%) for different values of the regularisation parameter $\lambda$ and a fixed step size $\gamma$. Each configuration is run with 100 random seeds for 50 particles and 5000 steps. The proximal map for all algorithms is computed approximately.}
    \label{fig:regularisation_parameter_analysis}
\end{figure}

\end{document}